%% file: main.tex
\documentclass[a4paper,UKenglish,cleveref]{lipics-v2021}

\usepackage{enumitem} 
\usepackage{stackengine} 
\usepackage{xspace} 
\usepackage{stmaryrd}
\usepackage[all]{xy}    
\usepackage[vskip=3pt]{quoting}
\usepackage{turnstile}

\DeclareFontFamily{OT1}{pzc}{}
\DeclareFontShape{OT1}{pzc}{m}{it}{<->s*[1.30]pzcmi7t}{}
\DeclareMathAlphabet{\mathpzc}{OT1}{pzc}{m}{it}
\def\ct#1{\emath{\mathpzc{#1}}}

\title{An Effectful Object Calculus} 

\author{Francesco Dagnino}{DIBRIS, Universit\`a di Genova, Italy}{francesco.dagnino@dibris.unige.it}{https://orcid.org/0000-0003-3599-3535}{}

\author{Paola Giannini}{DiSSTE, Universit\`a del Piemonte Orientale, Italy}{paola.giannini@uniupo.it}{https://orcid.org/0000-0003-2239-9529}{}

\author{Elena Zucca}{DIBRIS, Universit\`a di Genova, Italy}{elena.zucca@unige.it}{https://orcid.org/0000-0002-6833-6470}{}

\authorrunning{F. Dagnino, P. Giannini, and E. Zucca} 

\Copyright{Francesco Dagnino, Paola Giannini, and Elena Zucca} 


\funding{This work was partially funded by the MUR project ``T-LADIES'' (PRIN 2020TL3X8X) and has the financial support of the Universit\`{a}  del Piemonte Orientale.}

\nolinenumbers 

\EventEditors{John Q. Open and Joan R. Access}
\EventNoEds{2}
\EventLongTitle{42nd Conference on Very Important Topics (CVIT 2016)}
\EventShortTitle{CVIT 2016}
\EventAcronym{CVIT}
\EventYear{2016}
\EventDate{December 24--27, 2016}
\EventLocation{Little Whinging, United Kingdom}
\EventLogo{}
\SeriesVolume{42}
\ArticleNo{23}

\input{macros}
\begin{document}

\maketitle

\begin{abstract}
We show how to smoothly incorporate in the object-oriented paradigm constructs to raise, compose, and handle effects in an arbitrary monad. The underlying pure calculus is meant to be a representative of the last generation of OO languages, and the effectful extension is manageable enough for ordinary programmers; notably, constructs to raise effects are just special methods. We equip the calculus with an expressive type-and-effect system, which, again by relying on standard features such as inheritance and generic types, allows a simple form of effect polymorphism. The soundness of the type-and-effect system is expressed and proved by a  recently  introduced technique, where the semantics is formalized by a one-step reduction relation from language expressions into monadic ones, so that it is enough to prove progress and subject reduction properties on this relation. 
\end{abstract}

\maketitle

\input{intro}

\input{language}

\input{monadic-sem}

\input{effect-system}

\input{soundness}

\input{related}

\input{conclu}

\bibliographystyle{plainurl}
\bibliography{biblio}

\clearpage 

\appendix 
\input{appendix}

\end{document}

%% file: macros.tex
\newcommand{\Space}{\hskip 0,7em}
\newcommand{\BigSpace}{\hskip 1,5em}
\newcommand{\HugeSpace}{\hskip 3,5em}
\newcommand{\refToRule}[1]{\textsc{\small (#1)}}
\newcommand{\refItem}[2]{\cref{#1}(\ref{#1:#2})} 
\newenvironment{proofOf}[1]{\begin{proof}[Proof of \cref{#1}]}{\end{proof}} 

\newcommand{\emath}[1]{\ensuremath{#1}\xspace}
\newcommand{\ple}[1]{\emath{{\langle #1 \rangle}}}
\newcommand{\Pair}[2]{\ple{#1,#2}}

\newcommand{\fun}[3]{\emath{#1 : #2 \rightarrow #3}} 
\newcommand{\funtype}[2]{\emath{#1\rightarrow #2}} 
\newcommand{\pfun}[3]{\emath{#1 : #2 \rightharpoonup   #3}}
 
\newcommand{\dom}{\mathsf{dom}} 
\newcommand{\N}{\mathbb{N}} 
\newcommand{\Subst}[3]   {#1[#2/#3]}

\newcommand{\NamedRule}[4]{\scriptstyle{\textsc{(#1)}}\
\displaystyle                  
\frac{#2}{#3}         
\begin{array}{l}
#4     
\end{array}
}

\newenvironment{grammatica}{\begin{array}{lcll}}{\end{array}}
\newcommand{\produzione}[3]{#1&::=&#2&\mbox{#3}}

\newcommand{\kw}[1]{\texttt{#1}}
\newcommand{\aux}[1]{\mathsf{#1}}
\newcommand{\mvar}[1]{\mathit{#1}}
\newcommand{\produzioneinline}[2]{#1::=#2}
\newcommand{\seq}[1]{\overline{#1}} 

\newcommand{\nt}[3]{\emath{#1 : #2 \Rightarrow #3}} 
\newcommand{\Set}{\ct{Set}} 
\newcommand{\Id}{\mathsf{Id}} 
\newcommand{\id}{\mathsf{id}} 

\newcommand{\mnd}{\Mnd\mfun}
\newcommand{\mfun}{M}
\newcommand{\mun}{\eta}
\newcommand{\mmul}{\mu}
\newcommand{\mkl}[1]{{#1}^\dagger}

\newcommand{\mbind}{\mathbin{\gg=}} 
\newcommand{\Mmap}[3][]{\aux{map}\ifblank{#1}{}{_{#1}}\, #2\, #3}
\newcommand{\Mnd}[1]{\mathbb{#1}}
\newcommand{\Mun}[1]{\mun^{\Mnd{#1}}}
\newcommand{\Mmul}[1]{\mmul^{\Mnd{#1}}}

\newcommand{\PowerFun}{P}
\newcommand{\ListFun}{L} 
\newcommand{\elist}{\epsilon}
\newcommand{\cons}{\colon} 

\newcommand{\DistFun}{D}
\newcommand{\Supp}{\aux{supp}} 

\newcommand{\Object}{\aux{Object}}

\newcommand{\tname}{\aux{N}}
\newcommand{\m}{\aux{m}}
\newcommand{\prog}{\mvar{P}}
\newcommand{\tdec}{\mvar{td}} 
\newcommand{\X}{\mvar{X}}
\newcommand{\Y}{\mvar{Y}}
\newcommand{\NT}{\mvar{N}}
\newcommand{\md}{\mvar{md}}
\newcommand{\e}{\mvar{e}}
\newcommand{\x}{\mvar{x}}
\newcommand{\y}{\mvar{y}}
\newcommand{\ext}{\triangleleft}
\newcommand{\TDec}[3]{#1 \ext #2 \{#3\}} 
\newcommand{\MethDec}[4]{#1:#2\,#3\, #4}
\newcommand{\Gen}[2]{#1[#2]} 
\newcommand{\Ext}[2]{#1\ext#2}
\newcommand{\MCall}[4]{#1.#2[#3](#4)} 
\newcommand{\MCallNG}[3]{#1.#2(#3)}   
\newcommand{\Obj}[2]{ #1\{#2\} }
\newcommand{\MBody}[3]{\Pair{#1\, #2}{#3}}
\newcommand{\GenMBody}[4]{[#1]\MBody{#2}{#3}{#4}}

\newcommand{\T}{\mvar{T}}
\newcommand{\UT}{\mvar{U}}
\newcommand{\sig}{\mvar{s}}
\newcommand{\MT}{\mvar{MT}}
\newcommand{\kind}{\aux{k}}
\newcommand{\abs}{\aux{abs}}
\newcommand{\defn}{\aux{def}}
\newcommand{\MSig}[3]{#2:\MkT{#1}{#3}}
\newcommand{\MkT}[2]     {#1\,#2}
\newcommand{\mkind}{\aux{mkind}}
\newcommand{\TObj}[2]{\Obj{#1}{#2}} 
\newcommand{\TMeth}[3]{ [#1] #2 \to #3 } 
\newcommand{\TMethNG}[2]{ #1 \to #2 } 

\newcommand{\Fun}[3]{\lambda #1{:}#2.#3}
\newcommand{\apply}{\aux{apply}}

\newcommand{\mgc}{\aux{mgc}} 
\newcommand{\Continue}{\aux{c}}
\newcommand{\Stop}{\aux{s}}
\newcommand{\ve}{\mvar{v}} 
\newcommand{\mode}{\mu}
\newcommand{\cc}{\mvar{c}} 
\newcommand{\handler}{\mvar{h}}
\newcommand{\Ret}[1]{\kw{return}\ #1}
\newcommand{\Do}[3]{\kw{do}\ #1 = #2\kw{;}\ #3} 
\newcommand{\Try}[4]{\kw{try}\ #1\  \kw{with}\  \Handler{#2}{#3}{#4}}
\newcommand{\TryShort}[2]{\kw{try}\ #1\ \kw{with}\ #2}
\newcommand{\TryNoFinal}[2]{\kw{try}\, #1\,  \kw{with}\,  #2}
\newcommand{\Handler}[3]{#1, \Pair{#2}{#3}} 
\newcommand{\CC}[7]{#1.#2: \CExp{#3}{#4}{#5}{#6}{#7}}
\newcommand{\CExp}[5]{\GenMBody{#1}{#2}{#3}{#4}_{#5}}

\newcommand{\Exp}{\aux{Exp}}
\newcommand{\Val}{\aux{Val}} 
\newcommand{\me}{\textsc{e}} 
\newcommand{\mve}{\textsc{v}} 
\newcommand{\Res}{\aux{Res}}
\newcommand{\res}{r}
\newcommand{\conf}{c} 
\newcommand{\red}{\to}
\newcommand{\Red}{\Rightarrow}
\newcommand{\Redstar}{\Red^\star} 
\stackMath

\newcommand\xxrightarrow[2][]{\mathrel{%
  \setbox2=\hbox{\stackon{\scriptstyle#1}{\scriptstyle#2}}%
  \stackunder[0pt]{%
    \xrightarrow{\makebox[\dimexpr\wd2\relax]{$\scriptstyle#2$}}%
  }{%
   \scriptstyle#1\,%
  }%
}}
\newcommand{\ehole}{[\ ]} 
\newcommand{\purered}{\red_p} 
\newcommand{\mbody}{\aux{lookup}} 
\newcommand{\cmatch}{\aux{cmatch}} 
\newcommand{\mrun}[2]{\emath{\aux{run}_{#1,#2}}} 
\newcommand{\instanceof}[2]{#1\ \aux{instof}\ #2}
\newcommand{\Conf}{\aux{Conf}}
\newcommand{\Wrng}{\aux{Wr}}
\newcommand{\wrng}{\aux{wrong}} 
\newcommand{\finsem}[1]{\infsem[\star]{#1}}  
\newcommand{\infsem}[2][\infty]{\llbracket #2 \rrbracket_{#1}} 
\newcommand{\tred}{\xxrightarrow[\tredfun]{}}
\newcommand{\tredfun}{\aux{step}}

\newcommand{\mconf}{\textsc{c}} 
\newcommand{\mres}{\textsc{r}} 

\newcommand{\mctr}{\aux{res}} 
\newcommand{\ctr}{\mctr_0}
\newcommand{\mbot}{\bot}
\newcommand{\msup}{\bigsqcup}

\newcommand{\Nat}{{\tt Nat}}
\newcommand{\nat}[1]{\Nat_{#1}}
\newcommand{\One}{{\tt One}}
\newcommand{\Zero}{{\tt Zero}}

\newcommand{\succtt}{{\tt succ}}
\newcommand{\nottt}{{\tt not}}
\newcommand{\pred}{{\tt pred}}
\newcommand{\zero}{{\tt zero}}
\newcommand{\match}{{\tt match}}
\newcommand{\Even}{{\tt Even}}
\newcommand{\If}{{\tt if}}

\newcommand{\ifNatTE}{{\tt if[Nat\ ThenElse[Nat]]}}
\newcommand{\True}{{\tt True}}
\newcommand{\False}{{\tt False}}

\newcommand{\ExSet}{\aux{Exc}}
\newcommand{\ExceptFun}[1][]{E\ifblank{#1}{}{_{#1}}}
\newcommand{\Exc}{\aux{E}}
\newcommand{\MyExc}{\aux{MyE}}
\newcommand{\Exception}{{\tt Exception}}
\newcommand{\getExc}[1]{\aux{exc}(#1)}
\newcommand{\My}{{\tt My}}
\newcommand{\MyException}{{\tt MyException}}
\newcommand{\exc}{\aux{e}}
\newcommand{\throw}{{\tt throw}}
\newcommand{\MyTE}{{\tt MyTE}}
\newcommand{\throwNat}{{\tt throw[Nat]}}
\newcommand{\mt}{{\tt m}}
\newcommand{\effsem}[1]{\llbracket #1 \rrbracket}

\newcommand{\Failure}{{\tt Failure}}
\newcommand{\fail}{{\tt fail}}
\newcommand{\Fail}{\aux{Fail}}
\newcommand{\Test}{{\tt Test}}

\newcommand{\aU}{\texttt{"1"}}
\newcommand{\aD}{\texttt{"a"}}
\newcommand{\sumAsNat}{{\tt sumAsNat}}
\newcommand{\toNat}{{\tt toNat}}
\newcommand{\sumT}{{\tt sum}}
\newcommand{\nU}{{\tt n1}}
\newcommand{\nD}{{\tt n2}}
\newcommand{\nT}{{\tt n}}

\newcommand{\CCSh}[6]{#1.#2: \MBody{#3}{#4}{#5}_{#6}}
\newcommand{\RetS}[1]{\kw{ret}\, #1}
\newcommand{\DoS}[3]{\kw{do}\, #1 {=} #2\kw;#3} 

\newcommand{\Chooser}{{\tt Chooser}}
\newcommand{\choott}{{\tt choose}}
\newcommand{\ytt}{{\tt y}}
\newcommand{\ztt}{{\tt z}}
\newcommand{\mU}{{\tt m1}}
\newcommand{\mD}{{\tt m2}}
\newcommand{\MyTEU}{{\tt MyTE1}}
\newcommand{\MyTED}[1]{{\tt MyTE}^#1}
\newcommand{\List}[1]{[\,#1\,]}
\newcommand{\Two}{{\tt Two}}

\newcommand{\AllLift}[1][]{\mathbf{\forall}\ifblank{#1}{}{^{#1}}}
\newcommand{\ExLift}[1][]{\mathbf{\exists}\ifblank{#1}{}{^{#1}}} 
\newcommand{\quantifier}{\aux{q}}

\newcommand{\eff}{\mvar{E}} 
\newcommand{\eZero}{\bullet}
\newcommand{\EComp}[2]{#1{\vee}#2}
\newcommand{\eTop}{\top}
\newcommand{\eCall}[3]{#1.#2[#3]}
\newcommand{\TEff}[2]{#1{!}#2}
\newcommand{\TMethEff}[4]{ [#1] #2 \to \TEff{#3}{#4} } 
\newcommand{\MkTE}[5]{#1\,\TMethEff{#2}{#3}{#4}{#5}}

\newcommand{\EffRed}[3]{#1\vdash#2{\effred}#3} 
\newcommand{\effred}{\Downarrow}

\newcommand{\TEnv}{\Phi}
\newcommand{\ok}{\diamond}
\newcommand{\IsWFType}[2]{ #1 \vdash #2\ok }
\newcommand{\subt}{\leq}
\newcommand{\SubType}[3]{#1\vdash#2\subt#3}
\newcommand{\SubT}[2]{#1\subt#2}
\newcommand{\SubTNarrow}[2]{#1{\subt}#2}

\newcommand{\typeof}[3]{#1\vdash#2\leadsto#3}
\newcommand{\override}[2]{#1\{#2\}}
\newcommand{\sigPlus}{\uplus}
\newcommand{\bigSigPlus}{\biguplus}

\newcommand{\NoAbs}[1]{\aux{NoAbs}(#1)}
\newcommand{\NoMgc}[1]{\aux{NoMgc}(#1)}
\newcommand{\TVars}[2]{\seq#2{:}\seq#1}
\newcommand{\IsMType}[4]{\mtype_{#1}(#2,#3){=}#4}
\newcommand{\mtype}{\aux{mtype}}
\newcommand{\GetDef}[1]{\aux{GetDef}(#1)}

\newcommand{\hfilter}{\mvar{H}}
\newcommand{\FilterF}[1]{\mathcal{F}_{#1}}
\newcommand{\FilterFun}[2]{\FilterF{#2}(#1)}
\newcommand{\cfilter}{\mvar{C}}
\newcommand{\CFilter}[6]{#1.#2:[#3]\ple{#5}}
\newcommand{\HFilter}[2]{#1, #2} 
\newcommand{\EffSet}{\aux{Eff}}

\newcommand{\IsWFVal}[4]{#1;#2\vdash#3:#4}
\newcommand{\IsWFExp}[5]{#1;#2\vdash#3:\TEff{#4}{#5}}
\newcommand{\IsWFExpNarrow}[5]{#1;#2{\vdash}#3{:}\TEff{#4}{#5}}
\newcommand{\IsWFHandler}[6]{#1;#2;#3\vdash#4:\TEff{#5}{#6}}
\newcommand{\IsWFHandlerNarrow}[6]{#1;#2;#3{\vdash}#4{:}\TEff{#5}{#6}}
\newcommand{\IsWFClause}[5]{#1;#2;#3\vdash#4\Rightarrow#5}
\newcommand{\IsWFClauseNarrow}[5]{#1;#2;#3{\vdash}#4{\Rightarrow}#5}
\newcommand{\IsWFNType}[1]{\vdash#1\, \ok}
\newcommand{\IsWFMethod}[4]{#1;#2;#3\vdash#4\,\ok}
\newcommand{\TVar}[2]{#2:#1}
\newcommand{\TVarNarrow}[2]{#2{:}#1}

\newcommand{\WTVal}[2]{\vdash#1:#2}
\newcommand{\WTValNarrow}[2]{{\vdash}#1{:}#2}
\newcommand{\WTExp}[3]{\vdash#1:\TEff{#2}{#3}}
\newcommand{\WTExpNarrow}[3]{{\vdash}#1{:}\TEff{#2}{#3}}
\newcommand{\WTRes}[2]{\vdash#1:#2}

\newcommand{\WTMExp}[4]{\sststile{#4}{}#1:\TEff{#2}{#3}}
\newcommand{\WTMRes}[3]{\sststile{#3}{}#1:#2}
\newcommand{\WTMResQ}[4]{\sststile{#3}{#4}#1:#2}
\newcommand{\WTMVal}[3]{\sststile{#3}{}#1:#2}

\newcommand{\MEff}{\mathcal{E}}

\newcommand{\mlift}[1][]{\lambda\ifblank{#1}{}{^{#1}}} 

\newcommand{\PW}{\mathcal{P}}

%% file: intro.tex

\section{Introduction}

Every modern programming language needs to support a wide range of computational effects, such as exceptions, input/output, interaction with a memory and classical or probabilistic non-determinism. 
 Several approaches have been proposed for  designing  language constructs dealing with different computational effects in a uniform and principled way. 
The monad pattern \cite{Moggi89,Moggi91,Wadler95} and algebraic effects and handlers \cite{PlotkinP01,PlotkinP02,PlotkinP03,PlotkinPretnar09,PlotkinPretnar13,BauerP15,Pretnar15} stand out for their impact on real world  programming, and are now incorporated  into  many mainstream languages such as Haskell, OCaml 5 and Scala. 
However,  this area  of research  focuses on the functional paradigm.   
The aim of this paper, instead, is to show that constructs to raise, compose, and handle effects can be smoothly incorporated  into  object-oriented languages;  in fact,  the specific features of the paradigm help in allowing a \mbox{simple effectful extension. }

We illustrate  the effectful extension  on top of a pure object calculus, meant to be an evolution of Featherweight Java (FJ)\footnote{In its version with generic types.} \cite{IgarashiPW99},  for  more than twenty years the paradigmatic calculus for Java-like languages. 
Indeed, like FJ it only includes distinctive ingredients, such as nominal types, inheritance, and objects, and no imperative features. 
However, following the current trend in OO languages, and inspired by recent work \cite{WangZOS16,ServettoZ21,WebsterSD24}, object values are, rather than class instances with fields,   \emph{object literals}, extending a sequence of nominal types with  additional methods.  
 In other words, we select, as subset of Java features, interfaces with default methods and anonymous classes, rather than classes with fields and constructors. 
This choice leads to generality, in the sense that the language design encompasses both Java-like and functional calculi.  More precisely, FJ can be easily encoded by representing fields as  constant  methods, 
 and, at the same time, the simply-typed lambda calculus can be seen as a language subset, since  functions  are a special case of stateless objects. 

 The effectful extension builds on generic effects \cite{PlotkinP03} and handlers \cite{PlotkinPretnar09}, carefully adapted to be  more familiar   to  OO programmers. 
Constructs for  raising effects are just method calls of a special kind, called \emph{magic}.\footnote{This terminology is taken from \cite{WebsterSD24}, and also informally used in Java reflection and Python.} 
They  are made available 
 through predefined interfaces,  and do not specify an implementation,  like  abstract ones. However, rather than being deferred to subtypes, their implementation is provided ``by the system''. 
 For instance, throwing an exception is a magic call \lstinline{Exception.throw()}. 
This fits naturally in the OO paradigm, where computations happen in the context of  a program. Again to be familiar, the construct to handle effects is a generalized try-block with catch clauses. A raised effect (a magic call) is caught by a clause when the receiver's type is (a subtype of) that specified in the clause, as for Java exceptions, generalized to arbitrary magic calls. 
To this end, besides those interrupting the normal flow of execution, we allow   catch clauses replacing an effect with an alternative behaviour in a continuous manner. Altogether, the construct provides a good compromise between simplicity and expressivity: notably, it does not explicitly handle continuations, as in handlers of algebraic effects \cite{PlotkinPretnar09,PlotkinPretnar13,KammarLO13,HillerstromL18}, yet expressing a variety of handling mechanisms. 

We endow our calculus with  a type-and-effect system where effect types\footnote{The term ``effect'' is used in literature both as synonym of computational effect, and in the context of type-and-effect systems, as a static approximation of the former.   We will use ``effect'' when there is no ambiguity, otherwise  ``computational effect'' and ``effect type'', respectively.} again generalize sets of exceptions, as in \lstinline{throws} clauses, to sets of \emph{call-effects}, of shape $\eCall{\T}{\m}{\seq\T}$, approximating the computational effects of a call of $\m$ in type $\T$, with types $\seq\T$ instantiating the type variables of the method. 
 This approach, inspired by \cite{GarianoNS19},  provides more expressivity than  the effect systems for algebraic effects \cite{BauerP14,BauerP15,Pretnar15}, which only  track the  \emph{names} of operations raising effects, without any type information. 
Moreover, $\T$ can be a type variable $\X$, expressing a  call-effect  not  yet  specified, which will become, when $\X$ is instantiated to a specific type, the effect type declared there for  the  method.  
This allows  a rather sophisticated approximation of effects, encompassing a simple form of \emph{effect polymorphism} \cite{Leijen17,HillerstromLAS17,LindleyMM17,BiernackiPPS19,BrachthauserSO20jfp,BrachthauserSO20oopsla}. 

Last but not least, 
the soundness of the type-and-effect system is expressed and proved by a newly introduced approach \cite{DagninoGZ25}, 
building on a recent line of research on \emph{monadic operational semantics} for effectful languages \cite{GavazzoF21,GavazzoTV24}. 
That is, the semantics is 
formalized by a one-step reduction from language expressions into monadic ones, parameterized by a monad modeling the capabilities provided by the system.
In this way, we can uniformly deal with a wide range of  effects, by just providing an interpretation of magic methods in the monad. 
Then, applying a technique from \cite{DagninoGZ25}, to ensure type-and-effect soundness it is enough to prove progress and subject reduction on the monadic one-step reduction. 
Accordingly with its shape, subject reduction roughly means that reducing an expression to a monadic one preserves its type-and-effect; hence, to express and prove this property, we need \emph{monadic typing judgments}, which can be derived by \emph{lifting} the non-monadic ones.  
In this respect, the contribution of this paper is twofold: we show a significant application of the technique, for a language with a complex type system; moreover, the abstract notion of lifting in \cite{DagninoGZ25} is concretized here  by deriving monadic typing rules, so that the proof of subject reduction can be driven by the usual reasoning. 

 We describe the language in \cref{sect:lang}, and its  semantics in \cref{sect:sem}.
 The type-and-effect system is illustrated in \cref{sect:effect-system}, and  its soundness is shown in \cref{sect:results}. \cref{sect:related} discusses related work and \cref{sect:conclu} summarizes  our  contribution and outlines future work.  
Auxiliary definitions and omitted proofs can be found in the Appendix.

%% file: language.tex

\section{Language}\label{sect:lang}
To have a smoother presentation, we first provide the syntax, discuss the key features, and show some examples, for the effect-free subset of the language; then we illustrate the constructs to raise, compose, and handle computational effects.

\smallskip
\noindent\textit{Effect-free language} 
As anticipated, in our calculus objects are not created by invoking constructors, but directly obtained by extending a sequence of nominal types with additional methods. In this way, the notions of class (more in general, nominal type) and object almost coincide. The only difference is that nominal types are top-level entities in a program, possibly generic (declaring type variables), and abstract (including non-implemented methods). On the other hand, objects can be seen as anonymous types declared on-the-fly in the code.  

Syntax and types are given in \cref{fig:syntax}.   We assume \emph{type names} $\tname$, \emph{method names} $\m$, \emph{type variables} $\X, \Y$, and \emph{variables} $\x,\y$.   The metavariable $\seq{\tdec}$ stands for  a sequence $\tdec_1\ldots\tdec_n$, and analogously for other overlined metavariables.  An empty sequence will be denoted by $\epsilon$. 

\begin{figure}[th]
\begin{math}
\begin{grammatica}
\produzione{\prog}{\seq{\tdec}}{program}\\
\produzione{\tdec}{\TDec{\Gen\tname{\Ext{\seq\Y}{\seq\T}}}{\seq\NT}{\seq{\md}}}{(nominal) type declaration} \\ 
\produzione{\md}{\MethDec{\m}{\abs}{\MT}{}\mid\MethDec{\m}{\defn}{\MT}{\MBody{\x}{\seq\x}{\e}}}{method declaration}\\
\produzione{\e}{\x \mid \Obj{\seq\NT}{\seq{\md}} \mid \MCall\e\m{\seq\T}{\seq\e}}{expression: variable, object, method call} \\
\produzione{\T,\UT}{\X \mid \TObj{\seq\NT}{\sig} }{type}  \\ 
\produzione{\NT}{\Gen\tname{\seq\T}}{nominal type}\\
\produzione{\sig}{\seq{\m}:\seq{\kind}\ \seq\MT}{signature (structural type)}  \\
\produzione{\kind}{\abs \mid \defn}{(method) kind}\\
\produzione{\MT}{\TMeth{\Ext{\seq\X}{\seq\UT}}{\seq\T}{\T}}{method type} \\
\end{grammatica}
\end{math}
\caption{Effect-free syntax and types}\label{fig:syntax} 
\end{figure}

 A program $\prog$ is a sequence of declarations of \emph{(nominal) types}\footnote{We adopt this terminology to avoid more connoted terms such as class, interface, trait.}. This sequence is  assumed to be a map, that is, type names are distinct. Hence, we can safely use the notations $\prog(\tname)$ and $\dom(\prog)$, as we will do for other (sequences representing) maps. 
 
A type declaration $\TDec{\Gen\tname{\Ext{\seq\Y}{\seq\T}}}{\seq\NT}{\seq{\md}}$ introduces a generic nominal type $\Gen\tname{\Ext{\seq\Y}{\seq\T}}$, inheriting from all those in $\seq\NT$. We assume that the notation $\Ext{\seq\Y}{\seq\T}$ represents a map from type variables to  types (their bounds),  that is, type variables are distinct and the two sequences have the same length; moreover, $\seq{\NT}$ represents a set, that is, order and repetitions are immaterial, and $\seq{\md}$ represents a map, that is, method names are distinct. Analogous assumptions hold in an object $\Obj{\seq\NT}{\seq\md}$ and a method type $\TMeth{\Ext{\seq\X}{\seq\UT}}{\seq\T}{\T}$. 

 Method declarations can be either abstract ($\abs$), that is, only specifying a method type, or defined ($\defn$), providing the (variables to be used as) parameters and the method body.  The programmer can choose an arbitrary variable $\x$ for the first  parameter in defined methods  to denote  the current object, rather than a fixed name as, e.g., $\aux{this}$ in Java.

 In the calculus, we adopt a uniform syntax for abstract and defined methods, to simplify the technical treatment;  a more realistic language would likely use a different concrete syntax, for instance declaring parameters together with their types. 
 
\begin{remark}\label{rem:lambda-encoding} 
Lambda-expressions can be encoded in the calculus, similarly to what is done in Java; that is, the object  $\Obj{}{\MethDec{\apply}{\defn}{\TMethNG{\T}{\UT}}{\MBody{\_}{\x}{\e}}}$ can be abbreviated $\Fun{\x}{\T}{\e}$.  
 Note that, in this way, encoded functions are possibly recursive and higher-order. 
\end{remark}
 
Types are either type variables $\X$ or \emph{object types}, of   shape $\TObj{\seq\NT}{\sig}$. Nominal types are (instantiations of) those declared in the program. 
 Object types are a form of intersection types, present also in Scala 3, denoting, without introducing a name, a type which is a subtype of all types in  $\seq\NT$, additionally providing methods in the signature $\sig$. 
Nominal types can be seen as special object types. We assume that a signature $\sig$ represents a map from method names to method kinds and types, that is, method names are distinct and the three sequences have the same length. Hence, we can safely use the notations $\sig(\m)$ and $\dom(\sig)$.

Note that type annotations written by the programmer are allowed to be arbitrary types, rather than only nominal.   In this respect, our calculus offers more generality than Java. 

 In the following, we will write 
$\seq\NT$ for $\TObj{\seq\NT}{}$;  note that such  a  syntactic form can be seen both as an object and the corresponding object type. Moreover, we will write $\Object$ for $\TObj{}{}$, which, as a type, is 
the top of the subtyping relation. 

\smallskip
\noindent\textit{Examples} 
In the code examples, we will use a number of other obvious conventions and abbreviations, such as omitting square brackets around an empty sequence in type/method declarations and method calls, and writing the pair of sequences $\Ext{\seq\X}{\seq\T}$ as a sequence of pairs, omitting bounds which are $\Object$. 
Finally, in both code and formalism, we will often use the wildcard $\_$ to  indicate  that some variable or meta-variable does not matter.
\begin{example}\label{ex:boolnat}
As an example of a program, we show the classic encoding of booleans and natural numbers in the object-oriented paradigm, through an inheritance hierarchy, 
 using the visitor pattern for defining functions over them.  
 \begin{lstlisting}
Bool {
  if: abs [X Y$\,\ext\,$ThenElse[X]] Y -> X 
  not: def -> Bool <b, b.if[Bool ThenElse[Bool]](
    ThenElse[Bool]{
      then: def -> Bool <$\_$, False>
      else: def -> Bool <$\_$, True>
    }
  )
}

True $\ext$ Bool { if: def [X Y$\,\ext\,$ThenElse[X]] Y -> X <$\_$ te, te.then()> }
False $\ext$ Bool { if: def [X Y$\,\ext\,$ThenElse[X]] Y -> X <$\_$ te, te.else()> }

ThenElse[X] { 
  then: abs -> X   
  else: abs -> X 
}
\end{lstlisting}
Let us focus first on the method \lstinline{if}, which takes a parameter of (a subtype of) \lstinline{ThenElse[X]}, expected to provide two alternative results of type \lstinline{X}. The method, abstract in \lstinline{Bool}, is implemented in \lstinline{True} and \lstinline{False} by selecting the \lstinline{then} and \lstinline{else} alternative provided by the argument, respectively. As said above, we use the wildcard as variable for the current object, since it does not occur in the body. 

The type \lstinline{ThenElse[X]} declares methods \lstinline{then} and \lstinline{else} to be implemented in  its  subtypes. To  illustrate  the language features, we define \lstinline{not} in \lstinline{Bool} by  invoking \lstinline{if} on the current object with, as argument, an object  which implements \lstinline{ThenElse[Bool]}.   Clearly  \lstinline{not} could  declared  abstract in \lstinline{Bool} and defined in the obvious way \mbox{in the two subtypes. }

Turning now to naturals,  a similar encoding through an inheritance hierarchy could be: 
\begin{lstlisting}
Nat { 
  succ: def -> Nat <n, Succ{pred: def -> Nat <$\_$, n>}
  match: abs [X] NatMatch[X] -> X
}

Zero $\ext$ Nat { match: def [X] NatMatch[X] -> X <$\_$ nm, nm.zero()> }
Succ $\ext$ Nat {
  pred: abs -> Nat
  match: def [X] NatMatch[X] -> X <n nm, nm.succ(n.pred())>
}

NatMatch[X] { zero: abs -> X   succ: abs Nat -> X }
\end{lstlisting}
Since in the calculus data are uniformly implemented by  stateless objects, numbers greater than $0$ are encoded as objects which extend \lstinline{Succ} by implementing the  predecessor method; that is, if $\nat{n}$ is the object encoding number $n$, then $n+1$ is encoded by
\lstinline{$\nat{n}$.succ()}, which evaluates to \lstinline@Succ{pred: def -> Nat <$\_$, $\nat{n}$>}@.

The type \lstinline{NatMatch[X]} declares methods \lstinline{zero} and \lstinline{succ} to be implemented in subtypes, offering a programming schema for definitions given by arithmetic induction.
For instance:
\begin{lstlisting}
Even $\ext$ NatMatch[Bool] {
  zero: def -> Bool <$\_$,True>
  succ: def Nat -> Bool <even n, n.match(even).not()>
}  
\end{lstlisting}
A different encoding of natural numbers could be provided with no inheritance hierarchy:
\begin{lstlisting}
Nat {
  match: def [X] NatMatch[X] -> X <$\_$ nm, nm.zero()>
  succ: def -> Nat 
    <n, Nat{match: def [X] NatMatch[X] -> X <$\_$ nm, nm.succ(n)>}}
}
\end{lstlisting}
In this version, $0$ is encoded by \lstinline@Nat@, and, if $\nat{n}$ is the object encoding number $n$, then $n+1$ is encoded by
\lstinline{Nat.succ($\nat{n}$)}, which evaluates to 
\begin{lstlisting}
Nat{ match: def [X] NatMatch[X] -> X <$\_$ nm, nm.succ($\nat{n}$)>} }
\end{lstlisting}
\end{example}

\smallskip
\noindent\textit{Effectful language} 
In \cref{fig:eff-syntax} we extend the language by adding constructs to raise, compose, and handle effects. 
\begin{figure}[th]
\begin{math}
\begin{grammatica}
\produzione{\prog}{\seq{\tdec}}{program}\\
\produzione{\tdec}{\TDec{\Gen\tname{\Ext{\seq\Y}{\seq\T}}}{\seq\NT}{\seq{\md}}}{type declaration} \\ 
\produzione{\md}{\MethDec{\m}{\abs}{\MT}{}\mid\MethDec{\m}{\defn}{\MT}{\MBody{\x}{\seq\x}{\e}}\mid\MethDec{\m}{\mgc}{\MT}{}}{method declaration}\\
\produzione{\ve}{\x \mid \Obj{\seq\NT}{\seq\md}}{value} \\ 
\produzione{\e}{ \MCall\ve\m{\seq\T}{\seq\ve}\mid \Ret\ve \mid \Do\x{\e_1}{\e_2}\mid\TryShort{\e}{\handler}}{expression} \\ 
\produzione{\handler}{\Handler{\seq\cc}{\x}{\e}}{handler}\\
\produzione{\cc}{\CC{\NT}{\m}{\seq\X}{\x}{\seq\x}{\e}{\mode}}{(catch) clause}\\
\produzione{\mode}{\Continue\mid\Stop}{mode}\\
\produzione{\kind}{\abs \mid \defn \mid \mgc}{(method) kind}\\
\end{grammatica}
\end{math}
\caption{Fine-grain syntax for the effectful language}\label{fig:eff-syntax} 
\end{figure}
As customary, we adopt a fine-grain syntax \cite{LevyPT03}, where \emph{values} are effect-free, whereas expressions may raise effects, and are also called \emph{computations}.
Methods can be, besides abstract and defined, \emph{magic} (abbreviated $\mgc$).  
 Magic methods  are made available to the programmer in the type declarations composing the program, and do not specify an implementation,  like  abstract methods. However, rather than being deferred to subtypes, their implementation is provided ``by the system''.  
We add the standard operators \lstinline{return} for embedding values into computations,  and \lstinline{do} for composing computations sequentially passing the result of the former to the latter.
Finally, we provide a try-block enclosing a computation with a \emph{handler}, consisting of a sequence of \emph{(catch) clauses}, and a  \emph{final expression}, parametric on a variable.
A clause specifies a type and method name, expected to identify a magic method declaration; the clause may catch calls of such method, by executing the \emph{clause expression} instead, parametric on type variables and variables, analogously to a method body.  After that, the final expression  is either executed or not depending on the \emph{mode}, 
either $\Continue$ or $\Stop$, for ``continue'' and ``stop'', respectively, of the clause. As illustrated in the following examples, a  $\Continue$-clause replaces an effect with an alternative behaviour in a continuous manner, whereas in  $\Stop$-clauses  handling the computational effect interrupts the \mbox{normal flow of execution.  }

%% file: monadic-sem.tex

\section{Monadic Operational Semantics}\label{sect:sem}

In this section, following the approach in \cite{DagninoGZ25}, we define a  \emph{monadic operational semantics}  for the language, parametric on an underlying monad.   

\smallskip
\noindent\textit{Monads} 
First of all we recall basic notions about monads,
referring  to standard textbooks \cite{Riehl17} for a detailed presentation. 

Monads \cite{EilenbergM65,Street72} are a fundamental notion in category theory, enabling an abstract and unified study of algebraic structures. 
Since Moggi's seminal papers \cite{Moggi89,Moggi91}, they have also  become a major tool in computer science, especially for describing the semantics of computational effects, as well as for integrating them in programming languages in a structured and principled way. 

A \emph{monad} $\mnd = \ple{\mfun,\mun,\mmul}$ (on \Set) consists of 
a functor \fun{\mfun}{\Set}{\Set} and two natural transformations \nt{\mun}{\Id}{\mfun} and \nt{\mmul}{\mfun^2}{\mfun} 
such that, for every set $X$, the following diagrams commute: 
\[\vcenter{\xymatrix{
  \mfun X \ar[r]^-{\mun_{\mfun X}} \ar[rd]_-{\id_{\mfun X}} 
& \mfun^2 X \ar[d]_-{\mmul_X} 
& \mfun X \ar[l]_-{\mfun\mun_X} \ar[ld]^-{\id_{\mfun X}} 
\\ 
& \mfun X 
}} \qquad 
\vcenter{\xymatrix{
  \mfun^3 X \ar[r]^-{\mfun\mmul_X} \ar[d]_-{\mmul_{\mfun X}} 
& \mfun^2 X \ar[d]^-{\mmul_X} 
\\
  \mfun^2 X \ar[r]^-{\mmul_X}
& \mfun X 
}}\]
The functor $\mfun$ specifies, for every set $X$, a set $\mfun X$ of monadic elements built over $X$, in a way compatible with functions. 
The map $\mun_X$, named \emph{unit}, embeds elements of $X$ into monadic elements in $\mfun X$, and 
the map $\mmul_X$, named \emph{multiplication}, flattens monadic elements built on top of other monadic elements into plain monadic elements.

Functions of type $X \to \mfun Y$ are called \emph{Kleisli functions} and play a special role: 
they can be seen as ``effectful functions'' from $X$ to $Y$, raising effects described by the monad $\mnd$. Given a Kleisli function, by \emph{Kleisli extension} we can define:
\begin{quoting}
\begin{math}
\begin{array}{ll}
\fun{\mkl{f}}{\mfun X}{\mfun Y}&
\mkl{f} = \mmul_Y \circ \mfun f
\end{array}
\end{math}
\end{quoting}
that is, 
first we lift $f$ through $\mfun$ to apply it to monadic elements and then we flatten the result using $\mmul_Y$. 
Moreover, given $\alpha\in \mfun X$, \fun{f}{X}{\mfun Y} and \fun{g}{X}{Y}, we set 
\begin{quoting}
\begin{math}
\begin{array}{ll}
\fun{-\mbind-}{\mfun X}{(X\to\mfun Y)\to\mfun Y}&\alpha \mbind f = \mkl{f}(\alpha) \\
\fun{\aux{map}}{(X\to Y)}{\mfun X \to \mfun Y}&\Mmap{g}{\alpha} = \mfun g(\alpha) 
\end{array}
\end{math}
\end{quoting}
The operator $\mbind$ is also called $\aux{bind}$. As its definition shows, it can be seen as an alternative description of the Kleisli extension, where the parameters are taken in inverse order.  
This view corresponds, intuitively, to the sequential composition of two expressions with effects, where the latter  depends on a parameter \emph{bound} to the result of the former. 
The operator $\aux{map}$ describes the effect of the functor $\mfun$ on functions. That is, the lifting of function $g$ through $\mfun$ is applied to a monadic value $\alpha$. 

\smallskip
\noindent\textit{Monadic reduction} 
Let $\Val$ and $\Exp$ be the sets of closed values and expressions, respectively. 
 In the following, we define a monadic (one-step) reduction  for the language, being a relation $\red$  on  $\Exp \times \mfun\Exp$, parametric on a monad  $\mnd = \ple{\mfun,\mun,\mmul}$.
More in detail, its definition depends on the following ingredients:
\begin{itemize}
\item The function $\fun{\mun_\Exp}{\Exp}{\mfun\Exp}$  embedding  language expressions  into their counterpart in the monad, written simply $\mun$ in the following.
\item The function $\fun{\aux{map}}{(\funtype{\Exp}{\Exp})}{\funtype{\mfun\Exp}{\mfun\Exp}}$  lifting functions from expressions to expressions  to their counterpart in the monad.
\end{itemize}
Moreover we assume:
\begin{itemize}
\item For every magic method $\m$ declared in  $\tname$, a partial function $\pfun{\mrun\tname\m}{\Val\times\Val^\star}{\mfun\Val}$,  returning a monadic value expressing the effects raised by a call. The function could be undefined, for instance, when arguments do not have the expected types.
\end{itemize}
The monadic reduction is modularly defined on top of a ``pure''  reduction $\purered$ on $\Exp\times\Exp$, which transforms calls of non-magic methods into  the corresponding bodies, as usual, and try-blocks into either magic calls or \lstinline{do} expressions, by distributing and possibly applying catch clauses; magic calls and \lstinline{do} expressions are, then, normal forms for the pure reduction. 

Rules defining the pure reduction are given in  \cref{fig:pure-red}.  The abbreviation $\handler$ for $\Handler{\seq\cc}{\x}{\e'}$ is intended to be used in all the rules below. 
\begin{figure}[t]
\begin{small}
\begin{math}
\begin{array}{l}
\NamedRule{invk}{ }
{ \MCall{\ve}{\m}{\seq\T}{\seq\ve} \purered \Subst{\Subst{\Subst{\e}{\seq\T}{\seq\X}}{\ve}{\x}}{\seq\ve}{\seq\x} }
{ \mbody(\ve,\m) = \ple{\seq\X,\x,\seq\x,\e} }
\\[4ex]

\handler=\Handler{\seq\cc}{\x}{\e'}

\\[2ex]

\NamedRule{try-ret}{ }
{ \TryShort{\Ret\ve}{\handler}\purered \Do\x{\Ret\ve}{\e'} }
{ } 

\\[3ex] 
\NamedRule{try-do}{ }
{ \TryShort{\Do\y{\e_1}{\e_2}}{\handler}\purered \Try{\e_1}{\seq\cc}{\y}{\TryShort{\e_2}{\handler}}}
{ }
\\[3ex]

\NamedRule{catch-continue}{ }
{ \begin{array}{l}
\TryShort{\MCall{\ve}{\m}{\seq\T}{\seq\ve}}{\handler} \purered\\
\HugeSpace\Do{\x}{\Subst{\Subst{\Subst{\e}{\seq\T}{\seq\X}}{\ve}{\x}}{\seq\ve}{\seq\x}}{\e'}
\end{array}
}
{ 
   \mbody(\ve,\m)=\Pair{\mgc}{\_}\\
  \cmatch( \MCall\ve\m{\seq\T}{\seq\ve},\seq\cc) = \CExp{\seq\X}{\x}{\seq\x}{\e}{\Continue} \\
}

\\[4ex]

\NamedRule{catch-stop}{ }
{ \TryShort{\MCall{\ve}{\m}{\seq\T}{\seq\ve}}{\handler} \purered\Subst{\Subst{\Subst{\e}{\seq\T}{\seq\X}}{\ve}{\x}}{\seq\ve}{\seq\x}
}
{ 
   \mbody(\ve,\m)=\Pair{\mgc}{\_}\\
  \cmatch( \MCall\ve\m{\seq\T}{\seq\ve},\seq\cc) = \CExp{\seq\X}{\x}{\seq\x}{\e}{\Stop} \\
}

\\[4ex]
\NamedRule{fwd}{ }
{ \TryShort{\MCall{\ve}{\m}{\seq\T}{\seq\ve}}{\handler} \purered \Do{\x}{\MCall{\ve}{\m}{\seq\T}{\seq\ve}}{\e'} }
{ 
   \mbody(\ve,\m)=\Pair{\mgc}{\_}\\
  \cmatch( \MCall\ve\m{\seq\T}{\seq\ve},\seq\cc)\  \mbox{undefined} 
}
\\[5ex]
\NamedRule{try-ctx}{
  \e_1 \purered \e_2
}{ \TryShort{\e_1}{\handler} \purered \TryShort{\e_2}{\handler} }
{ }

\end{array}
\end{math}
\end{small}
\caption{Pure reduction}\label{fig:pure-red} 
\end{figure}
The $\mbody$ function models method look-up, searching for the method declaration corresponding to a call, with two different successful outcomes: either a defined method, and the result are its type variables, variables, and body, or a magic method, necessarily in a type declaration, and the result is an  $\mgc$ tag, and the type name. 
The $\cmatch$ function is defined, on a pair consisting of a magic call and a sequence of clauses, when there is a (first) clause catching the call; in this case, the corresponding clause expression is returned. The formal definitions are given in \cref{fig:auxfun}.

Rule \refToRule{invk} is the standard rule for method invocation.  That is, method look-up finds a defined method, and the call reduces to the method's body where type variables and variables are replaced by actual types and arguments. 
The other rules model the behaviour of a computation enclosed in a try-block.  

When the enclosed computation has no effects, rule \refToRule{try-ret}, the try-block is reduced to the sequential composition of  this  computation with the final expression.
In case of a  \lstinline{do} composition of two computations,  
the \lstinline{do} is eliminated by reducing it to a try-block enclosing the first computation, with as final expression another try-block enclosing the second one; clauses are propagated to both. 

When the computation calls a magic method, the behaviour depends on whether a matching clause is found. If it is found, then the clause expression is executed, after 
replacing parameters by arguments, see rules \refToRule{catch-continue} and \refToRule{catch-stop}. In a $\Continue$-clause, the final 
expression is then executed. If no matching clause is found, then in rule \refToRule{fwd} the try-block is reduced to the sequential composition of the magic call with \mbox{the final expression. }

Finally, rule \refToRule{try-ctx} is the standard contextual rule. Note that, as already mentioned, in the pure reduction there are no rules for magic calls and \lstinline{do} expressions. Indeed, the pure reduction only handles expressions when they are enclosed in a try-block, so that effects which would be raised by magic calls \mbox{can be possibly caught. }

The $\mbody$ function is defined in \cref{fig:auxfun}. 
\begin{figure}[t]
\begin{small}
\begin{align*} 
\mbody(\TObj{\seq\NT}{\seq\md},\m) &=
  \begin{cases}
  \GenMBody{\seq\X}{\x}{\seq\x}{\e} 
    & \mbox{if}\ \seq{\md}(\m)=\MethDec{\m}{\defn}{\TMeth{\Ext{\seq\X}{\_}}{\_}{\_}}{\MBody{\x}{\seq\x}{\e}}\\ 
  \mbody(\seq\NT,\m) 
    &\text{otherwise} 
  \end{cases}
\\ 
\mbody(\NT_1\ldots\NT_n ,\m) &= 
  \begin{cases}
  \mbody(\NT_i,\m) & \mbox{if}\
 \mbody(\NT_i,\m)\ \mbox{defined for a unique}\ i\in 1..n \\ 
  \mbox{undefined} &\text{otherwise} 
  \end{cases}
 \\ 
\mbody(\Gen\tname{\seq\T},\m) &= 
  \begin{cases}
  \GenMBody{\seq\X}{\x}{\seq\x}{\Subst\e{\seq\T}{\seq\Y}} 
    & \mbox{if}\ \seq{\md}(\m)=\MethDec{\m}{\defn}{\TMeth{\Ext{\seq\X}{\_}}{\_}{\_}}{\MBody{\x}{\seq\x}{\e}}\\ 
      \Pair{\mgc}{\tname}
    & \mbox{if}\ \seq{\md}(\m)=\MethDec{\m}{\mgc}{\TMeth{\Ext{\_}{\_}}{\_}{\_}}{}\\ 
  \mbody(\Subst{\seq\NT}{\seq\T}{\seq\Y},\m) 
    &\text{otherwise} 
  \end{cases} \\ 
  & \text{where }\prog(\tname)=\TDec{\Gen\tname{\Ext{\seq\Y}{\_}}}{\seq\NT}{\seq{\md}} 
\end{align*}

\begin{quoting}
$\cmatch(\MCall{\ve}{\m}{\seq\T}{\seq\ve}, \cc)=
\begin{cases}
\CExp{\seq\X}{\x}{\seq\x}{\e}{\mode}&\mbox{if}\ \cc=\CC{\NT}{\m}{\seq\X}{\x}{\seq\x}{\e}{\mode}\ \mbox{and}\ \instanceof{\ve}{\NT}\\
\mbox{undefined}&\mbox{otherwise}
\end{cases}
$
\\
$\cmatch(\MCall{\ve}{\m}{\seq\T}{\seq\ve}, \cc\,\seq\cc)=
\begin{cases}
\cmatch(\MCall{\ve}{\m}{\seq\T}{\seq\ve}, \seq\cc)&\mbox{if}\ \cmatch(\MCall{\ve}{\m}{\seq\T}{\seq\ve}, \cc)\ \mbox{undefined}\\
\cmatch(\MCall{\ve}{\m}{\seq\T}{\seq\ve}, \cc)&\mbox{otherwise}\\
\end{cases}
$\\
$\cmatch(\MCall{\ve}{\m}{\seq\T}{\seq\ve}, \epsilon)=$\ \mbox{undefined}
\end{quoting}
\end{small}
\caption{Method look-up and matching}\label{fig:auxfun} 
\end{figure}
When a method is invoked on an object receiver, it is first searched among the method declarations in the object itself (first clause).   
Objects  cannot  declare magic methods, since they are made available to the programmer through predefined interfaces,  as illustrated by examples later.    If not found, then look-up is propagated to the parent  types. In this case, there should be exactly one parent type where method look-up is successful. In a nominal type, analogously, the method is first searched among those defined in the corresponding  declaration.  If a defined method is found, then the type parameters of the declaration are replaced by the corresponding type arguments in the nominal type.  If a magic method is found, necessarily in a type declaration, then the result is an  $\mgc$ tag, and the type name. If the method is not found, then look-up is propagated to the parents, again replacing type parameters  with  the corresponding type arguments. 

The $\cmatch$ function
is defined in \cref{fig:auxfun} as well.
We write $\instanceof{\ve}{\NT}$ meaning that the dynamic type of $\ve$, which can be extracted from $\ve$ by just erasing method bodies, defined in the obvious way, is a subtype of $\NT$.
Note that the extraction of the dynamic type and the subtyping check, being part of the runtime semantics, are purely syntactic. Notably, the extracted type could  violate constraints on conflicting methods and overriding (which will be checked by the type system) or even be ill-formed. 
If a (first) matching clause is found, then such clause provides an alternative body with its type parameters and parameters, which are instantiated with the corresponding arguments in the magic call.  This resembles very much what happens for the call of a defined method,  except that  look-up is performed in the catch clauses, following the syntactic order.

In \cref{fig:monadic-red} we give the rules for the monadic reduction.
\begin{figure}[t]
\begin{small}
\begin{math}
\begin{array}{c}
\NamedRule{pure}{
  \e \purered \e' 
}{ \e \red \mun(\e') }
{ } 
\qquad

\NamedRule{mgc}{ }
{ \MCall{\ve}{\m}{\seq\T}{\seq\ve} \red \Mmap{(\Ret{\ehole})}{\mrun\tname\m(\ve,\seq\ve) }  }
{ 
  \mbody(\ve,\m)=\Pair{\mgc}{\tname} 
}

\\[3ex]
\NamedRule{ret}{ }
{ \Do{\x}{\Ret\ve}{\e} \red \mun(\Subst\e\ve\x) }
{ } 
\qquad 
\NamedRule{do}{
  \e_1 \red \me 
}{ \Do\x{\e_1}{\e_2} \red \Mmap{(\Do\x{\ehole}{\e_2})}{\me} } 
{ }
\end{array}
\end{math}
\end{small}
\caption{Monadic (one-step) reduction}\label{fig:monadic-red} 
\end{figure}
An expression is reduced to a monadic expression either by propagating a pure step, embedding its result in the monad, as shown in rule \refToRule{pure},  or by interpreting in the monad magic calls and \lstinline{do} expressions, as shown in the following rules. 

In rule \refToRule{mgc}, method look-up finds a magic method in a type declaration named $\tname$, and the call reduces to the corresponding monadic expression. In other words, the effect is actually raised. 
To this end, we apply the function of type $\funtype{\mfun\Val}{\mfun\Exp}$ obtained by lifting, through $\aux{map}$, the context $\Ret{\ehole}$ to the monadic value obtained from the call.
 Here we identify the context $\Ret{\ehole}$,  which is an expression with a hole, with the function $\ve \mapsto \Ret[\ve]$ of type $\Val\to\Exp$. 
 
  Rules \refToRule{ret} and \refToRule{do} are the monadic version of the standard ones for these constructs.
More in detail, given a \lstinline{do} expression, when the first subterm does not raise effects, just returning a value, the expression can be  reduced to the monadic  embedding  of the second subterm, after replacing the variable with the returned value, as shown in rule \refToRule{ret}.
Rule \refToRule{do}, instead, propagates the reduction of the first subterm, taking into account possibly raised effects. To this end, we apply the function of type $\funtype{\mfun\Exp}{\mfun\Exp}$ obtained by lifting, through $\aux{map}$, the context $\Do\x{\ehole}{\e_2}$, to the monadic expression obtained from $\e_1$. 
 Analogously to above,  we identify the context $\Do\x{\ehole}{\e_2}$, which is an expression with a hole, with the function $\e \mapsto \Do{\x}{[\e]}{\e_2}$ of type $\Exp\to\Exp$. 
 
\smallskip
\noindent\textit{Monadic small-step reduction and semantics} 
Following the approach in \cite{DagninoGZ25}, we can define, on top of the monadic one-step reduction $\red$  on  $\Exp \times \mfun\Exp$: 
\begin{itemize}
\item
a small-step reduction on \emph{monadic configurations}
\item a \emph{finitary semantics} of expressions
\item  assuming an appropriate structure on the monad,  an \emph{infinitary semantics} of expressions.
\end{itemize} 
To apply the construction in  \cite{DagninoGZ25}, the monadic one-step reduction is required to be deterministic, and this is the case indeed.
\begin{proposition}[Determinism] \label{prop:det}
If $\e\red \me_1$ and $\e\red\me_2$ then  $\me_1 = \me_2$. 
\end{proposition}
We outline the construction, referring to \cite{DagninoGZ25} for technical details and proofs.

First of all, we set $\Conf = \Exp + \Res$, where $\Res = \Val + \Wrng$, $\Wrng = \{\wrng\}$. That is, a \emph{configuration} $\conf$ is either an expression or a result $\res$, which, in turn, is either a value, modelling successful termination, or $\wrng$, modelling  a stuck computation.
Then, we extend the  monadic  reduction $\red$ to configurations, getting the relation 
$\tred \subseteq \Conf \times \mfun\Conf$ shown in \cref{fig:mred-conf}.
\begin{figure}[t]
\begin{small}
\begin{math}
\begin{array}{c}
\NamedRule{exp}{
  \e\red\me 
}{ \e \tred \me }
{ } 
\BigSpace
\NamedRule{ret}{
}{\Ret{\ve}\tred \mun_\Conf(\ve) }
{ } 
\\[3ex]
\NamedRule{wrong}{
}{\e\tred \mun_\Conf(\wrng) }
{ \e\not\red\\
\e\ne\Ret{\ve}\ \mbox{for all}\ \ve\in\Val} 
\BigSpace
\NamedRule{res}{
}{\res \tred \mun_\Conf(\res) }
{ } 
\end{array}
\end{math}
\end{small}
\caption{Monadic (one-step) reduction on configurations}\label{fig:mred-conf} 
\end{figure}
Expressions which represent terminated computations reduce to the monadic embedding of the corresponding value or $\wrng$, respectively; moreover, results (either values or $\wrng$) conventionally reduce \mbox{to their monadic embedding.}

Now, we can define  a relation $\Red$ on $\mfun\Conf$ as follows:
\begin{quoting}
$\mconf\Red\mconf'$ iff $\mkl\tredfun(\mconf) = \mconf'$. 
\end{quoting}
In this way, computations on monadic configurations are described, as usual in small-step style, by sequences of $\Red$ steps. 

Then, the finitary semantics is the function \fun{\finsem{-}}{\Exp}{\mfun\Res + \{\infty\}} defined as follows\footnote{ Here and  in rule \refToRule{exp} we omit the injection into monadic configurations.}: 
\begin{quoting}
\begin{math}
\finsem\e = \begin{cases}
\mres & \text{if $\mun_\Conf(\e)\Redstar\mres$} \\ 
\infty & \text{otherwise} 
\end{cases}
\end{math}
\end{quoting}
where $\Redstar$ is the reflexive and transitive closure of $\Red$. 
This semantics describes only monadic results that can be reached in finitely many steps. 
In other words, all diverging computations are identified and no information on computational effects they may produce is available. 


To overcome this limitation, we introduce an \emph{infinitary semantics}.
As formally detailed in \cite{DagninoGZ25}, we assume, for each set $X$, a partial order on $\mfun X$ with a least element $\mbot_X$ and suprema $\msup$ of $\omega$-chains.
\label{res}
 Let \fun{\ctr}{\Conf}{\mfun\Res} be the function given by 
\begin{quoting}
\begin{math}
\ctr(\conf) = \begin{cases}
\mbot_\Res & \conf = \e \\ 
\mun_\Res(\res) & \conf = \res 
\end{cases}
\end{math}
\end{quoting}
For every $\e\in\Exp$ and $n\in\N$, we define 
$\infsem[n]{\e} = \mkl{\ctr}(\mconf)$ iff $\mun_\Conf(\e)\Red^n \mconf$. 
Note that $\infsem[n]{\e}$ is well-defined for all $n\in\N$ since $\Red$ is (the graph of) a total function and so is its $n$-th iteration.  
The sequence $(\infsem[n]{\e})_{n\in\N}$ turns out to be an $\omega$-chain, so we define the infinitary semantics as 
\begin{quoting}
\begin{math}
\infsem{\e} = \msup_{n\in\N} \infsem[n]{\e} 
\end{math}
\end{quoting}
Intuitively, $\infsem[n]{\e}$ is the portion of the  result that is reached after $n$ reduction steps. 
Hence, the actual result is obtained as the supremum of all such approximations, thus describing also the observable behaviour of  possibly diverging computations. 

\smallskip
\noindent\textit{Examples} 
We show now examples of type declarations providing the user interface to raise effects in an underlying monad, together with some reductions and semantics. 
Booleans, natural numbers, conditional, and checking that a number is even, are encoded as shown in \cref{ex:boolnat}, except that, for brevity, we assume that \lstinline{not} in \lstinline{True} directly returns \lstinline{False}. 
We recall the (inductive) definition of $\nat{n}$, the object \mbox{encoding the natural number $n$: }
\begin{quoting}
$\nat{0}=$ \lstinline{Zero}\HugeSpace
$\nat{n+1}=$ \lstinline@Succ{pred: def -> Nat <$\_$, $\nat{n}$>}@
\end{quoting}
 Finally,  a handler of shape  $\Handler{\seq\cc}{\x}{\Ret{\x}}$ is abbreviated  by  $\seq\cc$. 

\begin{example}[Exceptions]\label{ex:exception}
Let us fix a set $\ExSet$. 
The monad $\Mnd{\ExceptFun[\ExSet]} = \ple{\ExceptFun[\ExSet],\Mun{\ExceptFun[\ExSet]},\Mmul{\ExceptFun[\ExSet]}}$ is given by 
$\ExceptFun[\ExSet] X = X+\ExSet$, and 
\begin{quoting}
\begin{math}
\Mun{\ExceptFun[\ExSet]}(x) = x 
\HugeSpace
\alpha \mbind f = \begin{cases}
f(x)\ & \mbox{if}\ \alpha = x\in X \\
\alpha\ & \text{otherwise ($\alpha = \exc\in \ExSet$)} 
\end{cases}  
\end{math} 
\end{quoting}
where $+$ denotes disjoint union (coproduct) and, for simplicity, we omit the injections. 

We define a nominal type $\Exception$ of the (objects representing) exceptions. More formally, we assume that, for each $\ve$  such that $\instanceof{\ve}{\Exception}$, that is, the dynamic type of $\ve$ is a subtype of $\Exception$,  there is an associated exception in $\ExSet$, denoted $\getExc{\ve}$.  
As a minimal example, 
%
we define, besides $\Exception$, a subtype $\MyException$, and assume\footnote{In this simple example $\aux{exc}$ only depends on the exception type; generally, it could depend on a specific instance, for instance if exceptions provide error messages or other information.} $\getExc{\Exception}=\Exc$, and $\getExc{\MyException}=\MyExc$, with $\Exc$, $\MyExc$ elements of $\ExSet$. 
\begin{lstlisting}
Exception { throw : mgc[X] -> X }

MyException$\ext$Exception { throw : mgc[X] -> X }
\end{lstlisting}
Both types declare 
a magic method \lstinline{throw} that can be called to raise the exception corresponding to the receiver object. 
%
Note that the method \lstinline{throw} declares a type variable as result type, since it could be called in an arbitrary context.  
We have:
\begin{itemize}
\item $\pfun{\mrun\Exception\throw}{\Val\times\Val^\star}{\mfun\Val}$ only defined on $\Pair{\ve}{\epsilon}$  with $\instanceof{\ve}{\Exception}$\\
$\mrun{\Exception}{\throw}(\ve,\epsilon)=\getExc{\ve}$
\item $\pfun{\mrun{\MyException}\throw}{\Val\times\Val^\star}{\mfun\Val}$ only defined on $\Pair{\ve}{\epsilon}$  with $\instanceof{\ve}{\MyException}$\\
$\mrun{\MyException}{\throw}(\ve,\epsilon)=\getExc{\ve}$
\end{itemize}
The \lstinline{throw} magic method in $\MyException$ has the same behaviour of that in the parent type; however, redefining the method will be significant in the type-and-effect system (\cref{sect:effect-system}), to precisely track the possibly raised exceptions.

The type $\My$ below defines a method that, depending on whether its
argument (a natural number) is either even or odd, returns a natural or calls the method $\throw$ of $\MyException$.  
\begin{lstlisting}
My { 
  m : def Nat -> Nat<x y, y.match(Even).if[Nat ThenElse[Nat]](MyTE)}> 
}

MyTE = ThenElse[Nat]{
    then : def [Nat] ->Nat <_, return One> 
    else : def [Nat] ->Nat <_, Exception.throw[Nat]()>
} 

One = $\nat{1}$ = Succ{pred : def -> Nat <_,Zero>}
\end{lstlisting}
In the following reductions, 
we use the abbreviations $\MyTE$ and $\One$ above  and
write $\If$ for $\ifNatTE$.

Consider the expression:
\begin{quoting}
$\e_1 = \TryNoFinal{\MCallNG\My\mt{\One}}{\cc_1}$ where
$\cc_1=\CC{\Exception}{\throw}{\X}{}{\x}{\Ret\ \One}{ \Stop}$
\end{quoting}
We have the following small-step reduction sequences on monadic configurations\footnote{We omit the injections from monadic expressions and values.}:
\begin{quoting}
\begin{math}
\begin{array}{lcl}
\e_1&\Red&  \TryNoFinal{\MCallNG{\MCallNG\One\match{\Even}}{\If}{\MyTE}}{\cc_1}\\
&\Red & \TryNoFinal{\MCallNG{ \MCallNG\Even\succtt{\MCallNG\One\pred}}    {\If}{\MyTE}}{\cc_1}\\
&\Red & \TryNoFinal{\MCallNG{ \MCallNG{\MCallNG{\MCallNG\One\pred{} }\match\Even}\nottt{}  }    {\If}{\MyTE}}{\cc_1}\\
&\Red & \TryNoFinal{\MCallNG{ \MCallNG{\MCallNG{\Zero }\match\Even}\nottt{}  }    {\If}{\MyTE}}{\cc_1}\\
&\Red & \TryNoFinal{\MCallNG{ \MCallNG{\MCallNG{\Even }\zero{}}\nottt{}  }    {\If}{\MyTE}}{\cc_1}\\
& \Red& \TryNoFinal{\MCallNG{ \MCallNG{\True}\nottt{}  }    {\If}{\MyTE}}{\cc_1}\\
& \Red& \TryNoFinal{\MCallNG{\False} {\If}{\MyTE}}{\cc_1}\\
& \Red& \TryNoFinal{\MCallNG{\Exception} {\throwNat}{}}{\cc_1}\\
& \Red& \Ret\ \One \Red \One
\end{array}
\end{math}
\end{quoting}
 where all steps are derived by  rules ~\refToRule{pure} in \cref{fig:monadic-red} and \refToRule{exp} in \cref{fig:mred-conf}, except the last one, which is derived by rule \refToRule{ret} in \cref{fig:mred-conf}. 

Considering now:
\begin{quoting}
$\e_2 = {\TryNoFinal{\MCallNG\My\mt{\One}}{\cc_2}}$ where
$\cc_2=\CC{\MyException}{\throw}{\X}{\x}{}{\Ret\ \One}{}$
\end{quoting}
we have
\begin{quoting}
\begin{math}
\begin{array}{lcl}
\e_2&\Redstar&  \TryNoFinal{\MCallNG{\False} {\If}{\MyTE}}{\cc_2}\text{ first 8 steps as for $e_1$}\\
&\Red& \Do{\x}{\MCallNG\Exception\throw{}}{\Ret\ \Zero}\\
&\Red&    \Exc\ \ \text{since }\\
& & \Do{\x}{\MCallNG\Exception\throw{}}{\Ret\ \Zero}\red\Exc\text{  by  rule~\refToRule{do}} \\
\end{array}
\end{math}
\end{quoting}
\end{example}
 
As the reader may have noted,  the $\Stop$-clause models the expected behaviour of exceptions, which, even when caught, interrupt the normal flow of execution. In the following example we show a different interface, still using the exception monad, paired with $\Continue$-clauses, which replace an effect with an alternative behaviour in a continuous manner.
\begin{example}[Failure]\label{ex:fail}\footnote{This example is a minimal version of one in \cite{BrachthauserSO20jfp}.}
Consider the exception monad with  $\Fail\in\ExSet$, and the  following  type declarations,
with $\mrun{\Failure}{\fail}$ returning $\Fail$. 
\begin{lstlisting}
Failure[X]{ fail : mgc -> X }

String{ ...  toNat: def -> Nat < ... > }

Test {
 sumAsNat: def String String -> Nat
   <$\_$ s1 s2,do n1=s1.toNat();do n2=s2.toNat();do n=n1.sum(n2);return n>
}    
\end{lstlisting}
Method \lstinline{toNat}, whose implementation is omitted, is expected to return the natural number represented by a string, if any; otherwise, $\Fail$ is raised, through the magic call $\Failure[\Nat].\fail()$. 
We assume  a method \lstinline{sum} in \lstinline{Nat} returning the sum of two numbers. 

The  expression \lstinline{Test.sumAsNat(s1,s2)} clearly raises $\Fail$ if one of the two strings does not represent a natural number. 
However, we can catch such effect returning a default value:
\begin{lstlisting}
try Test.sumAsNat(s1,s2) with Failure[Nat].fail : <_ , return Zero>$_\Continue$ 
\end{lstlisting}
In this way, we always get a result; in particular, if \lstinline{s1} does not represent a natural number, and \lstinline{s2} represents the natural number \lstinline{n}, we get \lstinline{n}, as expected.
Instead, with a $\Stop$-clause, we would get \lstinline{Zero}, without performing the sum. An example of reduction sequence is below.\footnote{ We use $\kw{ret}$ for $\kw{return}$ to save space.}
%

\begin{small}
$\e = \TryNoFinal{\MCallNG\Test\sumAsNat{\aU,\aD}}{\cc}$ where
$\cc=\CCSh{\Failure[\Nat]}{\fail}{\_}{}{\Ret\ \Zero}{\Continue}$
\begin{math}
\begin{array}{lcl}
\e&\Red&  \TryNoFinal{\DoS{\nU}{\MCallNG\aU\toNat{}}  {\DoS{\nD}{\MCallNG\aD\toNat{}}  {\DoS{\nT}{\MCallNG\nU\sumT{\nD}}{\RetS{\nT}   }}    }  }{\cc}\\
&\Red & \TryNoFinal{\MCallNG\aU\toNat{}}{\cc,\MBody{\nU} {} {\TryNoFinal{    \DoS{\nD}{\MCallNG\aD\toNat{}}  {\DoS{\nT}{\MCallNG\nU\sumT{\nD}}{\RetS{\nT}   }}       }{\cc} }  }         \\
&\Redstar& \TryNoFinal{\RetS\One}{\cc,\MBody{\nU} {} {\TryNoFinal{    \DoS{\nD}{\MCallNG\aD\toNat{}} {\DoS{\nT}{\MCallNG\nU\sumT{\nD}}{\RetS{\nT}   }}   }{\cc} }  }         \\
&\Red & \Do {\nU}{\Ret \One}{ {\TryNoFinal{    \Do{\nD}{\MCallNG\aD\toNat{}}  {\Ret{\MCallNG\nU\sumT{\nD}}}}{\cc} } }         \\
&\Red &  {\TryNoFinal{    \DoS{\nD}{\MCallNG\aD\toNat{}}  {\DoS{\nT}{\MCallNG\One\sumT{\nD}}{\RetS{\nT}   }}    } {\cc} }       \\   
&\Red & \TryNoFinal{\MCallNG\aD\toNat{}}{\cc,\MBody{\nD} {} {\TryNoFinal{  {\DoS{\nT}{\MCallNG\One\sumT{\nD}}{\RetS{\nT}   }}  }{\cc} }  }         \\
&\Redstar & \TryNoFinal{\Failure[\Nat].{\fail}}{\cc,\MBody{\nD} {} {\TryNoFinal{  {\DoS{\nT}{\MCallNG\One\sumT{\nD}}{\RetS{\nT}   }}  }{\cc} }  }   (\ast)      \\
&\Red & \Do {\nD}{\Ret \Zero}{{\TryNoFinal{  {\DoS{\nT}{\MCallNG\One\sumT{\nD}}{\RetS{\nT}   }} }{\cc} } }         \\
&\Red &   {{\TryNoFinal{  {\DoS{\nT}{\MCallNG\One\sumT{\Zero}}{\RetS{\nT}   }}  }{\cc} } }         \\
&\Red &   {{\TryNoFinal{\MCallNG\One\sumT{\Zero}}  { \cc, \MBody {\nT}{}{\RetS{\nT} } } } }         \\
&\Redstar &   {{\TryNoFinal{\Ret\One}  { \cc, \MBody {\nT}{}{\RetS{\nT} } } } }         \\
&\Red  &  \Do {\nT}{\Ret \One}{\Ret\nT } \\
&    \Red  &  \Ret \One   \Red \One 
\end{array}
\end{math}
\end{small}

 If the catch $\Continue$-clause were a $\Stop$-clause, then the expression at line $(\ast)$ would  reduce to $\Ret\Zero$ and therefore the reduction would produce $\Zero$. 

It is worthwhile to compare the interfaces:
\begin{lstlisting}
Exception { throw : mgc[X] -> X }

Failure[X]{ fail : mgc String -> X }
\end{lstlisting} 
\lstinline{Exception} offers a generic method \lstinline{throw}, so that a magic call \lstinline{Exception.throw()} can occur in any context. On the other hand, \lstinline{Failure} is a parametric type, so that a magic call \lstinline{Failure[T].fail()} can only occur where a \lstinline{T} is expected. 
\end{example}

\begin{example}[Non-determinism]\label{ex:pow}
The monad $\Mnd\ListFun = \ple{\ListFun,\Mun\ListFun, \Mmul\ListFun}$ is given by
$\ListFun X$ the set of (possibly infinite) lists over $X$, coinductively defined by the following rules: 
$\elist\in\ListFun(X)$ and, 
if $x\in X$ and $l\in\ListFun(X)$, then $x\cons l \in\ListFun(X)$. 
We use the notation $[x_1,\ldots,x_n]$ to denote the finite list $x_1\cons \ldots\cons x_n\cons\elist$. 
Then, the unit is 
$\Mun\ListFun_X(x) = [x]$, and the bind is corecursively defined as follows: 
$\elist\mbind f = \elist$ and 
$(x\cons l)\mbind f = f(x)(l\mbind f)$, 
where juxtaposition denotes the concatenation of possibly infinite lists.
%
We define the nominal type $\Chooser$ declaring the magic method
$\choott$,  with $\mrun\Chooser\choott$  returning the list   (monadic value) consisting of  the  values $\True$ and $\False$
and we use the abbreviations $\MyTEU$ and $\MyTED{\y}$ defined below. 
\begin{lstlisting}
Chooser { choose : mgc -> Bool }

My {
  m1 : def -> Nat  <$\_$, do z = Chooser.choose();
                          z.if[Nat ThenElse[Nat]](MyTE1)>
  m2 : def Nat -> Nat <$\_$ y, do z = Chooser.choose(); 
                               z.if[Nat ThenElse[Nat]](MyTE$^\ytt$)>
}

MyTE1 = ThenElse[Nat]{
  then : def [Nat] -> Nat <_, return One> 
  else : def [Nat] -> Nat <_, return Zero>
} 

MyTE$^\ytt$ = ThenElse[Nat]{
  then : def [Nat] ->Nat <_, return y> 
  else : def [Nat] ->Nat <_, My.m2(y.succ())>
} 
 
Two = $\nat{2}$ = Succ{ pred: def ->Nat <_,Succ{pred: def ->Nat <_,Zero>}> }
\end{lstlisting}
We have the following small-step reduction sequences:

\begin{small}
\begin{math}
\begin{array}{lcl}
\List{\MCallNG\My\mU{}}&\Red&  \List{\DoS{\ztt}{\MCallNG\Chooser\choott{}}  {\MCallNG\ztt\If\MyTEU}} \\
& \Red & \List{\DoS{\ztt}{\RetS\True}  {\MCallNG\ztt\If\MyTEU},\DoS{\ztt}{\RetS\False}  {\MCallNG\ztt\If\MyTEU} }\\
& \Red & \List{{\MCallNG\True\If\MyTEU}, {\MCallNG\False\If\MyTEU} }\\
& \Red & \List{\RetS \One, \RetS\Zero }\\
& \Red & \List{ \One,\Zero }\\[1ex]
 \List{\MCallNG\My\mD{\Zero}}&\Red&  \List{\DoS{\ztt}{\MCallNG\Chooser\choott{}}  {\MCallNG\ztt\If{\MyTED{\ytt}}}} \\
& \Red & \List{\DoS{\ztt}{\RetS\True}  {\MCallNG\ztt\If{\MyTED{\ytt}}}\, ,\,\DoS{\ztt}{\RetS\False}  {\MCallNG\ztt\If{\MyTED{\ytt}}} }\\
& \Red & \List{{\MCallNG\True\If{\MyTED{\ytt}}}\, ,\,{\MCallNG\False\If{\MyTED{\ytt}}} }\\
& \Red & \List{\RetS \Zero\, ,\, \MCallNG\My\mD{\MCallNG\Zero\succtt{}}  }\\
& \Red & \List{ \Zero\, ,\, \MCallNG\My\mD{\One} }\\
& \Red & \List{ \Zero\, ,\, \DoS{\ztt}{\MCallNG\Chooser\choott{}}  {\MCallNG\ztt\If{\MyTED{\ytt}}}}\\
& \Red & \List{ \Zero\, ,\, \DoS{\ztt}{\RetS\True}  {\MCallNG\ztt\If{\MyTED{\ytt}}}\, ,\,\DoS{\ztt}{\RetS\False}  {\MCallNG\ztt\If{\MyTED{\ytt}}} }\\
& \Red & \List{ \Zero\, ,\, {\MCallNG\True\If{\MyTED{\ytt}}}\, ,\,{\MCallNG\False\If{\MyTED{\ytt}}}  }\\
& \Red & \List{ \Zero\, ,\, \RetS \One\, ,\,  \MCallNG\My\mD{\MCallNG\One\succtt{}}  }\\
& \Red & \List{ \Zero\, ,\,  \One\, ,\,  \MCallNG\My\mD{\Two }}\\
& \Redstar & \List{ \Zero\, ,\,  \One\, ,\,  \Two\, ,\,\MCallNG\My\mD{\MCallNG\Two\succtt{}}}\\
& \cdots &
\end{array}
\end{math}
\end{small}

Note that the second reduction is non-terminating,  in the sense that  a monadic result (a list of values) is never reached. Hence, with the finitary semantics, we get $\finsem{\List{\MCallNG\My\mD{\Zero}}}=\infty$.  
With the infinitary semantics, instead, we get the following $\omega$-chain:
\begin{quoting}
\lstinline{[]}, \ldots, \lstinline{[Zero]}, \ldots, \lstinline{[Zero, One]}, \ldots, \lstinline{[Zero, One, Two]}, \ldots, \lstinline{[$\nat{0}$, $\ldots$, $\nat{n}$]}, \ldots, 
\end{quoting}
whose supremum is the infinite list of the (objects representing the) natural numbers. 
 \end{example}

\begin{example}\label{ex:prob}
Denote by 
$\DistFun X$ the set of probability subdistributions $\alpha$ over $X$ with countable support, i.e.,
\fun{\alpha}{X}{[0..1]} with $\sum_{x\in X}\alpha(x) \leq 1$ and $\Supp(\alpha) = \{ x \in X \mid  \alpha(x)\ne 0 \}$ countable set. 
We write $r\cdot\alpha$ for the pointwise multiplication of a subdistribution $\alpha$ with a number $r\in[0,1]$. 
The monad $\Mnd\DistFun = \ple{\DistFun,\Mun\DistFun,\Mmul\DistFun}$ is given by 
\begin{quoting}
\begin{math}
\Mun\DistFun(x) = y \mapsto \begin{cases}
  1 & y = x \\
  0 & \text{otherwise}
\end{cases} \HugeSpace
\alpha \mbind f = \sum_{x\in X} \alpha(x)\cdot f(x) 
\end{math}
\end{quoting}
\end{example}

We use the same nominal type $\Chooser$ with the magic method $\choott$, now
returning the  distribution  consisting of  the   values $\True$ and $\False$  with probability $\frac{1}{2}$, that we denote by ${\List{\frac{1}{2}:\Ret\True,\frac{1}{2}:\Ret\False}}$.

The expressions $\List{1:\MCallNG\My\mU{}}$ and $\List{1:\MCallNG\My\mD{\Zero}}$ can be reduced  analogously to  the previous example:
\begin{quoting}
$\List{1:\MCallNG\My\mU{}}\Redstar\List{\frac{1}{2}:\One,\frac{1}{2}:\Zero}$\\
$\List{1:\MCallNG\My\mD{\Zero}}\Redstar \List{ \frac{1}{2}:\Zero,\frac{1}{4}:\One,\frac{1}{8}:\Two,\frac{1}{16}:\MCallNG\My\mD{\MCallNG\Two\succtt{}}}\ \ldots$
\end{quoting}

Again, the second reduction is non-terminating, hence, with the finitary semantics, we get $\finsem{\List{1:\MCallNG\My\mD{\Zero}}}=\infty$.  
With the infinitary semantics we get the following $\omega$-chain:
\begin{quoting}
\lstinline{[]}, \ldots, \lstinline{[$\frac{1}{2}$:Zero]}, \ldots, \lstinline{[$\frac{1}{2}$:Zero, $\frac{1}{4}$:One]}, \ldots, \lstinline{[$\frac{1}{2}$:$\nat{0}$, $\ldots$, $\frac{1}{2^{n+1}}$:$\nat{n}$]}, \ldots, 
\end{quoting}
whose supremum is the infinite list where each (object representing the) number $n$ has \mbox{probability $\frac{1}{2^{n+1}}$. }

%% file: effect-system.tex

\section{Type-and-effect System}\label{sect:effect-system}

In order to equip the language with  a type-and-effect  system, first of all we extend the syntax and the signatures, as shown in  \cref{fig:effects}.  
\begin{figure}[th]
\begin{math}
\begin{grammatica}
\produzione{\MT}{\TMeth{\Ext{\seq\X}{\seq\UT}}{\seq\T}{\TEff{\T}{\eff}}}{method type-and-effect} \\
\produzione{\md}{\MethDec{\m}{\defn}{\MT}{\Pair{\x\,\seq\x}{\e}}\mid\MethDec{\m}{\abs}{\MT}{}\mid\MethDec{\m}{\mgc}{\MT}{}}{method declaration}\\
\produzione{\eff}{\eZero\mid\eTop\mid\EComp{\eff}{\eff'}\mid\eCall{\T}{\m}{\seq{\T}}}{effect} \\
\produzione{\sig}{\seq\m:\seq{\kind}\ \seq\MT}{signature}  \\
\produzione{\kind}{\abs \mid \defn\mid\mgc}{(method) kind}\\
\end{grammatica}
\end{math}
\caption{Adding effects}\label{fig:effects} 
\end{figure}
In method declarations, method types are replaced by \emph{method type-and-effects}, where an \emph{effect} component is added; we maintain the same meta-variable for simplicity. 
They are considered equal up-to $\alpha$-renaming. That is,  $\TMethEff{\Ext{\seq\X}{\seq\UT}}{\seq\T}{\T}{\eff}=\Subst{(\TMethEff{\Ext{\seq\Y}{\seq{\UT'}}}{\seq{\T'}}{\T'}{\eff'})}{\seq\X}{\seq\Y}$. 

Effects are the empty effect, the top effect, union of effects, and  \emph{call-effects}. For magic methods, this component is assumed to have a  canonical  form, hence can be omitted in  the  concrete syntax. Notably, for a method  $\MethDec{\m}{\mgc}{\TMeth{\Ext{\seq\X}\_}{\_}{\_}}{}$ in the declaration of the nominal type $\Gen\tname{\Ext{\seq\Y}{\_}}$,  the effect is $\eCall{\Gen{\tname}{\seq\Y}}{\m}{\seq\X}$.
In signatures, the information associated to method names is analogously extended; moreover, the additional kind $\mgc$ is considered.  

\begin{figure}[th]
\begin{math}
\begin{grammatica}
\produzione{\TEnv}{{\seq\X}\subt{\seq\T}}{type environment}\\
\produzione{\Gamma}{{\seq\x}:{\seq\T}}{environment}
\end{grammatica}
\end{math}
\caption{Syntax of (type) environments}\label{fig:env} 
\end{figure}

Type environments and environments, defined in \cref{fig:env}, are assumed to be maps, from type variables to   types   (their bounds), and from variables to types, respectively. Hence (type) variables are distinct, and the two sequences have the same length.

Before the formal details, we illustrate the  most  distinctive feature of our type system.

 \smallskip
\noindent\textit{Variable call-effects and simplification} 
As shown in \cref{fig:effects},  except for  $\eTop$, effects are essentially (representations of) sets of call-effects, with $\eZero$ the empty set and $\EComp{}{}$ the union.  Sets of ``atomic'' effects are a rather natural idea, generalizing what happens, e.g., in Java \lstinline{throws}  clauses. What is interesting here is the nature of such atomic effects, and the role of $\eTop$. 
 A call-effect $\eCall{\T}{\m}{\seq{\T}}$ is a static approximation of the computational effects of a call to $\m$ with receiver of type $\T$ and type arguments $\seq\T$.  
More precisely, we only allow call-effects which are:

\begin{description}
\item[magic] $\eCall{\T}{\m}{\seq{\T}}$ with $\T$ object type (that is, not of shape $\X$), and $\m$ magic in $\T$. As expected, this means that this magic method could be possibily invoked, raising the corresponding computational effect; for instance, as will be shown in \cref{ex:simpl}, a call-effect \lstinline{Exception.throw} denotes that an exception could be possibly thrown.
\item[variable]  $\eCall{\X}{\m}{\seq\T}$. This call-effect can be assigned to code parametric on the type variable $\X$; the meaning is that, for each instantiation of $\X$ with an object type $\T$, this becomes, through a non-trivial process called \emph{simplification}, the effect of $\m$ in $\T$. In other words, $\eCall{\X}{\m}{\seq\T}$  is a parametric effect,  which can be made concrete in the types which replace $\X$. 
\end{description}
In addition to sets of call-effects, we include the \emph{top} effect, which plays the role of default for (typically abstract) methods which do not pose constraints on the effects in implementations. In a sense, this generalizes the meaning of a \mbox{\lstinline{throws Exception} clause in Java. }

We illustrate now the above features on an example, notably showing how (simplification of) variable call-effects allows a very precise approximation. 

\begin{example}\label{ex:simpl}
 Consider again the encoding of booleans in \cref{ex:boolnat} (\cref{sect:lang}), where we added effect annotations.
 \begin{lstlisting}
Bool { if : abs [X Y$\ext$ThenElse[X]] Y -> X ! Y.then $\vee$ Y.else }

True $\ext$ Bool { 
  if : def [X Y$\ext$ThenElse[X]] Y -> X ! Y.then <$\_$ te, return te.then()>
}

False $\ext$ Bool {
  if : def [X Y$\ext$ThenElse[X]] Y -> X ! Y.else <$\_$ te, return te.else()>
}

ThenElse[X] {
  then: abs -> X ! $\eTop$
  else: abs -> X ! $\eTop$
}
\end{lstlisting}
 The abstract method {if} in \lstinline{Bool} has, as parameter type,  the type variable \lstinline{Y}, expected to be instantiated with subtypes of \lstinline{ThenElse[X]}, hence providing  methods \lstinline{then} and \lstinline{else}. As effect, the method declares the union of two variable call-effects,  \lstinline{Y.then $\vee$ Y.else}. This means that the computational effects of this method can only be those propagated from calling either \lstinline{then} or \lstinline{else} on the argument.  Subtypes \lstinline{True} and \lstinline{False} specializes the effect of \lstinline{if} as expected, since they select the \lstinline{then} and the \lstinline{else} alternative provided by the argument, respectively. 

In the type  \lstinline{ThenElse[X]} the  methods \lstinline{then} and \lstinline{else} are abstract,  and their implementation in subtypes is allowed to raise arbitrary effects, as denoted by $\eTop$. 
 This makes sense, since we would like to instantiate \lstinline{ThenElse[X]} on types implementing these methods in arbitrary ways. 

For instance, consider the following code, where \lstinline{b} is an expression of type \lstinline{Bool}, and \lstinline{Exception}, $\MyTE$ are the type and object introduced in \cref{ex:exception}, 
 where the latter has been annotated with effects, and \lstinline{MyTEType} is the corresponding object type.

\begin{lstlisting}
b.if[Nat, MyTEType](MyTE)

MyTE = ThenElse[Nat]{
  then: def  -> Nat ! $\eZero$ <$\_$, return Zero>
  else: def  -> Nat ! MyException.throw[Nat] <$\_$, MyException.throw()>
}

MyTEType = 
ThenElse[Nat]{
  then: def -> Nat ! $\eZero$, 
  else: def -> Nat ! MyException.throw[Nat]
}
\end{lstlisting}
The argument passed to \lstinline{if} is an object implementing \lstinline{then} by returning $\Zero$, hence declaring no effects, whereas \lstinline{else} calls a magic method, and declares the corresponding effect. 

The effect computed for the call is, as intuitively expected, \lstinline@MyException.throw[Nat]@. This happens thanks to effect simplification, which will be formally specified in \cref{fig:simplify}.
Indeed, looking for method \lstinline{if} in the receiver's type \lstinline{Bool} gives the effect \lstinline{X.then $\vee$ X.else}, which is instantiated to the argument type,  giving \lstinline{MyTEType.then $\vee$ MyTEType.else}. This effect would be highly inaccurate. 
However, \lstinline{MyTEType.then} and \lstinline{MyTEType.else} are simplified to the effects of the corresponding methods, that is, $\eZero$ and \lstinline$MyException.throw[Nat]$, respectively. 

 In conclusion, our type-and-effect system supports a form of effect polymorphism which does not need additional ingredients, such as, e.g.,  explicit effect variables. Indeed, the effect of a method can be parametric on those of methods called on type variables in the context (variable call-effects): for instance, the effect of \lstinline{if} is parametric on \lstinline{Y.then} and \lstinline{Y.else}.   
This only relies on existing language features, notably on the OO paradigm. The key point is that a method is identified by, besides its name, the type where it is declared. Hence, to be parametric on the effect of a method, it is enough to be parametric on the type where it is declared, and this is for free since we have type variables. 
In other words, whereas effect variables stand for arbitrary effects, variable call-effects stand for effects declared by a method, and standard instantiation is complemented here by a \mbox{ non-trivial  step of simplification. }

\end{example}

\smallskip
\noindent\textit{Key formal definitions} 
 The typing rules for values and expressions are given in \cref{fig:typing-exp}. 
The typing judgment for values has  the shape $\IsWFVal{\TEnv}{\Gamma}{\ve}{\T}$, since values have no effects;  the one  for expressions,  instead, has  the  shape $\IsWFExp{\TEnv}{\Gamma}{\e}{\T}{\eff }$. 
\begin{figure}[th]
\begin{small}
\begin{math}
\begin{array}{c}
\NamedRule{t-var}{}
{ \IsWFVal{\TEnv}{\Gamma}{\x}{\T} }
{\Gamma(\x)=\T } 
\BigSpace
\NamedRule{t-obj}{
  \IsWFMethod{\TEnv}{\Gamma}{\TObj{\seq\NT}{\sig}}{ \seq{\md}}
  }
{ \IsWFVal{\TEnv}{\Gamma}{\Obj{\seq\NT}{\seq\md}}{\TObj{\seq\NT}{\sig}} }
{ 
\typeof{\TEnv}{\seq\md}{\sig}\\
\typeof{\TEnv}{\TObj{\seq\NT}{\sig}}{\sig'}\\
\NoMgc{\sig}\\
\NoAbs{\sig'}
}\\[6ex]
\NamedRule{t-invk}{
  \IsWFVal{\TEnv}{\Gamma}{\ve_0}{\T_0}\ \BigSpace \IsWFVal{\TEnv}{\Gamma}{\ve_i}{\T'_i}\ \forall i\in 1..n}
  {  \IsWFExp{\TEnv}{\Gamma}{\MCall{\ve_0}\m{\seq\T}{\ve_1,\ldots,\ve_n}}{\Subst{\T}{\seq\T}{\seq\X}}{\eff'}}
{\IsMType{\TEnv}{\T_0}{\m}{\MkTE{\_}{\Ext{\seq\X}{\seq{\UT}}}{\T_1\ldots\T_n}{\T}{\eff}}\\
   \SubType{\TEnv}{\seq\T}{\Subst{\seq{\UT}}{\seq\T}{\seq\X}}\\
     \SubType{\TEnv}{\T'_i}{\Subst{\T_i}{\seq\T}{\seq\X}}\ \forall i \in 1..n\\
 \EffRed{\TEnv}{\Subst{\eff}{\seq\T}{\seq\X}}{\eff'}
}
\\[9ex]
\NamedRule{t-ret}{
  \IsWFVal{\TEnv}{\Gamma}{\ve}{\T}
  }
{ \IsWFExp{\TEnv}{\Gamma}{\Ret\ve}{\T}{\eZero } }
{ } 
\BigSpace
\NamedRule{t-do}{
  \IsWFExp{\TEnv}{\Gamma}{{\e}}{\T}{\eff} \BigSpace
   \IsWFExp{\TEnv}{\Gamma,\TVar{\T}{\x}}{{\e'}}{\T'}{\eff'} 
  }
{ \IsWFExp{\TEnv}{\Gamma}{\Do\x{\e}{\e'}}{\T'}{\EComp{\eff}{\eff'}} }
{ 
} 
\\[4ex]
\NamedRule{t-try}{\IsWFExp{\TEnv}{\Gamma}{\e}{\T}{\eff}\BigSpace
\IsWFHandler{\TEnv}{\Gamma}{\T}{\handler}{\T'}{\hfilter}
}{\IsWFExp{\TEnv}{\Gamma}{\TryShort{\e}{\handler}}{\T'}{\FilterFun{\eff}{\hfilter}}}
{
}
\\[4ex]
\NamedRule{t-handler}
{\IsWFExp{\TEnv}{\Gamma,\TVar{\T}{\x}}{\e'}{\T'}{\eff'}\BigSpace
\IsWFClause{\TEnv}{\Gamma}{\T''}{\cc_i}{\cfilter_i}\ \forall i\in 1..n
}{\IsWFHandler{\TEnv}{\Gamma}{\T}{\Handler{\cc_1\ldots\cc_n}{\x}{\e'}}{\T''}{\HFilter{\cfilter_1\ldots\cfilter_n}{\eff'}}
}{
 \SubType{\TEnv}{\T'}{\T''} 
 }
\\[4ex]
\NamedRule{t-continue}{
  \IsWFExp{\TEnv}{\Gamma,\TVar{\NT_\x}{\x},\TVars{\T}{\x}} {\e} {\T''} {\eff}
}{\IsWFClauseNarrow{\TEnv}{\Gamma}{\T'}{\CC{\NT_\x}{\m}{\seq\X}{\x}{\seq\x}{\e}{\Continue}}{\CFilter{\T_\x}{\m}{\seq{\X}}{\seq\UT}{\eff}{\Continue}}}
{\IsMType{\TEnv}{\NT_\x}{\m}{\MkTE{\mgc}{\Ext{\seq\X}{\seq\UT}}{\seq\T}{\T}{\_}}\\
  \SubType{\TEnv}{\T''}{\T}  
}
\\[5ex]
\NamedRule{t-stop}{
  \IsWFExp{\TEnv}{\Gamma,\TVar{\NT_\x}{\x},\TVars{\T}{\x}} {\e} {\T''} {\eff}
}{\IsWFClauseNarrow{\TEnv}{\Gamma}{\T'}{\CC{\NT_\x}{\m}{\seq\X}{\x}{\seq\x}{\e}{\Stop}}{\CFilter{\T_\x}{\m}{\seq{\X}}{\seq\UT}{\eff}{\Continue}}}
{\IsMType{\TEnv}{\NT_\x}{\m}{\MkTE{\mgc}{\Ext{\seq\X}{\seq\UT}}{\seq\T}{\_}{\_}}\\
  \SubType{\TEnv}{\T''}{\T'}  
}
\end{array}
\end{math}
\end{small}
\caption{Typing rules for values and expressions}\label{fig:typing-exp}
\end{figure}

They rely on the following auxiliary notations:
\begin{enumerate}
\item $\EffRed{\TEnv}{\eff}{\eff'}$ meaning that the effect $\eff$ can be simplified to $\eff'$
\item  $\FilterF{\handler}$ the \emph{filter function} associated to the handler $\handler$
\item $\typeof{\TEnv}{\T}{\sig}$ meaning that we can safely extract a signature $\sig$ from a type; the notation $\IsMType{\TEnv}{\T}{\m}{\MkT{\kind}{\MT}}$ is an abbreviation for 
$\typeof{\TEnv}{\T}{\sig}$ and $\sig(\m)=\MkT{\kind}{\MT}$
\item $\SubType{\TEnv}{\T}{\T'}$ the subtyping relation
\end{enumerate}
We illustrate the typing rules and notations (1) and (2), whereas the formal definitions of (3) and (4), which are almost standard, are given in the extended version \cite{DagninoGZ25bis}. 
The reader should only know that (3) models extracting the structural type information; in this phase, constraints about no conflicting method definitions and safe overriding are checked. Then, typing rules model another phase where code (method bodies) is typechecked against this type information.\footnote{Differently from other type systems for Java-like languages \cite{IgarashiPW99}, here the two phases cannot be sequenced, since we still need signature extraction for the types introduced ``on the fly'' by objects. }

Rule \refToRule{t-var} is straightforward.  
An object is well-typed, rule \refToRule{t-obj}, if a signature can be safely extracted from its type. This essentially means that there are no conflicts among the parent types, and they are safely overriden by the method declarations in the object (first and second side conditions).  Moreover, an object cannot declare magic methods (third side condition), and should provide an implementation for all its (either inherited or declared) methods (last side condition). 
Finally, in the premise, (bodies of) declared methods should be well-typed with respect to an enclosing type which is the object type. 

In rule \refToRule{t-invk}, the type-and-effect of the invoked method is found in the (signature extracted from) the receiver's type,  as expressed by the first side condition.  In the second side condition, the type annotations in the call should (recursively) satisfy the bounds for the corresponding type variables, and, in the third side condition, the argument types  should be subtypes of the corresponding parameter types where type variables have been replaced by the corresponding type annotations. The type assigned to the call is the return type of the method, with the same replacement. The method effect is instantiated by replacing the type variables with the type annotations. However, the effect assigned to the call is obtained by a further \emph{simplification step} (last side condition). 

The formal definition of simplification is given in \cref{fig:simplify}, where $\EffRed{\TEnv}{\eff}{\eff'}$ means that the (not necessarily simplified) effect $\eff$ is simplified to $\eff'$, which only contains call-effects which are either magic or variable, as described before.

\begin{figure}[t]
\begin{math}
\begin{array}{c}
\NamedRule{ empty}{}{\EffRed{\TEnv}{\eZero}{\eZero}}{}
\BigSpace
\NamedRule{top}{}{\EffRed{\TEnv}{\eTop}{\eTop}}{}
\BigSpace
\NamedRule{var}{}{\EffRed{\TEnv}{\eCall{\X}{\m}{\seq{\T}}}{\eCall{\X}{\m}{\seq{\T}}}}{\X\in\dom(\TEnv)\\ }
\\[4ex]
\NamedRule{mgc}{}{\EffRed{\TEnv}{\eCall{\T}{\m}{\seq{\T}}}{\eCall{\T}{\m}{\seq{\T}}}}{
\IsMType{\TEnv}{\T}{\m}{\MkTE{\mgc}{\Ext{\seq\X}{\_}}{\_}{\_}{\eff}}\\
}
\\[5ex]
\NamedRule{simplify-non-mgc}{ \EffRed{\TEnv}{\Subst{\eff } {\seq\T}{\seq\X}}{\eff'} }
{ \EffRed{\TEnv}{\eCall{\T}{\m}{\seq{\T}}}{\eff'} } 
{\IsMType{\TEnv}{\T}{\m}{\MkTE{\kind}{\Ext{\seq\X}{\_}}{\_}{\_}{\eff}}\\
 \kind\neq\mgc
}
\\[5ex]
\NamedRule{simplify-union}{\EffRed{\TEnv}{\eff_i}{\eff'_i}\ i\in 1..2}{\EffRed{\TEnv}{\EComp{\eff_1}{\eff_2}}{\EComp{\eff'_1}{\eff'_2}}}{}

\end{array}
\end{math}
\caption{Simplification of effects}
\label{fig:simplify}
\end{figure}

We assume that effect  annotations are written by the programmer (or preliminarily reduced) in a simplified shape. Non-simplified  effects  can, however, appear during typechecking, when a method declared with an effect $\eff$ which is generic, that is, depending on the method's type variables $\seq\X$,  is invoked with actual type arguments\footnote{This happens in rule \refToRule{t-invk} in \cref{fig:typing-exp}. }
Indeed, in this case, instantiating the type variables,  in some variable call-effect $\eCall{\X}{\m}{\seq\T}$ in $\eff$, $\X$ can become an object type providing a (non-magic) method $\m$, allowing to simplify to the corresponding effect of $\m$ in the object type. This allows  us to give  a more refined approximations of the effects raised at runtime, as  illustrated by \cref{ex:simpl} above, where the non-simplified effect \lstinline{MyTEType.then $\vee$ MyTEType.else} is simplified to \lstinline{$\eZero$ $\vee$ MyException.throw[Nat]}. 
 To ensure decidability of typechecking, an issue we do not deal with in this paper, termination of effect simplification should be enforced by some standard technique, essentially by forbidding (mutual) recursion in effect annotations.

In rule \refToRule{t-ret}, a computation which is the embedding of a value has no effects, and, in \refToRule{t-do}, a sequential composition of two computations has the union of the two effects.

Typing rules for try-blocks rely on \emph{filter functions} associated to handlers, which describe how they
 transform effects, by essentially replacing calls of magic methods matching some clause with the effect of the clause expression, as formally defined in \cref{fig:filters}.  Filter functions allow to typecheck try-blocks in a very precise way, since effects of the expression in a clause are only added if the clause could be possibly applied. This generalizes to arbitrary effects what happens for Java catch clauses; however,  in Java this is part of the analysis of unreachable code, whereas here it is a feature of the type system. Filters are not a novelty of this type-and-effect system, since they were firstly used in \cite{DagninoGZ25} for a functional calculus; however, it is worthwhile to describe them in detail since they are a  new  feature,  and the general idea is applied here to rather different effects, leading to another formal definition.

\begin{figure}[t]
\begin{math}
\begin{grammatica}
\produzione{\handler}{\Handler{\seq\cc}{\x}{\e}}{handler}\\
\produzione{\hfilter}{\HFilter{\seq\cfilter}{\eff}}{filter}\\
\produzione{\cc}{\CC{\NT}{\m}{\seq\X}{\x}{\seq\x}{\e}{\mode}}{catch clause}\\
\produzione{\cfilter}{\CFilter{\NT}{\m}{\seq{\X}}{\seq\UT}{\eff}{\mode}}{clause filter}\\[1ex]
\end{grammatica}
\end{math}
\hrule
\begin{math}
\begin{array}{ll}
\\[-1.5ex]
 \FilterF{\hfilter}\   \mbox{defined by:}\\
\FilterFun{\eff}{\hfilter}=\EComp{\FilterFun{\eff}{\seq\cfilter}}{\eff'}\ \mbox{if}\ \hfilter=\HFilter{\seq\cfilter}{\eff'}\BigSpace
\\[1ex]
 \FilterF{\seq\cfilter}\  \mbox{defined by:}\\
\FilterFun{\eZero}{\seq\cfilter}=\eZero\BigSpace
\FilterFun{\eTop}{\seq\cfilter}=\eTop\BigSpace
\FilterFun{\EComp{\eff_1}{\eff_2}}{\seq\cfilter}=\EComp{\FilterFun{\eff_1}{\seq\cfilter}}{\FilterFun{\eff_2}{\seq\cfilter}}\\
\eff=\eCall{\T}{\m}{\seq{\T}}\BigSpace
\FilterFun{\eff}{\cfilter\seq\cfilter}=
\begin{cases}
\FilterFun{\eff}{\cfilter}& \mbox{if}\ \FilterFun{\eff}{\cfilter}\ \mbox{defined}\\
\FilterFun{\eff}{\seq\cfilter}&\mbox{otherwise}
\end{cases}
\BigSpace
\FilterFun{\eff}{\epsilon}=\eff
\\[2ex]
\cfilter=\CFilter{\NT}{\m}{\seq\X}{\_}{\eff}{\mode}\BigSpace
\FilterFun{\eCall{\T}{\m}{\seq{\T}}}{\cfilter}=
\begin{cases}
\Subst{\eff}{\seq\T}{\seq\X}&\mbox{if}\   \SubType{\TEnv}{\T}{\NT}
\\
\mbox{undefined}&\mbox{otherwise}
\end{cases}
\end{array}
\end{math}
\caption{Filters}\label{fig:filters}
\end{figure}

As shown in \cref{fig:filters}, filters are the type information which can be extracted from a handler, consisting of a sequence of clause filters and a final effect. The filter function $\FilterF{\hfilter}$ associated to $\hfilter$  transforms an effect by first applying the sequence of clause filters, and then adding the final effect. The transformation applying a sequence of clause filters is defined inductively. 
The significant case is a call-effect, which is either transformed by a (first) matching clause filter,  or remains unaffected.  A clause filter matches a call-effect with the same method name and a subtype of the nominal type; the call-effect is replaced by the effect of the clause filter, where type variables have been substituted \mbox{by the types in the call-effect. }

In rule \refToRule{t-try}, in order to typecheck a try-block, first we get the type and effect of the enclosed expression. 
 This type is then used to typecheck the handler,  as type of the parameter of the final expression, see rule \refToRule{t-handler}. 
 By typechecking the handler we get a type, being a supertype of the final expression,  which will be the type of the whole expression. Moreover, we extract from the handler a filter, which is used to transform the effect of the enclosed expression, getting the resulting effect of the whole expression.

In rule \refToRule{t-handler}, as said above, the type on the left of the judgment is used as type of the parameter of the final expression, 
 required to be a subtype of that of the handler. 
This latter type is also needed to typecheck  $\Stop$-clauses, see below.  
The filter extracted from the handler consists in a clause filter for each clause, and the effect of the final expression.

The filter extracted from a clause keeps the first three components (nominal type, method name, and type variables), and adds the effect obtained typechecking the clause expression, as shown in rules \refToRule{t-continue} and \refToRule{t-stop}.  A $\Continue$-clause is meant to provide alternative code to be executed before the final expression, hence the type of the clause expression should be (a subtype of) the return type of the operation. 
In a $\Stop$-clause, instead, the result of the clause expression becomes that of the whole expression with handler, hence the type of the former should be (a subtype of) the latter. 

Referring to \cref{ex:exception}, note that catching an exception with a $\Continue$-clause would be ill-typed. Indeed, the type of the clause expression should be (a subtype of) the return type of \lstinline{throw}, which is a type variable $\X$. Since no value has type $\X$, no value could be returned\footnote{Hence, the clause expression could only be another \lstinline{throw} or a diverging expression.}, as already noted in \cite{PlotkinP03}.

In \cref{fig:typing-dec} we show the typing rules for method and type declarations, which are mainly straightforward.
\begin{figure}[t]
\begin{math}
\begin{array}{c}
\NamedRule{t-meths}{\IsWFMethod{\TEnv}{\Gamma}{\T}{\md_i}\ \forall i\in 1..n}{\IsWFMethod{\TEnv}{\Gamma}{\T}{\md_1\ldots\md_n}}{}
\\[2ex]
\NamedRule{t-meth}{
  \IsWFExp{\TEnv,\Ext{\seq\X}{\seq{\UT}}}{\Gamma,\TVar{\T_\x}{\x}, \TVars{\T}{\x} } {\e} {\T'} {\eff'}              
}{ \IsWFMethod{\TEnv}{\Gamma}{\T_\x}{\MethDec{\m}{\defn}{\TMeth{\Ext{\seq\X}{\seq\UT}}{\seq\T}{\TEff{\T}{\eff}}}{\Pair{\x\,\seq\x}{\e}}} 
}
{\SubType{\TEnv,\Ext{\seq\X}{\seq{\UT}}}{\TEff{\T'}{\eff'}}{\TEff{\T}{\eff}}
} 
\\[3ex]
\NamedRule{t-prog}{\IsWFNType{\tdec_i}\ \forall i\in 1..n}{\IsWFNType{\tdec_1\ldots\tdec_n}}{}
\\[3ex]
\NamedRule{t-ntype}
{\IsWFMethod{\Ext{\seq\Y}{\seq{\T}}}{\emptyset}{\Gen\tname{\seq\Y}}{\GetDef{\seq\md}}}
{ \IsWFNType{ \TDec{\Gen\tname{\Ext{\seq\Y}{\seq\T}}}{\seq\NT}{\seq{\md}}} } 
{
\typeof{}{\TDec{\Gen\tname{\Ext{\seq\Y}{\seq\T}}}{\seq\NT}{\seq{\md}}}{\_}
}
\end{array}
\end{math}
\caption{Typing rules for method and type declarations}\label{fig:typing-dec}
\end{figure}
The typing judgment for method definitions has shape $\IsWFMethod{\TEnv}{\Gamma}{\T}{\seq\md}$. 
In rule \refToRule{t-meth},  a method definition is well-typed if the body is well-typed with respect to a type environment enriched by the type variables with their bounds, and an environment enriched by the variable denoting the current object with the enclosing type, and parameters with the corresponding parameter types.  The type-and-effect of the body should be a sub-type-and-effect of that declared for the method.  The typing judgment for type declarations has shape $\IsWFNType{\seq\tdec}$. 
In rule \refToRule{t-ntype}, 
 a  type declaration, assumed to have passed the extraction phase (side condition), is well-typed if its defined method declarations, denoted by $\GetDef{\seq\md}$, are well-typed with respect to the type environment consisting of the type variables with their bounds, the empty environment, and the declared type as enclosing type. Here the environment is empty since these method declarations are top-level, whereas in those inside objects, handled in rule \refToRule{t-obj}, there can be variables declared at an outer level.

%% file: soundness.tex

\section{Type-and-effect Soundness}\label{sect:results}
In this section, we express and prove soundness of our type-and-effect system, by applying definitions and results in \cite{DagninoGZ25}.
We focus on explaining the concepts, referring to  \cite{DagninoGZ25} for detailed formal definitions and proofs. 

\smallskip
\noindent\textit{Informal introduction} 
In \cref{sect:sem} we defined, on top of the monadic one-step reduction:
\begin{itemize}
\item a finitary semantics $\fun{\finsem{-}}{\Exp}{\mfun\Res + \{\infty\}}$
\item an infinitary semantics  $\fun{\infsem{-}}{\Exp}{\mfun\Res}$
\end{itemize}
In standard soundness we expect the result, if any, to be in agreement with the expression type. Here, since the expression has also an effect, approximating the computational effects raised by its execution, we expect  the monadic result, if any, to be in agreement with the expression type and effect. 

Let us write $\WTExp{\e}{\T}{\eff}{}$ for $\IsWFExp{\emptyset}{\emptyset}{\e}{\T}{\eff}$, and analogously for $\WTVal{\ve}{\T}$, and other ground judgments. Moreover, $\WTRes{\res}{\T}$ only holds if $\res=\ve$ and $\WTVal{\ve}{\T}$ holds.
 Note that the result $\wrng$ is never well-typed, that is, $\WTRes{\wrng}{\T}$ does not hold for any type $\T$.  

To formally express the above soundness requirement, we need analogous typing judgments\footnote{The effect is written under the turnstile to   emphasize    that this typing judgment is obtained by lifting the non-monadic one, as will be described in the following.}  $\WTMRes{\mres}{\T}{\eff}$, one for each type and effect, on monadic results.  
Assuming to have such judgments, we can express the soundness results as follows.

\begin{theorem}[Finitary type-and-effect soundness]\label{theo:mnd-sound-fin}
$\WTExp{\e}{\T}{\eff}$ and $\finsem{\e} = \mres$ imply $\WTMRes{\mres}{\T}{\eff}$.
\end{theorem}

\begin{theorem}[Infinitary type-and-effect soundness]\label{theo:mnd-sound-inf}
$\WTExp{\e}{\T}{\eff}$  implies $\WTMRes{\infsem{\e}}{\T}{\eff}$. 
\end{theorem}
 Note that  finitary soundness is vacuous when the finitary semantics of an expression is $\infty$. On the other hand, the infinitary semantics is always a monadic result.

Before explaining how the monadic typing judgments can be derived from the non-monadic ones, we show some examples.
\begin{example}\label{ex:exc-mlift}
Consider the exception monad as in \cref{ex:exception}. We have $\ExceptFun\Res= \Res+\ExSet$, that is, monadic results are either values, or $\wrng$, or exceptions. We expect the typing judgment $\WTMRes{\mres}{\T}{\eff}$ to be defined as follows:
\begin{quoting}
\begin{math}
\NamedRule{t-val}{\WTVal{\ve}{\T}}{\WTMRes{\ve}{\T}{\eff}}{}
\BigSpace
\NamedRule{t-exc}{}{\WTMRes{\exc}{\T}{\eff}}{\exc\in\effsem{\eff}}
\end{math}
\end{quoting}
That is, a monadic result is well-typed with a given type-and-effect if it is either a well-typed value with the type, or an exception belonging to (the set denoted by) the effect. Note that $\wrng$ is ill-typed. 
The set of exceptions $\effsem{\eff}$ represented by $\eff$ is defined by:
\begin{quoting}
$\effsem{\eZero}=\emptyset$\HugeSpace
$\effsem{\eTop}=\ExSet$\HugeSpace
$\effsem{\EComp{\eff}{\eff'}}=\effsem{\eff}\cup \effsem{\eff'}$\\
$\effsem{\tname\texttt{.throw}} = \{\exc\mid\exc = \getExc{\ve}, \instanceof{\ve}{\tname}\}$
\end{quoting}
That is, call-effects of the \lstinline{throw} method in $\tname$ denote the set of the exceptions (represented by objects)  of subtypes of $\tname$, and the other operators have the obvious set-theoretic meaning. 
 \end{example}
 
 \begin{example}\label{ex:pow-mlift}
 Consider the monad of non-determinism as in \cref{ex:pow}. We have $\ListFun \Res$ the set of (possibly infinite) lists of results (values or $\wrng$).
We show two different ways to define the typing judgment for monadic results, denoted $\WTMResQ{\mres}{\T}{\eff}{\quantifier}$ for $\produzioneinline{\quantifier}{\forall\mid\exists}$.
\begin{quoting}
\begin{math}
\NamedRule{t-no-res}{}{\WTMResQ{\epsilon}{\T}{\eff}{\quantifier}}{}
\BigSpace
\NamedRule{t-det}{\WTVal{\ve}{\T}}{\WTMResQ{[\ve]}{\T}{\eff}{\quantifier}}{\effsem\eff = 0}\\
\NamedRule{t-$\forall$}{\WTRes{\res}{\T}\ \forall\res\in\mres\  }{\WTMResQ{\mres}{\T}{\eff}{\forall}}{\effsem{\eff}=1 }\BigSpace
\NamedRule{t-$\exists$}{\WTRes{\res}{\T}}{\WTMResQ{\mres}{\T}{\eff}{\exists}}{\effsem{\eff}=1\\ \res\in\mres}
\end{math}
\end{quoting}
where $\effsem{\eff}$ is $1$ if non-determinism is allowed, $0$ otherwise, as defined below:
\begin{quoting}
$\effsem{\eZero}=0$\HugeSpace
$\effsem{\eTop}=1$\HugeSpace
$\effsem{\EComp{\eff}{\eff'}}=\effsem{\eff}\vee \effsem{\eff'}$\\
$\effsem{\eCall{\texttt{Chooser}}{\texttt{choose}}{\ }}=1$
\end{quoting}
That is, monadic results (representing possible results of a computation) are well-typed with a type effect (denoting) $0$ if they have at most one element, and this element, if any, is well-typed; in other words, the computation is deterministic. 
On the other hand, they are well-typed with a type effect (denoting) $1$ if all the results in the  list  are well-typed values, or there is either no result, or at least one well-typed value, respectively. 
The two interpretations express ``must'' and ``may'' soundness of a non-deterministic computation: with $\forall$, we require every possible result to be well-typed, whereas, with $\exists$, it is enough to have a well-typed result, and the others could be even $\wrng$. 
\end{example}

\smallskip
\noindent\textit{Interpretation of effect types} Monadic typing judgments can be derived from non-monadic ones by providing an \emph{interpretation of effect types} into the considered monad.

Let us denote by $\EffSet$ the set of effects, assumed to be simplified  (\cref{fig:simplify}), ground, and well-formed (\cref{fig:wf-eff}), considered modulo the equivalence\footnote{That is, we consider effects as sets of call-effects, plus the top effect.} induced by the preorder $\subt$ in  \cref{fig:sub-eff}. 
As customary, with a slight abuse of notation, 
we write $\eff$ for the element of $\EffSet$ it represents, i.e., its equivalence class. 
Then, it is easy to see that $\EffSet$ is the carrier of an ordered monoid $\MEff$, where the unit is $\eZero$, the multiplication, which turns out to be idempotent and commutative,  is $\EComp{}{}$, and the order is $\subt$.  

For a set $X$, we denote by $\PW(X)$ the poset of all subsets (a.k.a. predicates or properties) on $X$, ordered by subset inclusion. 
For a function \fun{f}{X}{Y}, we have a monotone function \fun{\PW_f}{\PW(Y)}{\PW(X)}, given by the inverse image: 
for $A\subseteq Y$, $\PW_f(A) = \{ x \in X \mid f(x) \in A \}$.  

\begin{definition}[Interpretation of effect types] \label{def:mlift}
Let $\mnd = \ple{\mfun,\mmul,\mun}$ be a monad.
Then, an \emph{interpretation} of  $\MEff$  in $\mnd$
consists of a family $\mlift$ of monotone functions ${\fun{\mlift[\eff]_X}{\PW(X)}{\PW(\mfun X)}}$, 
\mbox{for every $\eff\in\EffSet$ and set $X$, such that}
\begin{enumerate}[series=enum-mlift] 
\item\label{def:mlift:nat}  $\mlift[\eff]_X(\PW_f(A)) = \PW_{\mfun f}(\mlift[\eff]_Y(A))$, for every $A\subseteq Y$ and function \fun{f}{X}{Y}
\item\label{def:mlift:mon}  $\eff\subt\eff'$ implies $\mlift[\eff]_X(A)\subseteq\mlift[\eff']_X(A)$, for every $A\subseteq X$
\item\label{def:mlift:unit}  $A \subseteq \PW_{\mun_X}(\mlift[\eZero]_X(A))$, for every $A\subseteq X$
\item\label{def:mlift:mul}  $\mlift[\eff]_{\mfun X}(\mlift[\eff']_X(A)) \subseteq \PW_{\mmul_X}(\mlift[\eff\vee\eff']_X(A))$, for every $A\subseteq X$. 
\end{enumerate} 
\end{definition}

 The family $\mlift = (\mlift[\eff])_{\eff\in\EffSet}$ is a family of predicate liftings \cite{Jacobs16} for the monad $\mnd$, indexed by effect types. 
 That is, for each effect type $\eff$, $\mlift[\eff]$ transforms predicates on $X$ into predicates on $\mfun X$, and can be regarded as the semantics of the effect type $\eff$.

\cref{def:mlift:nat} states that $\mlift[\eff]_X$ is natural in $X$. 
This ensures that the semantics of each effect type is independent from the specific set $X$.
\cref{def:mlift:mon} states that $\mlift[\eff]_X$ is monotone with respect to the order on effects, that is, 
computational effects described by $\eff$ are also described by $\eff'$. 
\cref{def:mlift:unit} states that monadic elements in the image of $\mun_X$ contain computational effects described by $\eZero$, that is, no computational effect. 
Finally, in \cref{def:mlift:mul} we consider elements of $\mfun^2\X$ whose computational effects are described by lifting predicates to $\mfun\X$ through $\eff'$,  and then by lifting through $\eff$.
By flattening such elements through $\fun{\mmul_\X}{\mfun^2\X}{\mfun\X}$ we obtain elements whose computational effects are described by the composition $\eff\vee\eff'$. 

In the proof of soundness for our calculus, we only need \cref{def:mlift:unit} and \cref{def:mlift:nat}, instantiated as will be detailed in \cref{lem:mlift}. \cref{def:mlift:mon} and \cref{def:mlift:mul} are only required in the general framework in \cite{DagninoGZ25} to derive soundness from progress and subject reduction.

We show now how the previous examples can be obtained through an appropriate interpretation of effect types. 

In \cref{ex:exc-mlift}, we can take $\mlift[\eff]_X(A) = A + \effsem{\eff}$. That is, the interpretation of an effect type $\eff$ transforms a predicate $A$ on $X$ into a predicate on $\mfun X$ which holds either on elements which satisfy the original predicate, or exceptions in the set denoted by $\eff$.

In \cref{ex:pow-mlift}, we can take the following two interpretations $\AllLift$ and $\ExLift$: 
\begin{quoting}
if $\effsem{\eff}=1$, then 
\begin{quoting}
$\AllLift^\eff_{\!X}(A) = \{ B \in \PowerFun X \mid B \subseteq A \}$ and\\
$\ExLift^\eff_X(A)  = {\{ B \in \PowerFun X \mid B = \emptyset \text{ or }B \cap A \ne \emptyset \}}$
\end{quoting}
 if $\effsem{\eff}=1$, then 
$\AllLift^0_{\!X}(A) = \ExLift^0_X(A) = \{ B \in \PowerFun X  \mid B\subseteq  A \text{ and } \sharp B\leq 1\}$
\end{quoting}
That is, in both cases, the interpretation of $0$  forbids  non-determinism, while 
the interpretation of $1$ requires the predicate $A$ to be always satisfied, according to $\AllLift$, and 
satisfied in at least one case, according to $\ExLift$. 

 \smallskip
\noindent\textit{Proof of type-and-effect soundness} 
Thanks to a general result proved in \cite{DagninoGZ25}, to prove \cref{theo:mnd-sound-fin} and \cref{theo:mnd-sound-inf} we can use a technique similar to that  widely used to prove soundness of a type system with respect to a small-step semantics, that is, it is enough to prove  the  progress and subject reduction properties, stated below.  For  the  infinitary soundness, the  interpretation of effect types has to respect the additional structure of the monad, that is, the least monadic element should be always well-typed, and the monadic typing judgment should be closed with respect to suprema of $\omega$-chains. We refer to \cite{DagninoGZ25} for the formal definition of such requirement, which trivially holds in our case.
 
Let us denote by $\WTMVal{\mve}{\T}{\hat{\eff}}$ and $\WTMExp{\me}{\T}{\eff}{\hat{\eff}}$ the typing judgments on monadic values and  expressions obtained by lifting the non-monadic ones through $\hat{\eff}$.  

\begin{theorem}[Monadic Progress]\label{theo:mnd-progress}
$\WTExp{\e}{\T}{\eff}$ implies either $\e = \Ret{\ve}$ for some $\ve\in\Val$, or 
$\e\red\me$ for some $\me\in\mfun\Exp$.
\end{theorem}

\begin{theorem}[Monadic Subject Reduction] \label{theo:mnd-sr}

$\WTExp{\e}{\T}{\eff}$ and $\e\red\me$ imply $\WTMExp{\me}{\T}{\eff'}{\hat{\eff}}$ for some $\T',\hat{\eff},\eff'$ such that $\SubT{\TEff{\T'}{\EComp{\hat{\eff}}{\eff'}}}{\TEff{\T}{\eff}}$. 

\end{theorem}

Monadic progress is standard: a well-typed expression either is  (the embedding of) a value  or can reduce. 
For monadic subject reduction, if a well-typed expression reduces to a monadic one, then  the monadic expression has a more specific type, and the expression effect is an upper bound  of the computational effects produced by the current reduction step, described by $\hat{\eff}$, composed with those produced by future reductions, described by $\eff'$. 

Remarkably enough, the proofs of monadic progress and subject reduction can be driven by the usual rule-based reasoning, without any need to know about \cref{def:mlift}, thanks to \cref{lem:mlift} below. 
Recall that the interpretation of effects \emph{lifts} predicates to monadic predicates, in our case typing judgments to monadic typing judgments. \cref{lem:mlift} 
shows that it is even possible to lift, in a sense, the type system itself. That is, we can derive a typing rule for each operator on monadic elements used in the semantics.
The proofs in the following only rely on such monadic type system.  

\begin{lemma}\label{lem:mlift} The following rules can be derived.
\begin{quoting}
\begin{math}
\NamedRule{t-unit}{\WTExp{\e}{\T}{\eff}}{\WTMExp{\mun(\e)}{\T}{\eff}{\eZero}}{}\BigSpace
\NamedRule{t-ret-lift}{
  \WTMVal{\mve}{\T}{\hat{\eff}}
  }
{ \WTMExp{\Mmap (\Ret{\ehole})\ \mve}{\T}{\eZero}{\hat{\eff}} }
{ } 
\\[2ex]
\NamedRule{t-do-lift}{
 \WTMExp{\me}{\T}{\eff}{\hat{\eff}} \BigSpace
   \IsWFExp{\emptyset}{\TVar{\T}{\x}}{\e'}{\T'}{\eff'}
     }
{ \WTMExp{\Mmap{(\Do\x{\ehole}{\e')}{\me} }}{\T'}{\EComp{\eff}{\eff'}}{\hat{\eff} }}
{ 
} 
\end{math}
\end{quoting}
\end{lemma}
\begin{proof}
By instantiating \refItem{def:mlift}{unit} with $X=\Exp$ and $A=\{\e\mid\ \WTExp{\e}{\T}{\eff}\}$, we get $A\subseteq\{\e\mid\ \WTMExp{\mun(\e)}{\T}{\eff}{\eZero}\}$, that is, $\WTExp{\e}{\T}{\eff}$ implies $\WTMExp{\mun(\e)}{\T}{\eff}{\eZero}$, \mbox{as expressed by rule \refToRule{t-unit}. }

\newcommand{\Pre}{\textit{Pre}}
\newcommand{\Co}{\textit{Cons}}

We show how to derive \refToRule{t-ret-lift}, \refToRule{t-do-lift} can be derived analogously. Since, by rule \refToRule{t-ret},  $\WTValNarrow{\ve}{\T}$ implies $\WTExpNarrow{\Ret{\ve}}{\T}{\eZero}$, by the monotonicity of $\mlift[\hat{\eff}]$ we have that, set $\Pre=\{\ve\mid\ \WTVal{\ve}{\T}\}$, $\Co=\{\ve\mid\ \WTExp{\Ret{\ve}}{\T}{\eZero}\}$, 
$\mlift[\hat{\eff}]_\Val(\Pre)\subseteq\mlift[\hat{\eff}]_\Val(\Co)$.
Since $\mlift[\hat{\eff}]_\Val(\Pre)=\{\mve\mid\ \WTMVal{\mve}{\T}{\hat{\eff}}\}$, and, by \refItem{def:mlift}{nat}, instantiated with function (context) $\fun{\Ret}{\Val}{\Exp}$ and $A=\{ \e\mid\ \WTExpNarrow{\e}{\T}{\eZero}\}$,
$\mlift[\hat{\eff}]_\Val(\Co)=\{\mve\mid\ \WTMExp{\Mmap (\Ret{\ehole})\ \mve}{\T}{\eZero}{\hat{\eff}}\}$,
we get that $\WTMVal{\mve}{\T}{\hat{\eff}}$ implies $\WTMExp{\Mmap (\Ret{\ehole})\ \mve}{\T}{\eZero}{\hat{\eff}}$, as expressed by \mbox{rule \refToRule{t-ret-lift}. }
\end{proof}

 \cref{lem:mlift} derives a typing rule for each monadic operator used in the semantics. 
In addition, the semantics uses monadic constants, which are the results of some $\mrun\tname\m$ function, depending on (the magic methods declared in) the specific program. 
Hence,  like standard constants, we have to assume that these are well-typed.
To this end,  we require a typing rule for such constants, as given below. Essentially, since the (canonical) effect of a magic method $\m$ declared in $\tname$ should be an approximation of the computational effects raised by a call, formalized by the monadic value obtained  as result of $\mrun\tname\m$, then such monadic value should be well-typed with respect to the type-and-effect of the method. 
\begin{small}
\begin{quote}
\begin{math}
\NamedRule{t-run}{\WTVal{\ve_0}{\T_0}\BigSpace \WTVal{\seq\ve}{\seq{\T'}}}{\WTMVal{\mrun\tname\m(\seq\ve)}{\Subst{\Subst{\T}{\seq{\T_\Y}}{\seq\Y}}{\seq{\T_\X}}{\seq\X}}{\eCall{\Gen{\tname}{\seq{\T_\Y}}}{\m}{\seq{\T_\X}}}}
{\IsMType{}{\TDec{\Gen\tname{\Ext{\seq\Y}{\seq{\UT_\Y}}}}{\_}{\_}}{\m}{\MkT{\mgc}{\MT}}\\
\MT=\TMethEff{\Ext{\seq\X}{\seq{\UT_\X}}}{\seq\T}{\T}{\eCall{\Gen{\tname}{\seq\Y}}{\m}{\seq\X}}\\
\SubT{\T_0}{\Gen{\tname}{\seq{\T_\Y}}}\\
\SubT{\seq{\T_\Y}}{\Subst{\seq{\UT_\Y}}{\seq{\T_\Y}}{\seq\Y}}\\
\SubT{\seq{\T'}}{\Subst{\Subst{\seq\T}{\seq{\T_\Y}}{\seq\Y}}{\seq{\T_\X}}{\seq\X}}\\
\SubT{\seq{\T_\X}}{\Subst{\Subst{\seq{\UT_\X}}{\seq{\T_\Y}}{\seq\Y}}{\seq{\T_\X}}{\seq\X}}
}
\end{math}
\end{quote}
\end{small}
More in detail, given a magic method $\m$ declared in $\tname$ (first and second side conditions, where the effect is the canonical one), the result of a $\mrun{\tname}{\m}$ function is well-typed if the first argument has a subtype of one obtained by instantiating the type variables in $\tname$, and each other argument has a subtype of that obtained by  instantiating the type variables in $\tname$ and those in $\m$  with  the corresponding parameter type.
The type and the effect of the result are obtained by an analogous instantiation of variables  of  the return type-and-effect.

 We provide below the proof of monadic subject reduction (\cref{theo:mnd-sr}), noteworthy since it relies on the monadic typing rules, that is, those derived in \cref{lem:mlift} and rule \refToRule{t-run}. 
The proofs of monadic progress  and subject reduction for the pure relation (\cref{lem:subject-reduction})  are in the Appendix, together with the standard inversion and substitution lemmas.

\begin{lemma}[Subject Reduction]\label{lem:subject-reduction}
If $\WTExpNarrow{\e}{\T}{\eff}$ and $\e\purered \e'$, then $\WTExpNarrow{\e'}{\T'}{\eff'}$
and  $\SubTNarrow{\TEff{\T'}{\eff'}}{\TEff\T\eff}$.
\end{lemma}

\begin{proofOf}{theo:mnd-sr}
Assuming $\WTExpNarrow{\e}{\T}{\eff}$ and $\e\red\me$, we have to prove that $\WTMExp{\me}{\T}{\eff'}{\hat{\eff}}$ for some $\T',\hat{\eff},\eff'$ such that $\SubT{\TEff{\T'}{\EComp{\hat{\eff}}{\eff'}}}{\TEff{\T}{\eff}}$. 
By induction on the reduction rules of \cref{fig:monadic-red}.
\begin{description}
\item [\refToRule{pure}] $\me$ is $\mun(\e')$ and $\e\purered \e'$. From $\WTExpNarrow{\e}{\T}{\eff}$ and  \cref{lem:subject-reduction} we get 
$\WTExpNarrow{\e'}{\T'}{\eff'}$ with $\SubT{\TEff{\T'}{\eff'}}{\TEff{\T}{\eff}}$.
From \cref{lem:mlift} by rule \refToRule{t-unit}
we derive   
$\WTMExp{\mun(\e')}{\T'}{\eff'}{\eZero}$, and $\SubT{\TEff{\T'}{\EComp{\eZero}\eff}}{\TEff{\T}{\eff}}$.

\item [\refToRule{mgc}] 
$\e$ is $\MCall{\ve_0}{\m}{\seq{\T_\X}}{\seq\ve}$ and $\me=\Mmap{(\Ret{\ehole})}{\mrun\tname\m(\ve,\seq\ve)}$ and $\mbody(\ve_0,\m)=\Pair{\mgc}{\tname}$. By  inversion  we get $\WTValNarrow{\ve_0}{\T_0}$ and $\WTValNarrow{\seq\ve}{\seq\T}$, and, from the definition of $\mbody$, $\mkind(\T_0, \m)=\mgc$. Hence, again  by inversion,  
     \begin{itemize}
\item
${\IsMType{}{\tdec}{\m}{\MkTE{\mgc}{\Ext{\seq\X}{\seq{\UT_\X}}}{\seq{\T''}}{\T''}{\eCall{\Gen{\tname}{\seq\Y}}{\m}{\seq\X}}}}$ for some $\tdec=\TDec{\Gen\tname{\Ext{\seq\Y}{\seq{\UT_\Y}}}}{\_}{\_}$ 
\item there is $\seq{\T_\Y}$ such that
 $\SubT{\T_0}{\Gen{\tname}{\seq{\T_\Y}}}$
  and $\SubT{\seq{\T_\Y}}{\Subst{\UT_\Y}{\seq{\T_\Y}}{\seq\Y}}$
\item $\SubT{\seq{\T_\X}}{\Subst{\Subst{\seq{\UT_\X}}{\seq{\T_\Y}}{\seq\Y}}{\seq{\T_\X}}{\seq\X}}$ and 
     $\SubT{\seq{\T'}}{\Subst{\Subst{\seq{\T''}}{\seq{\T_\Y}}{\seq\Y}}{\seq{\T_\X}}{\seq\X}}$ and $\T=\Subst{\Subst{\T''}{\seq{\T_\Y}}{\seq\Y}}{\seq{\T_\X}}{\seq\X}$ and $\eff=\eCall{\Gen{\tname}{\seq{\T_\Y}}}{\m}{\seq{\T_\X}}$. 
 \end{itemize}
Since the premises of  rule \refToRule{t-run} are satisfied, we get
$\WTMVal{\mrun\tname\m(\ve,\seq\ve)}{\T}{\eff}$.
From \cref{lem:mlift} by rule \refToRule{t-ret-lift} we get $\WTMExp{\Mmap (\Ret{\ehole})\ \mrun\tname\m(\ve,\seq\ve)}{\T}{\eZero}{\eff}$, and 
$\SubT{\TEff{\T}{\EComp{\eff}{\eZero}}}{\TEff{\T}{\eff}}$.

\item [\refToRule{ret}] $\e$ is $\Do{\x}{\Ret\ve}{\e'}$ and $\me=\mun(\Subst{\e'}\ve\x)$.  By inversion   
$\WTValNarrow{\ve}{\T_1}$ and 
$\IsWFExpNarrow{\emptyset}{\TVarNarrow{\T_1}{\x}}{\e'}{\T}{\eff}$.
Therefore,  by substitution,   we get $\WTExpNarrow{\Subst{\e'}{\ve}{\x}}{\T}{\eff}$. 
From \cref{lem:mlift} by \refToRule{t-unit} 
we derive   
$\WTMExp{\mun(\Subst{\e'}\ve\x)}{\T}{\eff}{\eZero}$,  and $\SubT{\TEff{\T}{\EComp{\eZero}{\eff}}}{\TEff{\T}{\eff}}$.

\item [\refToRule{do}]  
$\e=\Do\x{\e_1}{\e_2}$ and $\me=\Mmap{(\Do\x{\ehole}{\e_2})}{\me_1} $ and $\e_1\red\me_1$. 
 By inversion $\eff'=\EComp{\eff_1}{\eff_2}$ and 
$\WTExpNarrow{\e_1}{\T_1}{\eff_1}$
and $\IsWFExp{\emptyset}{\TVar{\T_1}{\x}}{\e_2}{\T}{\eff_2}$. 
By induction hypothesis $\WTMExp{\me_1}{\T'_1}{\eff'_1}{\hat{\eff}_1}$ 
for some $\hat{\eff}_1$ and $\eff'_1$ with 
$\SubT{\TEff{\T'_1}{\EComp{\hat{\eff_1}}{\eff'_1}}}{\TEff{\T_1}{\eff_1}}$. 
From \cref{lem:mlift} by rule \refToRule{t-do-lift} we get
$\WTMExp{\Mmap{(\Do\x{\ehole}{\e')}{\me} }}{\T}{\EComp{\eff'_1}{\eff_2}}{\hat{\eff}_1 }$U. Therefore, we get the thesis noting that 
$ \SubT{\TEff{\T}{\EComp{\hat{\eff}_1}{\EComp{\eff'_1}{\eff_2}}}}{\TEff{\T}{\EComp{\eff_1}{\eff_2}}}\subt \TEff{\T}{\eff}$.  
\end{description}
\end{proofOf}

%% file: related.tex

\section{Related Work}\label{sect:related}
\smallskip
\noindent\textit{Design of the pure language} 
As mentioned, the pure calculus we choose as basis for the effectful extension is inspired by recent work \cite{WangZOS16,ServettoZ21,WebsterSD24} proposing an object-oriented paradigm based on interfaces with default methods and anonymous classes, rather than class instances with fields. In particular, classless objects were also adopted in \emph{interface-based programming} \cite{WangZOS16}, enabling multiple inheritance by decoupling state and behaviour. Then, in \emph{$\lambda$-based object-oriented programming} \cite{ServettoZ21}, objects are classless and stateless, but limited to be lambdas implementing a single abstract method. In the Fearless language \cite{WebsterSD24}, as in our calculus, this approach is generalized to arbitrary Java anonymous classes with no fields. 

\smallskip
\noindent\textit{Design of the effectful language} 
In our effectful language,  constructs to raise effects, similarly to generic effects  \cite{PlotkinP03},  are just method calls of a special kind, called \emph{magic}.
As mentioned, this terminology is taken from Fearless \cite{WebsterSD24}, and also informally used in Java reflection and Python. 
Concerning constructs for handling effects, handlers of algebraic effects \cite{PlotkinPretnar09,PlotkinPretnar13,KammarLO13,HillerstromL18} are an extremely powerful programming abstraction, describing the semantics of algebraic operations in the language itself, thus enabling the simulation of several effectful programs, such as stream redirection or cooperative concurrency \cite{PlotkinPretnar13,KammarLO13,BauerP15,Pretnar15}. 
The code executed when a call is caught can resume the original computation, using a form of continuation-passing style. 
Differently from exception handling, all occurrences of an operation within the scope of a handler are handled. 
Other forms of handlers have been considered, e.g., shallow handlers  \cite{KammarLO13,HillerstromL18}, where only the first call is handled. 
The handling mechanism in our language is meant to provide a good compromise between simplicity and expressivity: notably, it does not explicitly handle continuations, as in handlers of algebraic effects, hence we cannot resume the original computation more than once; yet we can express a variety of \mbox{handling mechanisms. }

\smallskip
\noindent\textit{Type-and-effect system} 
Type-and-effect systems, or simply effect systems \cite{Wadler98,NielsonN99,WadlerT03,MarinoM09,Katsumata14}, 
are a popular way of statically controlling computational effects. Many of them have been designed for specific notions of computational effect and also implemented in mainstream programming languages, the most  well-known  being the  mechanism of Java exceptions. 
Katsumata \cite{Katsumata14} recognized that effect systems share a common algebraic structure, notably they form an ordered monoid, and gave them denotational semantics through parametric monads, using a structure equivalent to our notion of interpretation. 

Our effects are essentially sets,  like  in formal models of  Java  exceptions \cite{AnconaLZ01}, 
 tracking which magic calls an expression is allowed to do.  
Differently from effect systems for algebraic effects \cite{BauerP14,BauerP15,Pretnar15}, also implemented in Eff \cite{EffLang},  and inspired by  \cite{GarianoNS19}, 
we enrich the effects with type information, not using just the name of magic methods. 
 By using this type information, which is polymorphic, we also allow effects to depend on the effect of other methods, smoothly enabling a form of effect polymorphism.
Effect systems supporting effect polymorphism have been widely studied for algebraic effects 
\cite{Leijen17,HillerstromLAS17,LindleyMM17,BiernackiPPS19,BrachthauserSO20jfp,BrachthauserSO20oopsla},
especially in the form of row polymorphism. 
A precise comparison with these systems is left \mbox{for future work. }

\smallskip
\noindent\textit{Semantics for effectful languages} 
Algebraic effects are typically treated as uninterpreted operations, that is, 
the evaluation process just builds a tree  of operation calls \cite{JohannSV10,SimpsonV20,Pretnar15}. 
Monadic operational semantics for $\lambda$-calculi with algebraic effects are also considered, mainly in the form of a monadic definitional interpreter (see, e.g., \cite{LiangHJ95,DaraisLNH17,Gavazzo18,DalLagoG22,DagninoG24}). 
That is, they directly define a function from expressions to monadic values, which essentially corresponds to our infinitary semantics. 
Small-step approaches are also considered by \cite{GavazzoF21,GavazzoTV24}. 


An alternative way of interpreting algebraic operations is by means of runners, a.k.a. comodels 
\cite{PowerS04,PlotkinP08,Uustalu15,AhmanB20}. Roughly, runners describe how operations are executed by the system, that is, 
how they transform the environment where  they are  run. This essentially amounts to giving an interpretation of operations in the state monad. 
More general runners, where the system is modelled in a more expressive way, are considered by \cite{AhmanB20}, where the state monad is combined with errors and system primitive operations. 

The motivation to have methods interpreted in a monad, rather than implemented in code, is the same nicely illustrated for runners by \cite{AhmanB20}. They observe that, in languages modeling computational effects, they are represented in the language itself. 
However,  some kinds of effects, such as input/output, even though represented in the language as algebraic operations (as
in Eff) or a monad (as in Haskell), are handled  by an operating system
functionality. Hence, semantics lacks modularity,  which can be recovered by introducing runners. 
 The same modularity is achieved in our calculus by providing the interpretations of magic methods as an independent ingredient of the semantics, which can be regarded as the operational counterpart of (effectful) runners. 
Differently from runners of \cite{AhmanB20}, our magic methods can be interpreted in an arbitrary monad. 
Moreover, our interpretation is a parameter of the operational semantics, while runners of \cite{AhmanB20} are described in the language itself through a mechanism similar to effect handlers, with a limited use of continuations.

%% file: conclu.tex

\section{Conclusion}\label{sect:conclu}
This paper provides a positive answer to the question:
\begin{quoting}
\emph{is it possible to smoothly incorporate in the object-oriented paradigm constructs to raise, compose, and handle effects in an arbitrary monad?}
\end{quoting}
In particular, we consider a pure OO paradigm based on interfaces with default methods and anonymous classes, exemplified by a higher-order calculus subsuming both FJ and $\lambda$-calculus. 
The effectful extension is designed to be familiar for OO programmers, with magic methods to raise effects, and an expressive, yet simple, handling mechanism.  
Effect types are designed to be familiar, being sets of call-effects generalizing exception names;  the key novel feature is that type effects containing type variables can be refined in concrete calls, through a non-trivial process called \emph{simplification}.  
The semantics and the soundness proof  are given by a  novel approach \cite{DagninoGZ25}, providing a significant application for a language with a complex type system (both nominal and structural types, subtyping and generic types).
Instantiating such  an  approach, we derive from the type system a \emph{monadic type system}, allowing to carry the proof of subject reduction by the usual rule-based reasoning.

The integration of effectful features turned out to be, in fact, even smoother than in the functional case, fitting
naturally in the OO paradigm, where computations happen in the context of a given program (class table in Java). That is, magic methods are, from the point of view of the programmer, as other methods offered by the environment (system, libraries), with the only difference that their calls can be handled.
Another advantage is that, since methods ``belong to a type'', atomic effects (call-effects in our paper) naturally include a type, rather than being just operation names as in previous literature handling the functional case. Thus, effect polymorphism is obtained for free from the fact that types can be type variables, without any need of ad-hoc effect variables.

Concerning the specific calculus, we could have likely used FJ (in its version with generic types) as well. However we have, besides nominal, structural and intersection types, tht is, a paradigm which is partly object-based rather than class-based.  This leads to more flexibility, hence the tracking of effects is more flexible as well (can be done on an object's basis). In particular, FJ is not higher-order, meaning that it does not support functions being first-class values as customary in functional languages. Formally, in the translation of $\lambda$-calculus shown in \cref{rem:lambda-encoding}, any arbitrary function is encoded by an object literal; without object literals, only a finite collection of functions could be encoded, each one by a class, in a program.

 This is a foundational paper, hence we focus on the introduction of design features, formalisation, and results through a toy language. Of course, any step towards an implementation and/or the application to a realistic language would be an important contribution, as it would be to mechanize some proofs. Other directions for future work are outlined  below. 

We plan to analyze better the handling mechanism proposed in the paper, and its relation with other approaches. 
The type-and-effect system we design is, as said above, novel and expressive; however, it is very simple in the sense that effect types are essentially sets. Hence, it is not adequate to approximate computational effects where the order matters, such as, e.g., sequences of write operations. A type-and-effect system with a non-commutative operator  on effects is shown in \cite{DagninoGZ25}; however, effects there are purely semantics, and a syntactic representation should be investigated in order to apply the approach to real languages.  

Finally, the interpretation of magic methods could return monadic expressions, rather than values. 
This would enable a more interactive behaviour with the system; for instance, the semantics of a magic method, instead of returning an unrecoverable error,  could  return a call to the method \lstinline{throw} of an exception, which then could be handled by the program.


%
%

%% file: appendix.tex

\section{Auxiliary Definitions for the Type System}\label{sec:aux-typing}

In \cref{fig:extract-sig} we model the extraction of structural type information (that is, signatures) from a program. 

\begin{figure}[ht]
\begin{math}
\begin{array}{l}
\NamedRule{obj-sig}{\typeof{\TEnv}{\seq\NT}{\sig}}{\typeof{\TEnv}{\TObj{\seq\NT}{\sig'}}{\override{\sig}{\sig'}}}{
  }
\BigSpace
\NamedRule{ntypes-sig}{\typeof{\TEnv}{\NT_i}{\sig_i}\ \forall i\in 1..n}{\typeof{\TEnv}{ \NT_1\ldots\NT_n}{ \bigSigPlus_{i\in 1..n}\sig_i}}{}
\\[4ex]

\NamedRule{ntype-sig}{\typeof{}{\prog(\tname)}{\sig}}
{\typeof{\TEnv}{ \Gen\tname{\seq\T}}{\Subst{\sig}{\seq{\T}}{\seq{\Y}}}}
{\prog(\tname)=\TDec{\Gen\tname{\Ext{\seq\Y}{\_}}}{\_}{\_} \\
 \IsWFType{\TEnv}{\Gen\tname{\seq\T}}  \\
 }
 
 \\[4ex]
 \NamedRule{tdec-sig}{\typeof{\Ext{\seq\X}{\seq{\T}}}{\seq{\NT}}{\sig}\Space \typeof{\Ext{\seq\X}{\seq{\T}}}{\seq\md}{\sig'}}{\typeof{}{ \TDec{\Gen\tname{\Ext{\seq\X}{\seq\T}}}{\seq\NT}{\seq{\md}}}{ \override{\sig}{\sig'}}}{\IsWFType{\Ext{\seq\X}{\seq{\T}}}{\seq{\T}}  }
\\[3ex]
\NamedRule{meths-sig}{\typeof{\TEnv}{\md_i}{\sig_i}\ \forall i\in 1..n}{\typeof{\TEnv}{\md_1\ldots\md_n}{ \bigSigPlus_{i\in 1..n}\sig_i}}{}
\\[3ex]
\NamedRule{def-meth-sig}{}{\typeof{\TEnv}{ \MethDec{\m}{\defn}{\MT}{\_}}{ \MSig{\defn}{\m}{\MT }}}
{\IsWFType{\TEnv}{\MT}
}
 \\[3ex] 
\NamedRule{abs-mgc-meth-sig}{}{\typeof{\TEnv}{\MethDec{\m}{\kind}{\MT}{}}{ \MSig{\kind}{\m}{\MT }}}
{\kind\neq\defn\\
\IsWFType{\TEnv}{\MT}
}

\end{array}
\end{math}
\caption{Extraction of type-and-effect information (signatures)}
\label{fig:extract-sig}
\end{figure}

As mentioned, in this phase constraints about no conflicting method definitions and safe overriding are checked. This is achieved, respectively, by the \emph{symmetric} and \emph{right-preferential sum} operations on signatures defined below.
\begin{description}
\item[Symmetric sum]
The signature $\sig\sigPlus\sig'$ is  defined when the following condition holds:
\begin{quoting}
$\sig(\m)=\MkT{\kind}{\MT}$ and $\sig'(\m)=\MkT{\kind'}{\MT'}$  then $\MT=\MT'$, and $\kind\neq\mgc$, $\kind'\neq\mgc$.
 \end{quoting}
In this case, it is defined as follows: 
\begin{quoting}
\begin{math}
(\sig\sigPlus\sig')(\m)=
\begin{cases}
\sig(\m)&\mbox{if}\ \m\in\dom(\sig){\setminus}\dom(\sig')\\
\sig'(\m)&\mbox{if}\ \m\in\dom(\sig'){\setminus}\dom(\sig)\\
\MkT{\kind\sigPlus\kind'}{\MT}&\mbox{if}\ \sig(\m)=\MkT{\kind}{\MT}\ \mbox{and}\ \sig'(\m)=\MkT{\kind'}{\MT}
\end{cases}
\end{math}
\end{quoting}
where $\abs\sigPlus\abs=\defn\sigPlus\defn=\abs$, and $\abs\sigPlus\defn=\defn\sigPlus\abs=\defn$.\\

Symmetric sum handles conflicts, that is, method names occurring on both sides. Here, since this is not our focus, we take a simple choice, imposing method type-and-effects to be the same (up-to $\alpha$-renaming). Moreover, 
a method can be abstract on one side and defined on the other, and even defined on both; in this case, in the sum it becomes abstract, thus forcing objects implementing the sum to provide a definition.  A magic method, instead, cannot be inherited more than once.

\item[Right-preferential sum] The signature $\override{\sig}{\sig'}$ is defined in $\TEnv$ \mbox{when the following condition holds:}
\begin{quoting}
if $\sig(\m)=\MkT{\kind}{\MT}$ and $\sig'(\m)=\MkT{\kind'}{\MT'}$ then
\begin{quoting}
$\MT=\TMethEff{\Ext{\seq\X}{\seq\UT}}{\seq\T}{\T}{\eff}$ and $\MT'=\TMethEff{\Ext{\seq\X}{\seq\UT}}{\seq\T}{\T'}{\eff'}$, with $\SubType{\TEnv}{\TEff{\T'}{\eff'}}{\TEff{\T}{\eff}}$\\
$\kind'=\abs$  implies $\kind=\abs$, $\kind=\mgc$ iff $\kind'=\mgc$, $\kind=\kind'=\mgc$ implies $\T=\T'$
\end{quoting}
\end{quoting}
In this case, it is defined as follows: 
\begin{quoting}
\begin{math}
(\override{\sig}{\sig'})(\m)=
\begin{cases}
\sig(\m)\ &\mbox{if}\ \m\in\dom(\sig){\setminus}\dom(\sig')\\
\sig'(\m)\ &\mbox{otherwise}
\end{cases}
\end{math}
\end{quoting}
Right-preferential sum handles overriding. We allow method types  to be refined, abstract methods to only override abstract methods, magic methods to only override/be overriden by magic methods, and in this case \mbox{only the effect can be refined.}
\end{description}
Different policies could be  modularly  obtained by changing the above definitions. Note that, differently from left-preferential, our simple symmetric sum does not depend on $\TEnv$; it could be the case, e.g., if type-and-effects of conflicting methods could be in the subtyping relation.

 The subtyping and subeffecting relations are defined in \cref{fig:sub-eff}. 
 \begin{figure}[ht]
\begin{small}
\begin{math}
\begin{array}{c}

\NamedRule{sub-refl-var}{ 
}{ \SubType{\TEnv}{\X}{\X}}
{} 
\BigSpace
\NamedRule{sub-refl-ntype}{
}{ \SubType{\TEnv}{\Gen\tname{\seq\T}}{\Gen\tname{\seq\T} }}
{} 

\BigSpace
\NamedRule{sub-obj}{ \SubType{\TEnv}{\seq{\NT}}{\seq{\NT'}} \Space \SubType{\TEnv}{\sig}{\sig'}
}{ \SubType{\TEnv}{\TObj{\seq\NT}{\sig}}{\TObj{\seq\NT'}{\sig'}} }
{} 
\\[5ex]
\NamedRule{sub-ntypes-left}{\SubType{\TEnv}{\NT_i}{\T}}{\SubType{\TEnv}{\NT_1\ldots\NT_n}{\T}}{
i\in 1..n
}
\BigSpace
\NamedRule{sub-ntypes-right}{\SubType{\TEnv}{\T}{\NT_i}\ \forall i\in 1...n}{\SubType{\TEnv}{\T}{\NT_1\ldots\NT_n}}{}
\\[3ex]
\NamedRule{sub-ntype}{
}{ \SubType{\TEnv}{\Gen\tname{\seq\T}}{\Subst{\NT_i}{\seq\T}{\seq\Y}} }
{\prog(\tname)=\TDec{\Gen\tname{\Ext{\seq\Y}{\_}}}{\NT_1\ldots\NT_n}{\_}\\
i\in 1..n} 
\\[4ex]
\NamedRule{sub-empty}{}{\SubType{\TEnv}{\eZero}{\eff}}{}\BigSpace
\NamedRule{sub-top}{}{\SubType{\TEnv}{\eff}{\eTop}}{}
\\[3ex]
\NamedRule{sub-union-left}{\SubType{\TEnv}{\eff_1}{\eff}\Space\SubType{\TEnv}{\eff_2}{\eff}}{\SubType{\TEnv}{\EComp{\eff_1}{\eff_2}}{\eff}}{ }
\BigSpace
\NamedRule{sub-union-right}{\SubType{\TEnv}{\eff}{\eff_i}}{\SubType{\TEnv}{\eff}{\EComp{\eff_1}{\eff_2}}}{ i\in 1..2}
\\[3ex]
\NamedRule{sub-var-call}{\SubType{\TEnv}{\eff'}{\eff''}}{\SubType{\TEnv}{\eCall{\X}{\m}{\seq{\T}}}{\eff''}}{
\IsMType{\TEnv}{\TEnv(\X)}{\m}{\MkTE{\_}{\Ext{\seq\X}{\_}}{\_}{\_}{\eff}}\\
 \EffRed{\TEnv}{\Subst{\eff}{\seq\T}{\seq\X}}{\eff'}\\
}
\\[5ex]
\NamedRule{sub-mgc-call}{\SubType{\TEnv}{\T}{\T'}\Space\SubType{\TEnv}{\seq\T}{\seq\UT}}{\SubType{\TEnv}{\eCall{\T}{\m}{\seq{\T}}}{\eCall{\T'}{\m}{\seq{\UT}}}}{
\typeof{\TEnv}{\T}{\sig}\\
\typeof{\TEnv}{\T'}{\sig'}\\
\mkind(\sig,\m)=\mkind(\sig',\m)=\mgc
}
\\[5ex]
\NamedRule{sub-typeff}{\SubType{\TEnv,\Ext{\seq\X}{\seq\UT}}{\T'}{\T}\Space\SubType{\TEnv,\Ext{\seq\X}{\seq\UT}}{\eff'}{\eff}}{\SubType{\TEnv,\Ext{\seq\X}{\seq\UT}}{\TEff{\T'}{\eff'}}{\TEff{\T}}{\eff}}{}
\\[5ex]
\NamedRule{sub-mtypeff}{\SubType{\TEnv,\Ext{\seq\X}{\seq\UT}}{\TEff{\T'}{\eff'}}{\TEff{\T}}{\eff}}{\SubType{\TEnv}{\TMethEff{\Ext{\seq\X}{\seq\UT}}{\seq\T}{\T'}{\eff'}}{\TMethEff{\Ext{\seq\X}{\seq\UT}}{\seq\T}{\T}}{\eff}}{}
\\[3ex]
\NamedRule{sub-sig}{\SubType{\TEnv}{\MT_i}{\MT'_i}\ \forall i\in 1..n
}
{\SubType{\TEnv}{\sig}{\sig'}}
{\sig=\MSig{\m_1}{\kind_1}{\MT_1} \ldots \MSig{\m_n}{\kind_n}{\MT_{n+h}}\\
\sig'=\MSig{\m_1}{\kind'_1}{\MT'_1} \ldots \MSig{\m_n}{\kind'_n}{\MT'_n}\\
\SubType{}{\kind_i}{\kind'_i}\ \forall i\in 1..n}

\end{array}
\end{math}
\end{small}
\caption{Subtyping and subeffecting }\label{fig:sub-eff}
\end{figure}

Rules are mainly straightforward, we only comment those handling  call-effects.  In \refToRule{sub-var-call}, upper bounds of a variable effect are those of the bound of the variable in the type environment, instantiated on the actual type arguments, and simplified. In rule \refToRule{sub-mgc-call}, a magic call-effect can be refined by another for the same method,  where the receiver and the type arguments are subtypes.

Finally, well-formedness of type-and-effects  is defined in \cref{fig:wf-eff}.  
\begin{figure}[ht]
\begin{small}
\begin{math}
\begin{array}{c}
\NamedRule{WF-var}{
}{  \IsWFType{\TEnv}{\X} }
{\X\in\dom(\TEnv) } 
\BigSpace

\NamedRule{WF-otype}{
}{\IsWFType{\TEnv}{\TObj{\seq\NT}{\sig}}}
{\typeof{\TEnv}{\TObj{\seq\NT}{\sig}}{\_}
}
\\[3ex]
\NamedRule{WF-types}{
\IsWFType{\TEnv}{\T_i}\Space \forall  i \in 1..n
} 
{   \IsWFType{\TEnv}{\T_1\ldots\T_n} }
{ } 
\BigSpace
\NamedRule{WF-ntype}{
  \IsWFType{\TEnv}{\seq\T}
  } 
{   \IsWFType{\TEnv}{\Gen\tname{\seq\T}} }
{\seq\T=\T_1\ldots\T_n\\
\prog(\tname)=\TDec{\Gen\tname{\Ext{\seq\Y}{\T'_1\ldots\T'_n}}}{\_}{\_} \\
 \SubType{\TEnv}{\T_i}{\Subst{\T'_i}{\seq\T}{\seq\Y}}\ \forall i\in 1..n
} 
\\[5ex]
\NamedRule{WF-empty}{
}{  \IsWFType{\TEnv}{\eZero} }
{ }\BigSpace\NamedRule{WF-top}{
}{  \IsWFType{\TEnv}{\eTop} }
{ } 
\BigSpace
\NamedRule{WF-union}{\IsWFType{\TEnv}{\eff}\Space\IsWFType{\TEnv}{\eff'}
}{  \IsWFType{\TEnv}{\EComp{\eff}{\eff'}} }
{ } 
\\[3ex]
\NamedRule{WF-call}{\IsWFType{\TEnv}{\T}\Space\IsWFType{\TEnv}{\seq\T}}{\IsWFType{\TEnv}{\eCall{\T}{\m}{\seq{\T}}}}
{
\IsMType{\TEnv}{\T}{\m}{\MkT{\mgc}{\TMeth{\Ext{\seq\X}{\seq\UT}}{\_}{\_}}}\\
 \SubType{\TEnv}{\seq\T}{\Subst{\seq\UT}{\seq\T}{\seq\X}}
}
\\[6ex]
\NamedRule{WF-mtypeff}{\IsWFType{\TEnv,\Ext{\seq\X}{\seq{\UT}}}{\seq{\UT}}\Space\IsWFType{\TEnv,\Ext{\seq\X}{\seq{\UT}}}{\seq{\T}}\Space\IsWFType{\TEnv,\Ext{\seq\X}{\seq{\UT}}}{\T}
\Space\IsWFType{\TEnv,\Ext{\seq\X}{\seq{\UT}}}{\eff}}{\IsWFType{\TEnv}{\TMethEff{\Ext{\seq\X}{\seq\UT}}{\seq\T}{\T}{\eff}}}{}
\\[4ex]
\NamedRule{WF-sig}{
  \IsWFType{\TEnv}{\MT_i}\ \forall i\in 1..n
  } 
{   \IsWFType{\TEnv}{\m_1:\kind_1\ \MT_1\ \ldots\ \m_n:\kind_n\ \MT_n} }
{
} 
\end{array}
\end{math}
\end{small}
\caption{Well-formedness of type-and-effects}\label{fig:wf-eff}
\end{figure}

A type-and-effect is well-formed if, roughly, free variables are defined in the type environment, nominal types are instances, respecting the bounds, of those declared in the program, see the side conditions of rule \refToRule{WF-ntype}, and object types respect the constraints about conflicts and overriding, see the side condition in rule \refToRule{WF-otype}, where the judgment  $\typeof{\TEnv}{\TObj{\seq\NT}{\sig}}{\_}$ means that we can safely extract a signature from the type, as described before.

\section{Proofs}\label{sec:proofs}

\begin{lemma}[Inversion]\label[lemma]{lem:inversion}\
\begin{enumerate}
\item \label{lem:inversion:ret} If $ \WTExp{\Ret\ve}{\T}{\eff}$,\ then $\eff=\eZero$ and 
$\WTVal{\ve}{\T}$.
\item \label{lem:inversion:do} If $\WTExp{\Do\x{\e_1}{\e_2}}{\T}{\eff} $ then 
 $\WTExp{{\e_1}}{\T_1}{\eff_1}$ and  $  \IsWFExp{\emptyset}{\TVar{\T_1}{\x}}{{\e_2}}{\T}{\eff_2}$ and  
$\eff'=\EComp{\eff_1}{\eff_2}$.
\item \label{lem:inversion:try} If $\WTExp{\TryShort\e\handler}{\T}{\eff}$ then 
$\WTExp{\e}{\T'}{\eff'}$ and $\IsWFHandlerNarrow{\emptyset}{\emptyset}{\T'}{\handler}{\T}{\hfilter}$ and $\eff{=}\FilterFun{\eff'}{\hfilter}$. 
 \end{enumerate}
\end{lemma}

 \begin{lemma}[Inversion for method calls]\label[lemma]{lem:inversion:invk}
If $\WTExp{\MCall{\ve_0}\m{\seq{\T_\X}}{\seq\ve}}{\T}{\eff}$, then 
\begin{enumerate}
\item \label{invk:one} $\WTVal{\ve_0}{\T_0}$ and $\WTVal{\seq\ve}{\seq\T}$
\item  \label{invk:two} $\IsMType{}{\T_0}{\m}{\MkTE{\kind}{\Ext{\seq\X}\seq{\UT_\X}}{\seq{\T'}}{\T'}{\eff'}}$, with $\kind\neq\abs$ 
\item   \label{invk:three}  $\SubT{\seq{\T_\X}}{\Subst{\seq{\UT_\X}}{\seq{\T_\X}}{\seq\X}}$ and 
     $\SubT{\seq\T}{\Subst{\seq{\T'}}{\seq{\T_\X}}{\seq\X}}$ and $\T=\Subst{\T'}{\seq{\T_\X}}{\seq\X}$ and $\EffRed{}{\Subst{\eff'}{\seq{\T_\X}}{\seq\X}}{\eff}$
\item   \label{invk:four} 
if $\kind=\mgc$, then there are
\begin{itemize}
\item  $\tdec=\TDec{\Gen\tname{\Ext{\seq\Y}{\seq{\UT_\Y}}}}{\_}{\_}$ such that ${\IsMType{}{\tdec}{\m}{\MkTE{\mgc}{\Ext{\seq\X}{\seq{\UT_\X}}}{\seq{\T''}}{\T''}{\eCall{\Gen{\tname}{\seq\Y}}{\m}{\seq\X}}}}$
\item $\seq{\T_\Y}$ such that $\SubT{\T_0}{\Gen{\tname}{\seq{\T_\Y}}}$ and $\SubT{\seq{\T_\Y}}{\Subst{\UT_\Y}{\seq{\T_\Y}}{\seq\Y}}$ and $\mtype(\T_0,\m)=\mtype(\Gen{\tname}{\seq{\T_\Y}},\m)$
\end{itemize}
hence 
 $\seq{\T'}=\Subst{\seq{\T''}}{\seq{\T_\Y}}{\seq\Y}$, $\T'=\Subst{\T''}{\seq{\T_\Y}}{\seq\Y}$, and $\eff'=\eCall{\Gen{\tname}{\seq{\T_\Y}}}{\m}{\seq\X}$, which implies $\Subst{\eff'}{\seq{\T_\X}}{\seq\X}=\eCall{\Gen{\tname}{\seq{\T_\Y}}}{\m}{\seq{\T_\X}}$ which, being already simplified, is $\eff$.
 \end{enumerate}
\end{lemma}
\begin{proof}
From rule \refToRule{t-invk} of \cref{fig:typing-exp} we get \Cref{invk:one,invk:two,invk:three}, apart $\kind\neq\abs$ which follows from the fact that $\ve_0$ is an object and this is required by rule \refToRule{t-obj}.From \cref{invk:two}, and the definition of $\mtype$, derived from the signature extraction in \cref{fig:extract-sig}, we get \cref{invk:four}.
\end{proof}

\begin{lemma}[Progress for method calls]\label{lem:meth-call-progress} 
If $\WTExp{\MCall{\ve_0}\m{\seq{\T_\X}}{\seq\ve}}{\T}{\eff}$, then either
$\e\purered\e'$ for some $\e'$ or 
$\e\red\me$ for some $\me$.
\end{lemma}
\begin{proof}
We assume $\WTExp{\MCall{\ve_0}\m{\seq\T}{\seq\ve}}{\T}{\eff}$ and have to show that either
$\e\purered\e'$ for some $\e'$ or 
$\e\red\me$ for some $\me$.

By \cref{lem:inversion:invk} we get $\IsMType{}{\T_0}{\m}{\MkTE{\kind}{\Ext{\seq\X}\seq{\UT_\X}}{\seq{\T'}}{\T'}{\eff'}}$ with $\kind\neq\abs$. There are two cases:
\begin{itemize}
\item 
$\kind=\defn$. In this case, from the definition of   $\mtype(\T_0,\m)$,  we get $\mbody(\ve_0,\m) =\GenMBody{\seq\X}{\x}{\seq\x}{\e}$ with $\seq\x$ and $\seq\ve$ of the same length, hence rule \refToRule{invk} of \cref{fig:pure-red} is applicable, and \mbox{we get $\e\purered\e'$ for some $\e'$. }

\item $\kind=\mgc$. In this case, again from \cref{lem:inversion:invk}, we have 
\begin{itemize}
\item $\WTVal{\ve_0}{\T_0}$ and $\WTVal{\seq\ve}{\seq\T}$
\item
there is a type declaration $\TDec{\Gen\tname{\Ext{\seq\Y}{\seq{\UT_\Y}}}}{\_}{\_}$ such that 
\begin{quoting}
${\IsMType{}{\TDec{\Gen\tname{\Ext{\seq\Y}{\seq{\UT_\Y}}}}{\_}{\_}}{\m}{\MkTE{\mgc}{\Ext{\seq\X}{\seq{\UT_\X}}}{\seq{\T''}}{\T''}{\eCall{\Gen{\tname}{\seq\Y}}{\m}{\seq\X}}}}$
\end{quoting}
\item there is $\seq{\T_\Y}$ such that $\SubT{\T_0}{\Gen{\tname}{\seq{\T_\Y}}}$ and $\SubT{\seq{\T_\Y}}{\Subst{\UT_\Y}{\seq{\T_\Y}}{\seq\Y}}$
\item $\SubT{\seq\T}{\Subst{\seq{\UT_\X}}{\seq\T}{\seq\X}}$ and 
     $\SubT{\seq\T}{\Subst{\Subst{\seq{\T''}}{\seq{\T_\Y}}{\seq\Y}}{\seq\T}{\seq\X}}$ and $\T=\Subst{\Subst{\T''}{\seq{\T_\Y}}{\seq\Y}}{\seq\T}{\seq\X}$ and $\eff=\eCall{\Gen{\tname}{\seq{\T_\Y}}}{\m}{\seq\T}$. 
 \end{itemize}

 Hence, from the definition of  $\mtype(\T_0,\m)$, we get $\mbody(\ve_0,\m)=\Pair{\mgc}{\tname}$. Since the premises of rule \refToRule{t-run} are verified, $\mrun\tname\m(\ve_0,\seq\ve)$ is defined, so \refToRule{mgc} of \cref{fig:monadic-red} is applicable and $\e\red\me$ for some $\me$. 
\end{itemize}
\end{proof}

\begin{lemma}[Progress for try expressions]\label{lem:mnd-pure-red-progress}If $\WTExp{\TryShort{\e}{\handler}}{\T}{\eff}$, then $\e\purered\e'$ \mbox{for some $\e'$.}
\end{lemma}
\begin{proof}
Let $\WTExp{\TryShort{\e}{\handler}}{\T}{\eff}$. By cases on $\e$ and induction on try expressions.
\begin{description}
\item [$\e=\Ret\ve$]  Rule \refToRule{try-ret} is applicable.
\item[$\e=\Do\y{\e_1}{\e_2}$]   Rule \refToRule{try-do}  is applicable.
\item[$\e=\MCall{\ve_0}\m{\seq{\T}}{\seq\ve}$] From \refItem{lem:inversion}{try} we have  $\WTExp{\MCall{\ve_0}\m{\seq\T}{\seq\ve}}{\T'}{\eff'}$ 
for some $\T'$ and $\eff'$. From \cref{lem:meth-call-progress} we have that either $\e\purered\e_1$ for some $\e_1$ or 
$\e\red\me$ for some $\me$.  In the first case rule \refToRule{try-ctx} is applicable and $\TryShort{\e}{\handler}\purered\TryShort{\e_1}{\handler}$.
In the second case $ \mbody(\ve,\m)=\Pair{\mgc}{\_}$. Therefore either one of the rules \refToRule{catch-continue} or \refToRule{catch-stop}
is applicable or if there os no handler rule \refToRule{fwd} can be applied. 
\item[$\e=\TryShort{\e_1}{\handler_1}$] From \refItem{lem:inversion}{try} we have  $\WTExp{\TryShort{\e_1}{\handler_1}}{\T'}{\eff'}$ 
for some $\T'$ and $\eff'$. By induction hypothesis $\e\purered\e'$, so rule  rule \refToRule{try-ctx} is applicable. 
\end{description}
\end{proof}

\begin{proofOf}{theo:mnd-progress}
Assuming $\WTExp{\e}{\T}{\eff}$, we have to prove that  either $\e = \Ret{\ve}$ for some $\ve\in\Val$, or 
$\e\red\me$ for some $\me\in\mfun\Exp$.
By induction on the typing rules of \cref{fig:typing-exp}.
\begin{description}
\item [\refToRule{t-invk}] In this case $\e$ is $\MCall{\ve_0}\m{\seq{\T_\X}}{\seq\ve}$.  By \cref{lem:meth-call-progress} we get that  either 
$\e\purered\e'$ for some $\e'$ or $\e\red\me$ for some $\me$. In the first case by rule \refToRule{pure} of \cref{fig:monadic-red} we get
$\e\red \mun(\e') $, and in the second case we are done.
\item [\refToRule{t-do}] In this case $\e$ is $\Do\x{\e_1}{\e_2}$. From \refItem{lem:inversion}{do} we have $\WTExp{{\e_1}}{\T_1}{\eff_1}$ and $\IsWFExp{\emptyset}{\TVar{\T_1}{\x}}{{\e_2}}{\T'}{\eff_2}$ and $\eff'=\EComp{\eff_1}{\eff_2}$. By induction hypothesis we get that either $\e_1 = \Ret\ve$ for some $\ve\in\Val$, or 
$\e_1\red\me$ for some $\me$. In the first case we can apply rule \refToRule{ret}  and in the second rule \refToRule{do} of \cref{fig:monadic-red}. 
\item [\refToRule{t-try}] In this case $\e$ is $\TryShort{\e_1}{\handler}$. By \cref{lem:mnd-pure-red-progress} we get that 
$\e\purered\e'$, so by rule \refToRule{pure} of \cref{fig:monadic-red} we get
$\e\red \mun(\e')$.
\end{description}
\end{proofOf}

 \begin{lemma}[Inversion for handlers]\label[lemma]{lem:h-inversion} 
If $\IsWFHandler{\emptyset}{\emptyset}{\T'}{\Handler{\seq\cc}{\x_0}{\e_0}}{\T}{\hfilter}$ where
$
\cc_i=\CC{\NT_i}{\m_i}{\seq\X^i}{\x_i}{\seq\x^i}{\e_i}{\mode_i}
$
for $i\in 1..n$, then $\IsMType{}{\NT_i}{\m_i}{\MkTE{\mgc}{\Ext{\seq\X^i}{\seq\UT^i}}{\seq\T^i}{\T_i}{\_}}$ and
\begin{enumerate}
\item \label{lem:h-inversion:one}  $\IsWFExp{\emptyset}{\TVar{\T'}{\x_0}}{\e_0}{\T_0}{\eff_0}$ with  $\SubType{\TEnv}{\T_0}{\T}$
\item \label{lem:h-inversion:two}  $\hfilter=\HFilter{\cfilter_1\dots \cfilter_n}{\eff_0}$ with 
$\cfilter_i=\CFilter{\T_i}{\m_i}{\seq{\X}^i}{\seq\UT^i}{\eff_i}{}$ and
\item \label{lem:h-inversion:three} $\IsWFClause{\emptyset}{\emptyset}{\T}{\cc_i}{\cfilter_i}$ and
\begin{enumerate}
\item  \label{lem:h-inversion:four}$ \IsWFExp{\emptyset}{\TVar{\NT_i}{\x},\TVars{\T^i}{\x^i}} {\e_i} {\T'_i} {\eff}$ with $\SubType{\TEnv}{\T'_i}{\T_i}$ if $\mode_i=\Continue$
\item \label{lem:h-inversion:five}$ \IsWFExp{\emptyset}{\TVar{\NT_i}{\x},\TVars{\T^i}{\x^i}} {\e_i} {\T'_i} {\eff}$ with $\SubType{\TEnv}{\T'_i}{\T_0}$  if $\mode_i=\Stop$
\end{enumerate}
\end{enumerate}
\end{lemma}

 \begin{lemma}[Properties of Effect Simplification]\label{lem:eff-simplification}
 If $\EffRed{\TEnv}{\eff}{\eff'}$ and  $\SubType{\TEnv}{\eff_1}{\eff}$ and $\EffRed{\TEnv}{\eff_1}{{\eff'_1}}$ , then $\SubType{\TEnv}{{\eff'_1}}{\eff'}$.
 \end{lemma}

 \begin{lemma}[Properties of $\FilterF{}$]\label{lem:filter-prop}\
 \begin{enumerate}
\item \label{lem:filter-prop:uno} Let  $\hfilter=\HFilter{\cfilter_1\ldots\cfilter_n}{\eff_{\hfilter}}$ and $\hfilter'=\HFilter{\cfilter_1\ldots\cfilter_n}{\eff'_{\hfilter}}$ and 
 $\SubType{\TEnv}{\eff_{\hfilter}}{\eff'_{\hfilter}}$ and  $\SubType{\TEnv}{{\eff}}{\eff'}$. Then $\SubType{\TEnv}{\FilterFun{\eff}{\hfilter}}{\FilterFun{\eff'}{\hfilter'}}$.
 \item \label{lem:filter-prop:due} $\FilterFun{\EComp{\eff_1}{\eff_2}}{\HFilter{\seq\cfilter}{\eff}}=\FilterFun{{\eff_1}}{\HFilter{\seq\cfilter}{\eff'}}$ where
 $\eff'=\FilterFun{{\eff_2}}{\HFilter{\seq\cfilter}{\eff}}$.
 \end{enumerate}
 \end{lemma}
\begin{proof}
We prove the second point, since the first one derives from monotonicity of $\vee$.\\
From definition  of $\FilterF{\hfilter}$ in \cref{fig:filters} we get $\eff'=\EComp{\FilterFun{{\eff_2}}{\seq\cfilter}} {\eff}$ and \\
\centerline{
$
\FilterFun{{\eff_1}}{\HFilter{\seq\cfilter}{\eff'}}=\EComp{\FilterFun{\eff_1}{\seq\cfilter}}{\EComp{\FilterFun{{\eff_2}}{\seq\cfilter}} {\eff}}
= \EComp{\FilterFun{\EComp{\eff_1}{\eff_2}}{\seq\cfilter}} {\eff}
= \FilterFun{\EComp{\eff_1}{\eff_2}}{\HFilter{\seq\cfilter}{\eff}}
$}
\end{proof}

We prove the subsumption lemma considering the different kinds of judgements occurring in a type derivation which are the judgments for expressions, values, handlers and catch clauses. 

Define $\SubType{\TEnv}{{\cfilter}}{\cfilter'}$ iff $\cfilter={\CFilter{\T}{\m}{\seq{\X}}{\seq\UT}{\eff}{\_}}$ and $\cfilter'={\CFilter{\T}{\m}{\seq{\X}}{\seq\UT}{\eff'}{\_}}$ and  $\SubType{\TEnv}{{\eff}}{\eff'}$.
\begin{lemma}[Context Subsumption]\label{lem:subsumption}
Let $\SubType{\TEnv}{\seq{\UT}}{\seq\T}$.
\begin{itemize}
\item If  $\IsWFExp{\TEnv}{\Gamma,\TVar{\seq\T}{\seq\x}}{\e}{\T}{\eff}$ then $\IsWFExp{\TEnv}{\Gamma,\TVar{\seq{\UT}}{\seq\x}}{\e}{\T'}{\eff'}$ and $\SubType{\TEnv}{\TEff{\T'}{\eff'}}{\TEff\T\eff}$.
\item If  $\IsWFVal{\TEnv}{\Gamma,\TVar{\seq\T}{\seq\x}}{\ve}{\T}$, then $\IsWFVal{\TEnv}{\Gamma,\TVar{\seq\UT}{\seq\x}}{\ve}{\T'}$ and $\SubType{\TEnv}{{\T'}}{\T}$.
\item If  $\IsWFHandler{\TEnv}{\Gamma,\TVar{\seq\T}{\seq\x}}{\T'}{\handler}{\T}{\HFilter{\cfilter_1\ldots\cfilter_n}{\eff}}$, then   $\IsWFHandler{\TEnv}{\Gamma,\TVar{\seq\UT}{\seq\x}}{\T'}{\handler}{\T}{\HFilter{\cfilter'_1\ldots\cfilter'_n}{\eff'}}$ with $\SubType{\TEnv}{{\eff'}}{\eff}$ and $\SubType{\TEnv}{{\cfilter_i}}{\cfilter'_i}$ for all $i\in 1..n$.
\item If $\IsWFClause{\TEnv}{\Gamma,\TVar{\seq\T}{\seq\x}}{\T}{\cc}{\cfilter}$, then $\IsWFClause{\TEnv}{\Gamma,\TVar{\seq\UT}{\seq\x}}{\T}{\cc}{\cfilter'}$ with $\SubType{\TEnv}{{\cfilter}}{\cfilter'}$.
\end{itemize}
\end{lemma}
\begin{proof}
By induction on the type derivation of \cref{fig:typing-exp}.
Consider the last rule applied in the derivation.
 \begin{description}
 \item [\refToRule{T-Var}]  Let $\IsWFVal{\TEnv}{\Gamma,\TVar{\seq\T}{\seq\x}}{\x}{\T}$, then either $\x=\x_i\in\seq\x$ and 
 $\IsWFVal{\TEnv}{\Gamma,\TVar{\seq\UT}{\seq\x}}{\x}{\UT_i}$ with $\SubType{\TEnv}{{\UT_i}}{\T_i}$ or $\x\not\in\seq\x$  and $\T=\T'$.
  \item [\refToRule{T-Obj}] Since method declaration do not have free variables the result is immediate.
\item [\refToRule{T-Invk}]   In this case $\e$ is $\MCall{\ve_0}\m{\seq{\T_\X}}{\seq\ve}$. 
From rule \refToRule{T-Invk} we have $\IsWFVal{\TEnv}{\Gamma,\TVar{\seq\T}{\seq\x}}{\ve_0}{\T_0}$ and $\IsWFVal{\TEnv}{\Gamma,\TVar{\seq\T}{\seq\x}}{\seq\ve}{\seq\T}$ and $\IsMType{}{\T_0}{\m}{\MkTE{\kind}{\Ext{\seq\X}\seq{\UT_\X}}{\seq{\T'}}}{\T_\m}{\eff_\m}$ and $\SubType{\TEnv}{\seq{\T_\X}}{\Subst{\seq{\UT_\X}}{\seq{\T_\X}}{\seq\X}}$ and 
     $\SubType{\TEnv}{\seq\T}{\Subst{\seq{\T'}}{\seq{\T_\X}}{\seq\X}}$ and $\T=\Subst{\T_\m}{\seq{\T_\X}}{\seq\X}$ and $\EffRed{\TEnv}{\Subst{\eff_\m}{\seq{\T_\X}}{\seq\X}}{\eff}$. By induction hypothesis on $\IsWFVal{\TEnv}{\Gamma,\TVar{\seq\T}{\seq\x}}{\ve_0}{\T_0}$, we get
     $\IsWFVal{\TEnv}{\Gamma,\TVar{\seq\UT}{\seq\x}}{\ve_0}{\T_0'}$ and $\SubType{\TEnv}{{\T_0'}}{\T_0}$. By the overriding rules enforced by the the ``right-preferential rules'', $\IsMType{}{\T'_0}{\m}{\MkTE{\kind}{\Ext{\seq\X}\seq{\UT_\X}}{\seq{\T'}}}{\T'_{\m}}{\eff'_\m}$ with $\SubType{\TEnv,\Ext{\seq\X}\seq{\UT_\X}}{\TEff{\T'_{\m}}{\eff'_\m}}{\TEff{\T_{\m}}{\eff_\m}}$. 
By induction hypothesis on $\IsWFVal{\TEnv}{\Gamma,\TVar{\seq\T}{\seq\x}}{\seq\ve}{\seq\T}$ we get 
 $\IsWFVal{\TEnv}{\Gamma,\TVar{\seq\UT}{\seq\x}}{\seq\ve}{\seq\T'}$ and $\SubType{\TEnv}{{\seq\T}}{\seq\T'}$. \\
 Applying rule \refToRule{T-Invk} we get $\IsWFExp{\TEnv}{\Gamma,\TVar{\seq{\UT}}{\seq\x}}{\e}{\T'}{\eff'}$ where $\T'=\Subst{\T'_\m}{\seq{\T_\X}}{\seq\X}$ and $\EffRed{\TEnv}{\Subst{\eff'_\m}{\seq{\T_\X}}{\seq\X}}{\eff'}$. From $\SubType{\TEnv}{\Subst{\eff'_\m}{\seq{\T_\X}}{\seq\X}}{\Subst{\eff_\m}{\seq{\T_\X}}{\seq\X}}$ and
\cref{lem:eff-simplification}  we get $\SubType{\TEnv}{\eff'}{\eff}$. Therefore $\SubType{\TEnv}{\TEff{\T'}{\eff'}}{\TEff\T\eff}$.
\item [\refToRule{T-Ret}]  In this case $\e$ is ${\Ret\ve}$ and $\IsWFExp{\TEnv}{\Gamma,\TVar{\seq\T}{\seq\x}}{\Ret\ve}{\T}{\eZero}$. From rule \refToRule{T-Rrt} we have $\IsWFVal{\TEnv}{\Gamma,\TVar{\seq\T}{\seq\x}}{\ve}{\T}$. By induction hypothesis, we get
     $\IsWFVal{\TEnv}{\Gamma,\TVar{\seq\UT}{\seq\x}}{\ve}{\T'}$ and $\SubType{\TEnv}{{\T'}}{\T}$.  Applying rule \refToRule{T-Ret} we get 
   $\IsWFExp{\TEnv}{\Gamma,\TVar{\seq\UT}{\seq\x}}{\Ret\ve}{\T'}{\eZero}$.
\item [\refToRule{T-Do}]  In this case $\e$ is $\Do\y{\e_1}{\e_2}$ and, from rule \refToRule{T-Do},
$ \IsWFExp{\TEnv}{\Gamma,\TVar{\seq\T}{\seq\x}}{{\e_1}}{\T_1}{\eff_1}$ and $\IsWFExp{\TEnv}{\Gamma,\TVar{\seq\T}{\seq\x},\TVar{\T_1}{\x}}{{\e_2}}{\T}{\eff_2}$ and 
and $\eff=\EComp{\eff_1}{\eff_2}$. By induction hypothesis on $ \IsWFExp{\TEnv}{\Gamma,\TVar{\seq\T}{\seq\x}}{{\e_1}}{\T_1}{\eff_1}$ we get
$\IsWFExp{\TEnv}{\Gamma,\TVar{\seq{\UT}}{\seq\x}}{\e_1}{\T'_1}{\eff'_1}$ and $\SubType{\TEnv}{\TEff{\T'_1}{\eff'_1}}{\TEff{\T_1}{\eff_1}}$.
Again by inductive hypothesis on $\IsWFExp{\TEnv}{\Gamma,\TVar{\seq\T}{\seq\x},\TVar{\T_1}{\x}}{{\e_2}}{\T}{\eff_2}$ we get
$\IsWFExp{\TEnv}{\Gamma,\TVar{\seq{\UT}}{\seq\x},\TVar{\T'_1}{\x}}{\e_2}{\T'}{\eff'_2}$ and $\SubType{\TEnv}{\TEff{\T'}{\eff'_2}}{\TEff{\T}{\eff_2}}$.
 Applying rule \refToRule{T-Do} we get $\IsWFExp{\TEnv}{\Gamma,\TVar{\seq{\UT}}{\seq\x}}{\e}{\T'}{\EComp{\eff'_1}{\eff'_2}}$ with 
 $\SubType{\TEnv}{\TEff{\T'}{{\EComp{\eff'_1}{\eff'_2}}}}{\TEff{\T}{{\EComp{\eff_1}{\eff_2}}}}$.
\item [\refToRule{T-Try}]   In this case $\e$ is  $\TryShort{\e_1}{\handler}$  and, from rule \refToRule{T-Try}, 
$\IsWFExp{\TEnv}{\Gamma,\TVar{\seq\T}{\seq\x}}{\e_1}{\T_1}{\eff_1}$ and
$\IsWFHandler{\TEnv}{\Gamma,\TVar{\seq\T}{\seq\x}}{\T_1}{\handler}{\T}{\hfilter}$ where $\hfilter=\HFilter{\cfilter_1\ldots\cfilter_n}{\eff_{\hfilter}}$ and $\eff=\FilterFun{\eff_1}{\hfilter}$.
By induction hypothesis  we get $\IsWFExp{\TEnv}{\Gamma,\TVar{\seq{\UT}}{\seq\x}}{\e_1}{\T'_1}{\eff'_1}$ and $\SubType{\TEnv}{\TEff{\T'_1}{\eff'_1}}{\TEff{\T_1}{\eff_1}}$. Let $\handler=\Handler{\cc_1\ldots\cc_n}{\x}{\e'}$.
By induction hypothesis on $\IsWFHandler{\TEnv}{\Gamma,\TVar{\seq\T}{\seq\x}}{\T_1}{\handler}{\T}{\hfilter}$ we get  $\IsWFHandler{\TEnv}{\Gamma,\TVar{\seq{\UT}}{\seq\x}}{\T_1}{\handler}{\T'}{\hfilter'}$
where $\hfilter'=\HFilter{\cfilter_1\ldots\cfilter_n}{\eff'_{\hfilter}}$ and  $\SubType{\TEnv}{\TEff{\T'}{\eff'_{\hfilter}}}{\TEff{\T}{\eff_{\hfilter}}}$. 
Therefore applying rule \refToRule{T-Try} we get $\IsWFExp{\TEnv}{\Gamma,\TVar{\seq{\UT}}{\seq\x}}{\e}{\T'}{\eff'}$ where $\eff'=\FilterFun{\eff'_1}{\hfilter'}$.
From \cref{lem:filter-prop} we get that  $\SubType{\TEnv}{\TEff{\T'}{\eff'}}{\TEff{\T}{\eff}}$ which proves the result.
 \item [\refToRule{T-Handler}]  In this case we have $\IsWFHandler{\TEnv}{\Gamma,\TVar{\seq\T}{\seq\x}}{\T'}{\handler}{\T}{\HFilter{\seq\cfilter}{\eff}}$.
 where $\handler=\Handler{\seq\cc}{\x_0}{\e_0}$. 
 From rule \refToRule{T-Handler}
  $\IsWFExp{\TEnv}{\Gamma,\TVar{\seq\T}{\seq\x},\TVar{\T'}{\x_0}}{\e_0}{\T_0}{\eff}$  and $\IsWFClause{\TEnv}{\Gamma,\TVar{\seq\T}{\seq\x}}{\T}{\cc_i}{\cfilter_i}$ for all $i\in 1..n$.
  By induction hypothesis we get 
  $\IsWFExp{\TEnv}{\Gamma,\TVar{\seq{\UT}}{\seq\x},\TVar{\T'}{\x_0} }{\e_0}{\T'_0}{\eff'}$ and $\SubType{\TEnv}{\TEff{\T'_0}{\eff'}}{\TEff{\T_0}{\eff}}$. By induction
  hypotheses on $\IsWFClause{\TEnv}{\Gamma,\TVar{\seq\T}{\seq\x}}{\T}{\cc_i}{\cfilter_i}$, for all $i\in 1..n$ ,we have
  $\IsWFClause{\TEnv}{\Gamma,\TVar{\seq\UT}{\seq\x}}{\T}{\cc_i}{\cfilter'_i}$  with $\SubType{\TEnv}{{\cfilter_i}}{\cfilter'_i}$  for all $i\in 1..n$. Applying rule \refToRule{T-Handler} we get 
   $\IsWFHandler{\TEnv}{\Gamma,\TVar{\seq\UT}{\seq\x}}{\T'}{\handler}{\T}{\HFilter{\seq\cfilter'}{\eff'}}$.
\item [\refToRule{T-Continue}]  In this case we have $\IsWFClause{\TEnv}{\Gamma,\TVar{\seq\T}{\seq\x}}{\T}{\CC{\NT_\x}{\m}{\seq\X}{\x}{\seq\x}{\e_\m}{\Continue}}{\CFilter{\T_\x}{\m}{\seq{\X}}{\seq\UT}{\eff_\m}{\Continue}}$. From rule \refToRule{T-Continue} we get $\IsWFExp{\TEnv}{\Gamma,\TVar{\NT_\x}{\x},\TVars{\T}{\x}} {\e_\m} {\T'} {\eff_\m}$ where 
$\IsMType{\TEnv}{\NT_\x}{\m}{\MkTE{\mgc}{\Ext{\seq\X}{\seq\UT}}{\seq\T}{\T_\m}{\_}}$ with  $\SubType{\TEnv}{\T'}{\T_\m}$. By induction hypothesis 
$\IsWFExp{\TEnv}{\Gamma,\TVar{\NT_\x}{\x},\TVars{\UT}{\x}} {\e_\m} {\T''} {\eff'_\m}$ such that $\SubType{\TEnv}{\TEff{\T''}{\eff'_\m}}{\TEff{\T'}{\eff_\m}}$. Since
by transitivity $\SubType{\TEnv}{\T''}{\T_\m}$ we can apply rule \refToRule{T-Continue} and we get 
$\IsWFClause{\TEnv}{\Gamma,\TVar{\seq\UT}{\seq\x}}{\T}{\CC{\NT_\x}{\m}{\seq\X}{\x}{\seq\x}{\e_\m}{\Continue}}{\CFilter{\T_\x}{\m}{\seq{\X}}{\seq\UT}{\eff'_\m}{\Continue}}$.
\item [\refToRule{T-Stop}]   In this case we have ${\IsWFClause{\TEnv}{\Gamma,\TVar{\seq\T}{\seq\x}}{\T}{\CC{\NT_\x}{\m}{\seq\X}{\x}{\seq\x}{\e_\m}{\Stop}}{\CFilter{\T_\x}{\m}{\seq{\X}}{\seq\UT}{\eff_\m}{\Continue}}}$. From rule \refToRule{T-Stop} we get $\IsWFExp{\TEnv}{\Gamma,\TVar{\NT_\x}{\x},\TVars{\T}{\x}} {\e_\m} {\T'} {\eff_\m}$
where $\IsMType{\TEnv}{\NT_\x}{\m}{\MkTE{\mgc}{\Ext{\seq\X}{\seq\UT}}{\seq\T_\m}{\_}{\_}}$ and $  \SubType{\TEnv}{\T'}{\T}$.
By induction hypothesis  $\IsWFExp{\TEnv}{\Gamma,\TVar{\NT_\x}{\x},\TVars{\UT}{\x}} {\e_\m} {\T''} {\eff'_\m}$  such that $\SubType{\TEnv}{\TEff{\T''}{\eff'_\m}}{\TEff{\T'}{\eff_\m}}$.
 Since by transitivity $\SubType{\TEnv}{\T''}{\T}$ we can apply rule \refToRule{T-Stop} and we get 
 ${\IsWFClause{\TEnv}{\Gamma,\TVar{\seq\T}{\seq\x}}{\T}{\CC{\NT_\x}{\m}{\seq\X}{\x}{\seq\x}{\e_\m}{\Stop}}{\CFilter{\T_\x}{\m}{\seq{\X}}{\seq\UT}{\eff_\m}{\Continue}}}$.
 \end{description}
\end{proof}

\begin{lemma}\label{lem:h-subsumption}
If
$\IsWFHandler{\TEnv}{\Gamma}{\T}{\handler}{\T'}{\HFilter{\seq\cfilter}{\eff}}$ and $\SubType{\TEnv}{\T}{\T''}$, then   
$\IsWFHandler{\TEnv}{\Gamma}{\T'}{\handler}{\T'}{\HFilter{\seq\cfilter}{\eff'}}$ with $\SubType{\TEnv}{{\eff'}}{\eff}$.
\end{lemma}
\begin{proof}
By rule \refToRule{T-Handler} and \cref{lem:subsumption}.
\end{proof}

\begin{lemma}[Substitution for types]\label{lem:substituion:type}
If $\SubType{\Ext{\seq\X}{\seq{\UT}}}{\TEff{\T}{\eff}}{\TEff{\T'}}{\eff'}$ and $\SubType{}{\seq\T}{\Subst{\seq\UT}{\seq\T}{\seq\X}}$ then 
$\SubType{}{\TEff{\Subst{\T}{\seq\T}{\seq\X}}{\eff_1}}{\TEff{\Subst{\T'}{\seq\T}{\seq\X}}{\eff'_1}}$ 
where $\EffRed{}{\Subst{\eff}{\seq{\T}}{\seq\X}}{\eff_1}$  and  $\EffRed{}{\Subst{\eff'}{\seq{\T}}{\seq\X}}{\eff'_1}$. 
\end{lemma}
\begin{proof}
By induction on the rules of \cref{fig:sub-eff} using \cref{lem:eff-simplification}.
\end{proof}

\begin{lemma}[Substitution]\label{lem:substitution}
If
\begin{itemize}
\item $\IsWFExp{\TEnv,\Ext{\seq\Y}{\seq{\UT_\Y}}}{\Gamma,\TVar{\seq{\UT_\x}}{\seq\x}}{\e}{\T}{\eff}$ 
and 
\item $\SubType{\TEnv}{\seq{\T_\X}}{\Subst{\seq{\UT_\Y}}{\seq{\T}}{\seq\Y}}$
and  
\item $\IsWFVal{\TEnv}{\Gamma}{\seq\ve}{\seq{\T_\x}}$ and 
$\SubType{\TEnv}{{\seq{\T_\x}}}{\Subst{\seq{\UT_\x}}{\seq{\T_\X}}{\seq\Y}}$
\end{itemize}
then 
 $\IsWFExp{\TEnv}{\Gamma}{\Subst{\Subst\e{\seq\ve}{\seq\x}}{\seq{\T_\X}}{\seq\Y}}{\T'}{\eff'}$ 
$\IsWFExp{\TEnv}{\Gamma}{\Subst{\Subst\e{\seq\ve}{\seq\x}}{\seq{\T_\X}}{\seq\Y}}{\T'}{\eff'}$ and  $\SubType{\TEnv}{\TEff{\T'}{\eff'}}{\TEff{\Subst{\T}{\seq{\T_\X}}{\seq\Y}}{\eff''}}$
where $\EffRed{}{\Subst{\eff}{\seq{\T_\X}}{\seq\Y}}{\eff''}$.
\end{lemma}
\begin{proof}
By induction on the the rules of \cref{fig:typing-exp} using \cref{lem:subsumption}.
\end{proof}

\begin{proofOf}{lem:subject-reduction}
We assume $\WTExp{\e}{\T}{\eff}$ and $\e\purered \e'$, and have to show $\WTExp{\e'}{\T'}{\eff'}$
such that  $\SubT{\TEff{\T'}{\eff'}}{\TEff\T\eff}$.
By induction on the reduction rules of \cref{fig:pure-red}. 
\begin{description}
\item [\refToRule{invk}] 

 In this case $\e$ is $\MCall{\ve_0}\m{\seq{\T_\X}}{\seq\ve}$ and $\e'$ is $\Subst{\Subst{\Subst{\e}{\seq{\T_\X}}{\seq\X}}{\ve}{\x}}{\seq\ve}{\seq\x}$ and $\mbody(\ve,\m) = \ple{\seq\X,\x,\seq\x,\e}$. 
From \cref{lem:inversion:invk} we have $\WTVal{\ve_0}{\T_0}$ and $\WTVal{\seq\ve}{\seq\T}$ and $\IsMType{}{\T_0}{\m}{\MkTE{\kind}{\Ext{\seq\X}\seq{\UT_\X}}{\seq{\T'}}}{\T_\m}{\eff_\m}$ and $\SubT{\seq{\T_\X}}{\Subst{\seq{\UT_\X}}{\seq{\T_\X}}{\seq\X}}$ and 
     $\SubT{\seq\T}{\Subst{\seq{\T'}}{\seq{\T_\X}}{\seq\X}}$ and $\T=\Subst{\T_\m}{\seq{\T_\X}}{\seq\X}$ and $\EffRed{}{\Subst{\eff_\m}{\seq{\T_\X}}{\seq\X}}{\eff}$.
     From the definition of $\mbody$, we have $\kind=\defn$. 
     Hence, by rule \refToRule{t-meth},
$\IsWFExp{\Ext{\seq\X}{\seq{\UT_\X}}}{\TVar{\Gen\tname{\seq{\T_\Y}}}{\x}, \TVar{\seq\T}{\seq\x} } {\e} {\T''} {\eff''}$
with $\SubType{\Ext{\seq\X}{\seq{\UT_\X}}}{\TEff{\T''}{\eff''}}{\TEff{\T_\m}{\eff_\m}}$. From  \cref{lem:substitution} we get
\begin{quoting}
$\WTExp{\Subst{\Subst{\Subst{\e}{\seq{\T_\X}}{\seq\X}}{\ve}{\x}}{\seq\ve}{\seq\x}}{\T'}{\eff'}\quad
\SubT{\TEff{\T'}{\eff'}}{\TEff{\Subst{\T''}{\seq{\T_\X}}{\seq\X}}{\hat{\eff}}}\quad\EffRed{}{\Subst{\eff''}{\seq{\T_\X}}{\seq\X}}{\hat{\eff}}
$
\end{quoting}
From $\SubType{\Ext{\seq\X}{\seq{\UT_\X}}}{\TEff{\T''}{\eff''}}{\TEff{\T_\m}{\eff_\m}}$ and $\SubT{\seq{\T_\X}}{\Subst{\seq{\UT_\X}}{\seq{\T_\X}}{\seq\X}}$ and \cref{lem:substituion:type} we get
$\SubT{\TEff{\Subst{\T''}{\seq{\T_\X}}{\seq\X}}{\Subst{\eff''}{\seq{\T_\X}}{\seq\X} }}{\TEff{\Subst{\T_\m}{\seq{\T_\X}}{\seq\X}}{ \Subst{\eff_\m}{\seq{\T_\X}}{\seq\X} }}$, hence from \cref{lem:eff-simplification} we get $\SubT{\TEff{\T'}{\eff'}}{\TEff{\T}{\eff}}$.

\item [\refToRule{try-ret}]  In this case $\e$ is $\TryShort{\Ret\ve}{\handler}$ and 
$\e'$ is $\Do\x{\Ret\ve}{\e_0}$.  Let $\handler=\Handler{\seq\cc}{\x}{\e_0}$.  From \refItem{lem:inversion}{try} and \refItem{lem:inversion}{ret} we have 
$\WTExp{\Ret\ve}{\T''}{\eZero}$ and $\IsWFHandlerNarrow{\emptyset}{\emptyset}{\T''}{\handler}{\T}{\hfilter}$ and
$\eff{=}\FilterFun{\eZero}{\hfilter}$. From the definition of $\FilterF{\hfilter}\ $ in \cref{fig:filters} and 
\refItem{lem:h-inversion}{one} we have that 
 $\IsWFExp{\emptyset}{\TVar{\T''}{\x}}{\e_0}{\T_0}{\eff}$ with  $\SubType{}{\T_0}{\T}$.
 Applying rule \refToRule{t-do} of \cref{fig:typing-exp} we get  $\WTExp{\e'}{\T_0}{(\EComp\eZero{\eff})}$ and
 from $\SubType{}{\TEff{\T_0}{\eff}}{\TEff{\T}{(\EComp\eZero{\eff}})}$ we get the result.
\item [\refToRule{try-do}]  In this case $\e$ is $\TryShort{(\Do\y{\e_1}{\e_2})}{\handler}$ and 
$\e'$ is $\Try{\e_1}{\seq\cc}{\y}{\TryShort{\e_2}{\handler}}$. 
 From \cref{lem:inversion:try} and \cref{lem:inversion:do} of \cref{lem:inversion} we have $\WTExp{\Do\y{\e_1}{\e_2}}{\UT_2}{\EComp{\eff_1}{\eff_2}}$ where
 \begin{enumerate}    
\item \label{i1} $\WTExp{\e_1}{\UT_1}{\eff_1}$ and
\item  \label{i2} $\IsWFExp{\emptyset}{\TVar{\UT_1}{\y}}{\e_2}{\UT_2}{\eff_2}$ and
\item  \label{i3}  $\IsWFHandler{\emptyset}{\emptyset}{\UT_2}{\handler}{\T}{\hfilter}$ and  $\eff{=}\FilterFun{\EComp{\eff_1}{\eff_2}}{\hfilter}$.
  \end{enumerate}
   Let $\handler=\Handler{\seq\cc}{\x}{\e_0}$.  From \cref{i3} and 
\refItem{lem:h-inversion}{two} we have that 
  \begin{enumerate}  \setcounter{enumi}{3}
  \item \label{i6}   $\hfilter=\HFilter{\cfilter_1\dots \cfilter_n}{\eff_0}$ and $\IsWFClause{\emptyset}{\emptyset}{\T}{\cc_i}{\cfilter_i}$ 
   \end{enumerate}
    From \cref{i3}, by weakening we get $\IsWFHandler{\emptyset}{\TVar{\UT_1}{\y}}{\UT_2}{\handler}{\T}{\hfilter}$ and from 
   \cref{i2} and rule \refToRule{t-try} we derive $\IsWFExp{\emptyset}{\TVar{\UT_1}{\y}}{\TryShort{\e_2}{\handler}}{\T}{\FilterFun{{\eff_2}}{\hfilter}}$.Therefore,
   from \Cref{i6} and rule \refToRule{T-Handler} we get
    \begin{enumerate}  \setcounter{enumi}{4}
 \item \label{i7} 
 $\IsWFHandler{\emptyset}{\emptyset}{\UT_1}{ \Handler{\seq\cc}{\y}{\TryShort{\e_2}{\handler}}}{\T}{\FilterFun{{\eff_2}}{\hfilter}}$
   \end{enumerate}
    From \Cref{i1,i7} and rule \refToRule{t-try} we get $\WTExp{\e'}{\T}{\FilterFun{\eff_1}{\FilterFun{{\eff_2}}{\hfilter}}}$. Finally, from \refItem{lem:filter-prop}{due}
   and \cref{i3} we have ${\FilterFun{\eff_1}{\FilterFun{{\eff_2}}{\hfilter}}}=\eff$.
  \item [\refToRule{fwd}]  In this case $\e$ is $\Try{\MCall{\ve}{\m}{\seq\T}{\seq\ve}}{\seq\cc}{\x}{\e_0}$ and 
$\e_1$ is $ \Do{\x}{\MCall{\ve}{\m}{\seq\T}{\seq\ve}}{\e_0}$.  From \refItem{lem:inversion}{try} we have 
\begin{enumerate}    
\item \label{ii1} $\WTExp{\MCall{\ve}{\m}{\seq\T}{\seq\ve}}{\T'}{\eff'}$  and
\item  \label{ii2}$\IsWFHandlerNarrow{\emptyset}{\emptyset}{\T'}{\Handler{\seq\cc}{\x}{\e_0}}{\T}{\hfilter}$ and $\eff{=}\FilterFun{\eff'}{\hfilter}$. 
  \end{enumerate}
From \cref{ii2} and 
\refItem{lem:h-inversion}{one}(2) we have that 
  \begin{enumerate}  \setcounter{enumi}{2}
 \item  \label{ii3}  $\IsWFExp{\emptyset}{\TVar{\T'}{\x}}{\e_0}{\T_0}{\eff_0}$ with  $\SubType{}{\T_0}{\T}$ and
 \item \label{ii4}   $\hfilter=\HFilter{\cfilter_1\dots \cfilter_n}{\eff_0}$ 
   \end{enumerate}
 Since we applied \refToRule{fwd} we have that $\mbody(\ve,\m)=\ple{\mgc,\_,\_,\_} $ and
  $\cmatch( \MCall\ve\m{\seq\T}{\seq\ve},\seq\cc)$ is undefined.
  From  \cref{invk:four} of \cref{lem:inversion:invk} we derive that 
  $\eff'=\eCall{\Gen\tname{\seq{\T_\X}}}\m{\seq\T}$ for some $\Gen\tname{\seq{\T_\X}}$.
  From definition  of $\FilterF{\hfilter}$ in \cref{fig:filters} and \Cref{ii2,ii4} we get $\eff{=}\EComp{\eff_0}{\eff'}$.
   Finally, from \Cref{ii1,ii3}  applying rule \refToRule{T-do} we get $\WTExp{\e'}{\T}{\eff}$. 
\item [\refToRule{try-ctx}]  In this case $\e$ is $ \TryShort{\e_1}{\handler}$ and 
$\e'$ is $ \TryShort{\e_2}{\handler}$ and $\e_1\purered\e_2$. 
From \refItem{lem:inversion}{try} we have 
\begin{enumerate}    
\item \label{iii1} $\WTExp{\e_1}{\T_1}{\eff_1}$  and
\item  \label{iii2}$\IsWFHandlerNarrow{\emptyset}{\emptyset}{\T_1}{\handler}{\T}{\hfilter}$ and $\eff{=}\FilterFun{\eff_1}{\hfilter}$. 
  \end{enumerate}
  By \cref{iii1} and inductive hypothesis we derive $\WTExp{\e_2}{\T_2}{\eff_2}$ such that  $\SubType{}{\TEff{\T_2}{\eff_2}}{\TEff{\T_1}{\eff_1}}$.
  Let $\hfilter=\HFilter{\seq\cfilter}{\eff_0}$. From  \cref{iii2} and \cref{lem:h-subsumption} we get that
  $\IsWFHandlerNarrow{\emptyset}{\emptyset}{\T_2}{\handler}{\T}{\HFilter{\seq\cfilter}{\eff'_0}}$ where $\SubType{}{\eff'_0}{\eff_0}$.
  Applying rule \refToRule{T-Try} we get $\WTExp{\e'}{\T}{\FilterFun{\eff_2}{\eff'}}$ where $\eff'=\HFilter{\seq\cfilter}{\eff'_0}$. 
  From  \refItem{lem:filter-prop}{uno} we get that $\SubType{}{\eff'}{\eff}$ and so $\SubType{}{\TEff{\T}{\eff'}}{\TEff{\T}{\eff}}$
 
\end{description}
\end{proofOf}